\documentclass[11pt]{article}

\usepackage{authblk}
\usepackage{amsthm}
\usepackage[letterpaper, left=1in, right=1in, bottom=1in, top=1in]{geometry}

\usepackage{graphicx}
\usepackage{amsmath,amssymb}
\usepackage{amsfonts}
\usepackage{booktabs}
\usepackage{multirow}
\usepackage{multicol}
\usepackage{cases}

\usepackage{lineno,hyperref}
\modulolinenumbers[5]

\usepackage[T1]{fontenc}
\usepackage{diagbox}
\usepackage{url}
\usepackage{hyperref}
\usepackage{cleveref}

\pdfoutput=1

\usepackage{thmtools}

\newtheorem{theorem}{Theorem}
\newtheorem{lemma}{Lemma}
\newtheorem{corollary}{Corollary}

\theoremstyle{definition}
\newtheorem{definition}{Definition}

\newtheorem*{claim}{Claim}

\declaretheoremstyle[
  spaceabove=6pt, spacebelow=6pt,
  headfont=\bfseries,
  notefont=\bfseries\sffamily,
  notebraces = {(}{)},
  postheadspace=0.5em,
  numbered=no,
  qed=$\blacksquare$
]{claimproof}
\declaretheorem[title=Proof of claim, style=claimproof]{claimproof}

\newcommand{\ed}{\mathsf{ed}}

\newcommand{\MSensSub}{\mathsf{MS}_{\mathrm{sub}}}
\newcommand{\MSensIns}{\mathsf{MS}_{\mathrm{ins}}}
\newcommand{\MSensDel}{\mathsf{MS}_{\mathrm{del}}}

\newcommand{\ASensSub}{\mathsf{AS}_{\mathrm{sub}}}
\newcommand{\ASensIns}{\mathsf{AS}_{\mathrm{ins}}}
\newcommand{\ASensDel}{\mathsf{AS}_{\mathrm{del}}}

\newcommand{\CDAWG}{\mathsf{CDAWG}}
\newcommand{\BWT}{\mathsf{BWT}}

\newcommand{\zorig}{\mathit{z}_{\mathrm{77}}}
\newcommand{\zsrorig}{\mathit{z}_{\mathrm{77sr}}}
\newcommand{\zend}{\mathit{z}_{\mathrm{End}}}
\newcommand{\zss}{\mathit{z}_{\mathrm{SS}}}
\newcommand{\zsrss}{\mathit{z}_{\mathrm{SSsr}}}
\newcommand{\zseveneight}{\mathit{z}_{\mathrm{78}}}

\newcommand{\LZorig}{\mathsf{LZ77}}
\newcommand{\LZorigsr}{\mathsf{LZ77sr}}
\newcommand{\LZSS}{\mathsf{LZSS}}
\newcommand{\LZSSsr}{\mathsf{LZSSsr}}
\newcommand{\LZEnd}{\mathsf{LZEnd}}
\newcommand{\LZseveneight}{\mathsf{LZ78}}

\newcommand{\grepair}{\mathit{g}_{\mathrm{rpair}}}
\newcommand{\glongest}{\mathit{g}_{\mathrm{long}}}
\newcommand{\ggreedy}{\mathit{g}_{\mathrm{grdy}}}
\newcommand{\gseq}{\mathit{g}_{\mathrm{seq}}}

\newcommand{\gavl}{\mathit{g}_{\mathrm{avl}}}
\newcommand{\gsimple}{\mathit{g}_{\mathrm{simple}}}
\newcommand{\galpha}{\mathit{g}_{\alpha}}
\newcommand{\gbsc}{\mathit{g}_{\mathrm{bsc}}}
\newcommand{\gis}{\mathit{g}_{\mathrm{is}}}

\newcommand{\Substr}{\mathsf{Substr}}
\newcommand{\occ}{\mathsf{occ}}

\newcommand{\str}{\mathsf{str}}

\newcommand{\polylog}{\mathrm{polylog}}
\newcommand{\rank}{\mathrm{rank}}
\newcommand\subcaption[1]{\begin{center}#1\end{center}}

\usepackage[final,mode=multiuser]{fixme}
\fxuseenvlayout{colorsig}
\fxusetargetlayout{color}
\FXRegisterAuthor{ta}{ata}{TA}
\FXRegisterAuthor{mf}{amf}{MF}
\FXRegisterAuthor{si}{asi}{SI}
\fxsetface{env}{}

\begin{document}

\title{Sensitivity of string compressors and repetitiveness measures}

\author[1]{Tooru~Akagi}
\author[1,2]{Mitsuru~Funakoshi}
\author[1,3]{Shunsuke~Inenaga}

\affil[1]{Department of Informatics, Kyushu University, Japan.

\{toru.akagi, mitsuru.funakoshi, inenaga\}@inf.kyushu-u.ac.jp}
\affil[2]{Japan Society for the Promotion of Science}
\affil[3]{PRESTO, Japan Science and Technology Agency, Japan}

\date{}
\maketitle

\begin{abstract}
  The \emph{sensitivity} of a string compression algorithm $C$ asks
  how much the output size $C(T)$ for an input string $T$ can increase
  when a single character edit operation is performed on $T$.
  This notion enables one to measure the robustness of compression algorithms
  in terms of errors and/or dynamic changes occurring in the input string.
  In this paper, we analyze the worst-case multiplicative sensitivity of
  string compression algorithms,
  which is defined by $\max_{T \in \Sigma^n}\{C(T')/C(T) : \ed(T, T') = 1\}$,
  where $\ed(T, T')$ denotes the edit distance between $T$ and $T'$.
  In particular, for the most common versions of the Lempel-Ziv 77 compressors,
  we prove that the worst-case multiplicative sensitivity is only
  a small constant (2 or 3, depending on the version of the Lempel-Ziv 77 and the edit operation type), 
  i.e., the size of the Lempel-Ziv 77 factorizations can be larger by only a small constant factor. 
  We strengthen our upper bound results by presenting matching lower bounds
  on the worst-case sensitivity for all these major versions of the Lempel-Ziv 77 factorizations.
  We generalize these results to the smallest bidirectional scheme $b$.
  In addition, we show that the sensitivity of a grammar-based compressor called GCIS (Grammar Compression by Induced Sorting)
  is also a small constant.
  Further, we extend the notion of the worst-case sensitivity
  to string repetitiveness measures such as
  the smallest string attractor size $\gamma$ and the substring complexity $\delta$,
  and show that the worst-case sensitivity of $\delta$ is also a small constant.
  These results contrast with the previously known related results
  such that the size $z_{\rm 78}$ of the Lempel-Ziv 78 factorization
  can increase by a factor of $\Omega(n^{1/4})$ [Lagarde and Perifel, 2018],
  and the number $r$ of runs in the Burrows-Wheeler transform
  can increase by a factor of $\Omega(\log n)$ [Giuliani et al., 2021] 
  when a character is prepended to an input string of length $n$.
  By applying our sensitivity bounds of $\delta$ or the smallest grammar
  to known results (c.f. [Navarro, 2021]),
  some non-trivial upper bounds for the sensitivities of important string compressors and repetitiveness measures
  including $\gamma$, $r$, LZ-End, RePair, LongestMatch, and AVL-grammar, are derived.
  We also exhibit the worst-case additive sensitivity
  $\max_{T \in \Sigma^n}\{C(T') - C(T) : \ed(T, T') = 1\}$,
  which allows one to observe more details in the changes of the output sizes. \\

  \noindent \textbf{keywords:} lossless data compression, Lempel-Ziv factorizations, run-length BWT, bidirectional scheme, string attractors, substring complexity, grammar compression, edit operations, sensitivity
\end{abstract}

\section{Introduction}

In this paper we introduce a new notion to quantify efficiency
of (lossless) compression algorithms, which we call the \emph{sensitivity} of compressors.
Let $C$ be a compression algorithm and let $C(T)$ denote
the size of the output of $C$ applied to an input text (string) $T$.
Roughly speaking, the sensitivity of $C$ measures how much
the compressed size $C(T)$ can change when a single-character-wise
edit operation is performed at an arbitrary position in $T$.
Namely, the worst-case \emph{multiplicative sensitivity} of $C$ 
is defined by $$\max_{T \in \Sigma^n}\{C(T')/C(T) : \ed(T, T') = 1\},$$
where $\ed(T, T')$ denotes the edit distance between $T$ and $T'$.
This new and natural notion enables one to measure the robustness of compression algorithms
in terms of errors and/or dynamic changes occurring in the input string.
Such errors and dynamic changes are commonly seen in real-world texts such as
DNA sequences and versioned documents.

The so-called highly repetitive sequences, which are strings containing a lot of repeated fragments, are abundant today: Semi-automatically generated strings via M2M communications, and collections of individual genomes of the same/close species are typical examples.
By intuition, such highly repetitive sequences should be highly compressible,
however, statistical compressors are known to fail to capture repetitiveness in a string~\cite{KreftN13}.
Therefore, other types of compressors, such as dictionary-based,
grammar-based, and/or lex-based compressors are often used to compress
highly repetitive sequences~\cite{LarssonM99,SirenVMN08,KuruppuPZ10,HoobinPZ11,Navarro21a}.

Let us recall two examples of well-known compressors:
The \emph{run-length Burrows-Wheeler Transform} (\emph{RLBWT}) is one kind of
compressor that is based on the lexicographically sorted rotations
of the input string.
The number $r$ of equal-character runs in the BWT of a string is known to be very small in practice: Indeed, BWT is used in the bzip2 compression format,
and several compressed data structures which support
efficient queries have been proposed~\cite{GagieNP20,BannaiGI20,NishimotoT21a,NishimotoT21b}.
The \emph{Lempel-Ziv 78 compression} (\emph{LZ78})~\cite{LZ78}
is one of the most fundamental dictionary based compressors
that is a core of in the gif and tiff compression formats.
While LZ78 only allows $\Omega(\sqrt{n})$ compression for any string of length $n$,
its simple structure allows for designing efficient compressed pattern matching algorithms and
compressed self-indices (c.f. \cite{KidaTSMA98,GasieniecR99,Gawrychowski12,Navarro04,FerradaN19} and references therein).

The recent work by Giuliani et al.~\cite{GiulianiILPST21}, however,
  shows that the number $r$ of runs in the BWT of a string of length $n$
  can grow by a multiplicative factor of $\Omega(\log n)$
  when a single character is prepended to the input string\footnote{It is well known that if the string ends with a unique end-marker $\$$, then the number $r$ of runs in the BWT increases additively by at most 2 after a character is prepended to the string. The work by Giuliani et al.~\cite{GiulianiILPST21}, however, shows that this is not the case without $\$$.}.
It is noteworthy that the family of strings discovered by Giuliani et al.~\cite{GiulianiILPST21} satisfies $r(T) = O(1)$ and $r(T') = \Omega(\log n)$,
where $r(T)$ and $r(T')$ respectively denote the number of runs
in the BWTs of $T$ and $T'$.
The other work by Lagarde and Perifel~\cite{LagardeP18} shows that
the size of the dictionary of LZ78,
which is equal to the number of factors in the respective LZ78 factorization,
can grow by a multiplicative factor of $\Omega(n^{1/4})$,
again when a single character is prepended to the input string.
Letting the LZ78 dictionary size be $\zseveneight$,
this multiplicative increase can also be described as $\Omega(\zseveneight^{3/2})$.
Lagarde and Perifel call the aforementioned phenomenon on LZ78 as ``one-bit catastrophe''.
Based on these known results, here we introduce the three following classes of string compressors
depending on their sensitivity.
\begin{enumerate}
  \item[(A)] Those whose sensitivity is $O(1)$;
  \item[(B)] Those whose sensitivity is $\polylog(n)$;
  \item[(C)] Those whose sensitivity is proportional to $n^c$ with some constant $0 < c \leq 1$.
\end{enumerate}
By generalizing the work of Lagarde and Perifel~\cite{LagardeP18},
we say that Class (C) is catastrophic in terms of the sensitivity.
Class (B) may not be catastrophic but the change in the compression size can still be quite large just for a mere single character edit operation to the input string.
Class (A) is the most robust against one-character edit operations among the three classes.
Recall that LZ78 $\zseveneight$ belongs to Class (C),
while it is not clear yet whether RLBWT $r$ belongs to Class (B) or (C)
(note that the work of Giuliani et al.~\cite{GiulianiILPST21} showed only a lower bound $\Omega(\log n)$).
In this paper, we show that the other major dictionary compressors,
the \emph{Lempel-Ziv 77 compression family}, belong to Class (A),
and thus such a catastrophe never happens with this family.
The \emph{LZ77} compression~\cite{LZ77}, which is the greedy parsing of
the input string $T$ where each factor of length more than one
refers to a previous occurrence to its left,
is the most important dictionary-based compressor both in theory and in practice.
    The LZ77 compression without self-references (resp. with self-references) can achieve $O(\log n)$ compression (resp. $O(1)$ compression) in the best case as opposed to the $\Omega(\sqrt{n})$ compression by the LZ78 counterpart,
    and the LZ77 compression is a core of common lossless compression formats including gzip, zip, and png.
In addition, its famous version called \emph{LZSS} (Lempel-Ziv-Storer-Szymanski)~\cite{StorerS82},
has numerous applications in string processing,
including finding repetitions~\cite{Crochemore84,KolpakovK99,GusfieldS04,BannaiIK17},
approximation of the smallest grammar-based compression~\cite{Rytter03,CharikarLLPPSS05},
and compressed self-indexing~\cite{BelazzouguiGGKO15,BilleEGV18,Navarro19,BelazzouguiCGGK21}, just to mention a few.

We show that the multiplicative
sensitivity of LZ77 with/without self-references
is at most $2$, namely, the number of factors in the respective LZ77 factorization can increase by at most a factor of $2$
for all types of edit operations (substitution, insertion, deletion of a character).
Then, we prove that the multiplicative sensitivity of 
LZSS with/without self-references is at most 3 for substitutions and deletions, and that it is at most 2 for insertions.
We also present matching lower bounds for the multiplicative sensitivity of
LZ77/LZSS with/without self-references
for all types of edit operations as well.
In addition, the multiplicative sensitivity of RLBWT $r$ turns out to be $O(\log r \log n)$,
which implies that $r$ belongs to Class (B)~\footnote{This $O(\log r \log n)$ upper bound for the sensitivity of $r$ follows from our result on the sensitivity of $\delta$ and our Lemma~\ref{lem:squeeze}, and from the known results between $r$ and $\delta$~\cite{kempa2020resolution,KociumakaNP20}.}
These results suggest that,
LZ77 and LZSS of Class (A) may better capture the repetitiveness of strings
  than RLBWT of Class (B) and LZ78 of Class (C),
since a mere single character edit operation should not much influence
the repetitiveness of a sufficiently long string.
We also consider the \emph{smallest bidirectional scheme}~\cite{StorerS82}
that is a generalization of the LZ family where each factor
can refer to its other occurrence to its left or right.
It is shown that for all types of edit operations,
the multiplicative sensitivity of the size $b$
of the smallest bidirectional scheme is at most 2,
and that there exist strings for which
the multiplicative sensitivity of $b$ is 2 with insertions and substitutions,
and it is 1.5 with deletions.
The \emph{smallest grammar} problem~\cite{CharikarLLPPSS05} is a famous NP-hard problem that asks to compute a grammar of the smallest size $g^*$ that derives only the input string.
We show that the multiplicative sensitivity of the smallest grammar size $g^*$ is at most 2.
Further, we extend the notion of the worst-case multiplicative sensitivity
to string repetitiveness measures such as
the size $\gamma$ of the \emph{smallest string attractor}~\cite{KempaP18}
and the \emph{substring complexity} $\delta$~\cite{KociumakaNP20},
both receiving recent attention~\cite{KempaPPR18,Prezza19,KutsukakeMNIBT20,MantaciRRRS21,ChristiansenEKN21}.
We prove
that the value of $\delta$ can increase by at most a factor of $2$
for substitutions and insertions,
and by at most a factor of 1.5 for deletions.
We show these upper bounds are also tight by presenting
matching lower bounds for the sensitivity of $\delta$.
We also present non-trivial upper and lower bounds for the sensitivity of $\gamma$.

As is mentioned above,
the work by Lagarde and Perifel~\cite{LagardeP18} considered
only the case of prepending a character to the string for the
multiplicative sensitivity of LZ78.
We show that the same lower bounds hold
for the multiplicative sensitivity of LZ78 in the case of substitutions and deletions,
and insertions inside the string,
by using a completely different instance from the one used in~\cite{LagardeP18}.

Studying the relations between different string repetitiveness measures/string compressor output sizes
has attracted much attention in the last two decades (for details see the survey~\cite{Navarro21a}).
Combining these known relations and our new sensitivity upper bounds mentioned above
gives us a kind of ``sandwich'' argument, which is formalized in Lemma~\ref{lem:squeeze}.
Using this lemma, some non-trivial upper bounds for the sensitivity of other measures can be driven,
including the LZ-End compressor~\cite{KreftN13}
and grammar-based compressors RePair~\cite{LarssonM99}, Longest-Match~\cite{KiefferY00},
Greedy~\cite{ApostolicoL00}, Sequential~\cite{YangK00}, LZ78~\cite{LZ78},
$\alpha$-balanced grammars~\cite{CharikarLLPPSS05},
AVL-grammars~\cite{Rytter03}, and Simple~\cite{Jez16}.
Theses upper bound results are reported as corollaries in the following sections.

Moreover, we consider the sensitivity of other compressors
and repetitiveness measures including Bisection~\cite{NelsonKC95},
GCIS~\cite{NunesLGAN18,NunesLGAN20}, and CDAWGs~\cite{BlumerBHME87}.

Table~\ref{tbl:multiplicative_sensitivity} summarizes our results on the multiplicative sensitivity of the string compressors and repetitiveness measures.

\begin{table}[bthp]
  \caption{Multiplicative sensitivity of the string compressors and string repetitiveness measures studied in this paper and in the literature, where $n$ is the input string length and $\Sigma$ is the alphabet. In the table ``sr'' stands for ``with self-references''.
  The upper bounds marked with ``$\dagger$'' are obtained by applying known results~\cite{KempaP18, KociumakaNP20, kempa2020resolution, KreftN13, KempaS22, CharikarLLPPSS05, Rytter03, Jez16} and our results on the sensitivity of the substring complexity $\delta$ or the smallest grammar $g^*$ to Lemma~\ref{lem:squeeze}.
  }
  \label{tbl:multiplicative_sensitivity}
  \centering
  \vspace{4mm}
  \begin{tabular}{l|c||c|c}
    \hline
    compressor/repetitiveness measure & edit type & upper bound & lower bound \\
    \hline \hline
    \multirow{2}{*}{Substring Complexity $\delta$} & ins./subst. & 2  & 2  \\ \cline{2-4}
    & deletion & 1.5 & 1.5  \\ \hline
    Smallest String Attractor $\gamma$ & all & $O(\log n)^\dagger$ 
    & 2 \\ \hline
    \multirow{2}{*}{RLBWT $r$} & insertion & \multirow{2}{*}{$O(\log n \log r)^\dagger$} & $\Omega(\log n)$~\cite{GiulianiILPST21} \\ \cline{2-2} \cline{4-4}
    & del./subst. & & - \\ \hline
    \multirow{2}{*}{Bidirectional Scheme $b$} & ins./subst. & 2  & 2  \\ \cline{2-4}
    & deletion & 2 & 1.5  \\ \hline
    LZ77 $\zorig$ & \multirow{2}{*}{all} & \multirow{2}{*}{2} & \multirow{2}{*}{2} \\ 
    LZ77sr $\zsrorig$ & & & \\ \hline 
    LZSS $\zss$ & del./subst. & 3  & 3 \\ \cline{2-4}
    LZSSsr $\zsrss$ & insertion & 2 & 2 \\ \hline
    \multirow{2}{*}{LZ78 $\zseveneight$} & insertion & \multirow{2}{*}{$O((n/ \log n)^{\frac{2}{3}})^\dagger$} & $\Omega(n^{\frac{1}{4}})$~\cite{LagardeP18} \\ \cline{2-2} \cline{4-4}
    & del./subst. &  & $\Omega(n^{\frac{1}{4}})$ \\ \hline
    LZ-End $\zend$ & all & $O(\log^2 (n/\delta))^\dagger$ & 2 \\ \hline
    Smallest grammar $g^*$ & all & 2 & - \\ \hline
    Repair $\grepair$ & \multirow{3}{*}{all} & \multirow{3}{*}{$O((n/ \log n)^{\frac{2}{3}})^\dagger$} & \multirow{3}{*}{-} \\ \cline{1-1}
    Longest match $\glongest$ &  & & \\ \cline{1-1}
    Greedy $\ggreedy$ & & & \\ \hline
    Sequential $\gseq$ & all & $O((n/ \log n)^{\frac{3}{4}})^\dagger$ & - \\ \hline
$\alpha$-balanced grammar $\galpha$ & \multirow{3}{*}{all} & \multirow{3}{*}{$O(\log (n / g^*))^\dagger$} & \multirow{3}{*}{-} \\ \cline{1-1}
    AVL grammar $\gavl$ &  & & \\ \cline{1-1}
    Simple $\gsimple$ & & & \\ \hline
    \multirow{2}{*}{Bisection $\gbsc$} & substitution & 2 & 2 \\ \cline{2-4}
    & ins./del. & $|\Sigma|+1$ & $|\Sigma|$ \\ \hline
    GCIS $\gis$ & all & 4 & 4 \\ \hline
    CDAWG $e$ & all & - & 2 \\ \hline
  \end{tabular}
\end{table}

In addition to the afore-mentioned multiplicative sensitivity,
we also introduce the worst-case \emph{additive sensitivity},
which is defined by $$\max_{T \in \Sigma^n}\{C(T') - C(T) : \ed(T, T') = 1\},$$
for all the string compressors/repetitiveness measures $C$ dealt in this paper.
We remark that the additive sensitivity
allows one to observe and evaluate more details in the changes of the output sizes,
as summarized in Table~\ref{tbl:additive_sensitivity}.
For instance, we obtain \emph{strictly tight upper and lower bounds}
for the additive sensitivity of LZ77 with and without self-references
in the case of substitutions and insertions.
Studying the additive sensitivities of string compressors is motivated by approximation of the Kolmogorov complexity.
Let $K(T)$ denote the Kolmogorov complexity of string $T$,
that is the length of a shortest program that produces $T$.
While $K(T)$ is known to be uncomputable,
the additive sensitivity $K(T')-K(T)$ for deletions is at most $O(\log n)$ \emph{bits}, since it suffices to add ``Delete the $i$th character $T[i]$ from $T$.'' at the end of the program.
Similarly, the additive sensitivity of $K$ for insertions and substitutions 
is at most $O(\log n + \log \sigma)$ \emph{bits}, where $\sigma$ is the alphabet size.
Therefore, a ``good approximation'' of the Kolmogorov complexity $K$ should have small additive sensitivity.

\begin{table}[thp]
  \caption{Additive sensitivity of the string compressors and string repetitiveness measures studied in this paper, where $n$ is the input string length and $\Sigma$ is the alphabet. Some upper/lower bounds are described in terms of both the measure and $n$. In the table ``sr'' stands for ``with self-references''.
  The upper bounds marked with ``$\dagger$'' are obtained by applying known results~\cite{KempaP18, KociumakaNP20, kempa2020resolution, KreftN13, KempaS22, CharikarLLPPSS05, Rytter03, Jez16} and our results on the sensitivity of the substring complexity $\delta$ or the smallest grammar $g^*$ to Lemma~\ref{lem:squeeze}.}
  \label{tbl:additive_sensitivity}
  \centering
  \vspace{4mm}
  \scalebox{0.88}{
    \begin{tabular}{l|c||c|c|c|c}
    \hline
    compressor/ & \multirow{2}{*}{edit type} & \multicolumn{2}{c|}{\multirow{2}{*}{upper bound}} & \multicolumn{2}{c}{\multirow{2}{*}{lower bound}} \\
    repetitiveness measure & & \multicolumn{2}{c|}{} & \multicolumn{2}{c}{} \\ \hline \hline
    Substring Complexity $\delta$ & all & \multicolumn{2}{c|}{1}  & \multicolumn{2}{c}{1} \\ \hline
    Smallest String Attractor $\gamma$ & all & \multicolumn{2}{c|}{$O(\delta \log n)^\dagger$
    } & $\gamma-3$ & $\Omega(\sqrt{n})$ \\ \hline
    \multirow{2}{*}{RLBWT $r$} & insertion & \multicolumn{2}{c|}{\multirow{2}{*}{$O(r \log n \log r)^\dagger$}} & \multirow{2}{*}{-} & $\Omega(\log n)$~\cite{GiulianiILPST21} \\ \cline{2-2} \cline{6-6}
    & del./subst. & \multicolumn{2}{c|}{} & & - \\ \hline
    Bidirectional Scheme $b$ & all & \multicolumn{2}{c|}{$b+2$} & $b/2-3$ & $\Omega(\sqrt{n})$ \\ \hline
    \multirow{2}{*}{LZ77 $\zorig$} & subst./ins. & \multicolumn{2}{c|}{$\zorig - 1$} & $\zorig - 1$ & \multirow{2}{*}{$\Omega(\sqrt{n})$} \\ \cline{2-5} & deletion & \multicolumn{2}{c|}{$\zorig - 2$} & $\zorig - 2$ & \\ \hline
\multirow{2}{*}{LZ77sr $\zorig$} & subst./ins. & \multicolumn{2}{c|}{\multirow{2}{*}{$\zsrorig$}} & $\zsrorig$ & \multirow{2}{*}{$\Omega(\sqrt{n})$} \\ \cline{2-2} \cline{5-5}
    & deletion & \multicolumn{2}{c|}{} & $\zsrorig - 2$ & \\ \hline
    \multirow{2}{*}{LZSS $\zss$} & del./subst. & \multicolumn{2}{c|}{$2\zss-2$}  & $2\zss - \Theta(\sqrt{\zss})$ & \multirow{2}{*}{$\Omega(\sqrt{n})$} \\ \cline{2-5}
    & insertion & \multicolumn{2}{c|}{$\zss$} & $\zss - \Theta(\sqrt{\zss})$ &  \\ \hline
    \multirow{2}{*}{LZSSsr $\zsrss$} & del./subst. & \multicolumn{2}{c|}{$2\zsrss$} & $2\zsrss - \Theta(\sqrt{\zsrss})$ &  \multirow{2}{*}{$\Omega(\sqrt{n})$} \\ \cline{2-5}
    & insertion & \multicolumn{2}{c|}{$\zsrss+1$} & $\zsrss - \Theta(\sqrt{\zsrss})$ & \\ \hline
    \multirow{2}{*}{LZ78 $\zseveneight$} & insertion & \multicolumn{2}{c|}{\multirow{2}{*}{$O(g^* \cdot (n/ \log n)^{\frac{2}{3}})^\dagger$}} & $\Omega((\zseveneight)^{\frac{3}{2}})$~\cite{LagardeP18} & $\Omega(n/\log n)$~\cite{LagardeP18} \\ \cline{2-2} \cline{5-6}
    & del./subst. & \multicolumn{2}{c|}{} & $\Omega((\zseveneight)^{\frac{3}{2}})$ & $\Omega(n^{\frac{3}{4}})$ \\ \hline
    LZ-End $\zend$ & all & \multicolumn{2}{c|}{$O(\zend \log^2 (n/\delta))^\dagger$} & $\zend - \Theta(\sqrt{\zend})$ & $\Omega(\sqrt{n})$ \\ \hline
    Smallest grammar $g^*$ & all & \multicolumn{2}{c|}{$g^*$} & \multicolumn{2}{c}{-} \\ \hline
    Repair $\grepair$ & \multirow{3}{*}{all} & \multicolumn{2}{c|}{\multirow{3}{*}{$O(g^* \cdot (n/ \log n)^{\frac{2}{3}})^\dagger$}} & \multicolumn{2}{c}{\multirow{3}{*}{-}} \\ \cline{1-1}
    Longest match $\glongest$ & & \multicolumn{2}{c|}{} & \multicolumn{2}{c}{} \\ \cline{1-1}
    Greedy $\ggreedy$ & & \multicolumn{2}{c|}{} & \multicolumn{2}{c}{} \\ \hline
    Sequential $\gseq$ & all & \multicolumn{2}{c|}{$O(g^* \cdot (n/ \log n)^{\frac{3}{4}})^\dagger$} & \multicolumn{2}{c}{-} \\ \hline    
$\alpha$-balanced grammar $\galpha$ & \multirow{3}{*}{all} & \multicolumn{2}{c|}{\multirow{3}{*}{$O(g^* \log (n / g^*))^\dagger$}} & \multicolumn{2}{c}{\multirow{3}{*}{-}} \\ \cline{1-1}
    AVL grammar $\gavl$ & & \multicolumn{2}{c|}{} & \multicolumn{2}{c}{} \\ \cline{1-1}
    Simple $\gsimple$ & & \multicolumn{2}{c|}{} & \multicolumn{2}{c}{} \\ \hline    
    \multirow{2}{*}{Bisection $\gbsc$} & substitution & $\gbsc$ & $\lceil \log_2 n \rceil$ & $\gbsc-4$ & $2 \log_2 n - 4$ \\ \cline{2-6}
    & ins./del. & \multicolumn{2}{c|}{$|\Sigma| \gbsc$} & $\Omega(|\Sigma|\gbsc)$ & $\Omega(|\Sigma|^2 \log \frac{n}{|\Sigma|})$ \\ \hline 
    GCIS $\gis$ & all & \multicolumn{2}{c|}{$3\gis$} & $3\gis -29$ & $(3/4)n+1$ \\ \hline
    CDAWG $e$ & all & \multicolumn{2}{c|}{-} & $e$ & $n$ \\ \hline
  \end{tabular}
  }
\end{table}

\subsection{Related work}

\subsubsection{String monotonicity}
A string repetitiveness measure $C$ is called
\emph{monotone} if, for any string $T$ of length $n$,
$C(T') \leq C(T)$ holds
with any of its prefixes $T' = T[1..i]$ and suffixes $T' = T[j..n]$~\cite{KociumakaNP20}.
Kociumaka et al~\cite{KociumakaNP20} pointed out that
$\delta$ is monotone, and posed a question whether $\gamma$ or
the size $b$ of the smallest bidirectional macro scheme~\cite{StorerS82}
are monotone.
This monotonicity for $C$ can be seen as a special and extended
case of our sensitivity for deletions,
namely, if we restrict $T'$ to be the string
obtained by deleting either the first or the last character from $T$,
then it is equivalent to asking whether
$\max_{T \in \Sigma}\{C(T')/C(T) : T' \in\{ T[1..n-1], T' = T'[2..n]\}\} \leq 1$.
Mantaci et al.~\cite{MantaciRRRS21} proved that
$\gamma$ is not monotone, by showing a family of strings $T$
such that $\gamma(T) = 2$ and $\gamma(T') = 3$ with $T' = T[1..n-1]$,
which immediately leads to a lower bound $3/2 = 1.5$ for 
the multiplicative sensitivity of $\gamma$.
In this paper, we present a new lower bound for the multiplicative sensitivity of $\gamma$, which is $2$.
Mitsuya et al.~\cite{MitsuyaNIBT2021} considered the monotonicity of LZ77 without self-references $\zorig$
presented a family of strings $T$ for which
$\zorig(T') / \zorig(T) \approx 4/3$ with $T' = [2..n]$.
Again, our matching upper and lower bounds for the
multiplicative sensitivity of $\zorig$, which are both $2$,
improve this $4/3$ bound.

\subsubsection{Comparison to sensitivity of other algorithms}

The notion of the sensitivity of (general) algorithms was
first introduced by Varma and Yoshida~\cite{VarmaY21}.
They studied the \emph{average} sensitivity
of well-known graph algorithms,
and presented interesting lower and upper bounds on
the expected number of changes in the output of an algorithm $A$,
when a randomly chosen edge is deleted from the input graph $G$.
The worst-case sensitivity of a graph algorithm
for edge-deletions and vertex-deletions
was considered by Yoshida and Zhou~\cite{YoshidaZ21}.

As opposed to these existing work on the sensitivity of graph algorithms,
our notion of the sensitivity of string compressors focuses
on the \emph{size} of their compressed outputs and does not formulate
the perturbation of their structural changes.
This is because the primary task of data compression is
to represent the input data with as little memory as possible,
and the structural changes of the compressed outputs can be
of secondary importance.

We remark that most instances of $\Sigma^n$ are not compressible,
or in other words, a randomly chosen string $T$ from $\Sigma^n$ is not compressible.
Such a string $T$ does not become highly compressible just after a one-character edit operation, and hence $C(T)$ and $C(T')$ are expected to be almost the same.
Therefore, considering the average sensitivity of 
string compressors and repetitiveness measures does not seem worth discussing,
and this is the reason why we focus on the worst-case sensitivity of string compressors and repetitiveness measures.

Still, our notion permits one to evaluate the worst-case size changes of
several known \emph{compressed string data structures}
in the dynamic setting, as will be discussed in the following subsection.

\subsubsection{Compressed string data structures}
A compressed string data structure is built on a compressed representation of the string and supports efficient queries such as pattern matching and substring extraction within compressed space. Since the string compressors and string repetitiveness measures that we deal with in this paper are models for highly repetitive strings, we mention some compressed string indexing structures for highly repetitive sequences below.

The Block tree of a string of length $n$ uses $O(\zss \log (n / \zss))$ words of space and supports random access queries in $O(\log (n / \zss))$ time.
Navarro~\cite{Navarro19} proposed an LZ-based indexing structure that uses $O(\zss \log (n/\zss))$ words of space and counts the number of occurrences of a query pattern in the text string in $O(m \log^{2+\epsilon} n)$ time, where $m$ is the length of the pattern and $\epsilon > 0$ is any constant.
An $O(\log n)$-time longest common extension (LCE) data structure that takes $O(\zss \log(n / \zss))$ space and is based on Recompression~\cite{Jez16} was proposed by I~\cite{I17}.
Nishimoto et al.~\cite{NishimotoIIBT20} presented a dynamic $O(\min\{\zss \log n \log^* n, n\})$-space compressed data structure that supports pattern matching and substring insertions/deletions in $O(m \cdot \polylog(n))$ time, where $m$ is the length of the pattern/substring.
Kociumaka et al.~\cite{KociumakaNP20} proposed a compressed indexing structure that uses $O(\delta \log(n / \delta))$ words of space, performs random access in $O(\log(n / \delta))$ time, and finds all the $\occ$ occurrences of a given pattern of length $m$ in $O(m \log n + \occ \log^{\epsilon} n)$ time.
Very recently, Kociumaka et al.~\cite{abs-2206-00781} proposed an improved data structure of $O(\delta \log(n / \delta))$-space that supports pattern matching queries in $O(m + (\occ + 1) \log^{\epsilon} n)$ time.
Two independent compressed indexing structures, which are based on grammar compression called GCIS (Grammar Compression by Induced Sorting)~\cite{NunesLGAN18} have been proposed~\cite{AkagiKNIBT21,Diaz-DominguezN21}.
Our constant upper bounds on the multiplicative sensitivity for $\zss$, $\delta$, and $\gis$ imply that the afore-mentioned compressed data structures retain their asymptotic space complexity even after one-character edit operation at an arbitrary position, though they may incur a certain amount of structural changes.

The r-index~\cite{GagieNP20}, the refined r-index~\cite{BannaiGI20},
and the OptBWTR~\cite{NishimotoT21a} are efficient indexing structures which are built on the RLBWT and use $O(r)$ words of space. The result by Giuliani et al.~\cite{GiulianiILPST21}, which uses a family of strings of length $n$ with $r = O(1)$, shows that the space complexity of these indexing structures can grow from $O(1)$ words of space to $O(\log n)$ words of space, after appending a character to the string. In turn, our upper bound for the sensitivity of $r$ implies that after a one-character edit operation, the space usage of these indexing structures is bounded by $O(r \log r \log n)$ for any string of length $n$.

There also exist compressed data structures based on other string compressors and/or repetitiveness measures:
Kempa and Prezza~\cite{KempaP18} presented an $O(\gamma \tau \log_{\tau} (n / \gamma))$-space data structure that allows for extracting substrings of length-$\ell$ in $O(\log_\tau (n / \gamma) + \ell \log(\sigma)/\omega)$ time, where $\tau \geq 2$ is an integer parameter, $\sigma$ is the alphabet size, and $\omega$ is the machine-word size in the RAM model.
Navarro and Prezza~\cite{NavarroP19} gave a data structure of size $O(\gamma \log (n / \gamma))$ that supports pattern matching queries in $O(m \log n + \occ \log^\epsilon n)$ time.
Christiansen et al.~\cite{ChristiansenEKN21} introduced a compressed indexing structure that occupies $O(\gamma \log (n / \gamma) \log^\epsilon n)$ space and finds all the $\occ$ pattern occurrences in optimal $O(m+ \occ)$ time (for other trade-offs between the space and the query time are also reported, see~\cite{ChristiansenEKN21}).
Gawrychowski et al.~\cite{GawrychowskiKKL18} presented a data structure for maintaining a dynamic set of strings, which is based on Recompression by Je\.{z}~\cite{Jez16}.
Kempa and Saha~\cite{KempaS22} developed a compressed data structure
that occupies $O(\zend)$ space and supports random access and LCE queries in $O(\polylog(n))$ time.
A compressed indexing structure that can be built directly from
the LZ77-compressed text is also known~\cite{kempa2020resolution,abs-1910-10631}.
For other compressed string indexing structures, see this survey~\cite{Navarro21b}.

\subsection{Paper organization}
Section~\ref{sec:preliminaries} introduces necessary notations.
We then present the worst-case sensitivity of
string compressors and repetitiveness measures in the increasing order
of their respective sizes:
from $\delta$ to $\gamma$, LZ77 family, LZ-End, and grammars: 
Section~\ref{sec:delta} deals with the substring complexity $\delta$;
Section~\ref{sec:gamma} deals with the smallest string attractor $\gamma$,
Section~\ref{sec:rlbwt} deals with the RLBWT $r$,
Section~\ref{sec:bs} deals with the smallest bidirectional scheme $b$,
Section~\ref{sec:lz77} deals with the LZ77 with/without self-references $\zorig$ and $\zsrorig$;
Section~\ref{sec:lzss} deals with the LZSS with/without self-references $\zss$ and $\zsrss$.
Section~\ref{sec:lzend} deals with the LZ-End $\zend$;
Section~\ref{sec:lz78} deals with the LZ78 $\zseveneight$;
Section~\ref{sec:grammar} deals with the smallest grammar $g^*$, and its applications to practical and/or approximation grammars
RePair $\grepair$, LongestMatch $\glongest$, Greedy $\ggreedy$, Sequential $\gseq$, LZ78 $\zseveneight$,
$\alpha$-balanced grammar $\galpha$, AVL-grammar $\gavl$, and Simple grammar $\gsimple$.
Section~\ref{sec:gcis} deals with the GCIS grammar $\gis$;
Section~\ref{sec:bisection} deals with the Bisection grammar $\gbsc$;
Section~\ref{sec:cdawg} deals with the CDAWG size $e$.
In Section~\ref{sec:conclusions} we conclude the paper and list
several open questions of interest.
 \section{Preliminaries}
\label{sec:preliminaries}

\subsection{Strings, factorizations, and grammars}

Let $\Sigma$ be an {\em alphabet} of size $\sigma$.
An element of $\Sigma^*$ is called a {\em string}.
For any non-negative integer $n$,
let $\Sigma^n$ denote the set of strings of length $n$ over $\Sigma$.
The length of a string $T$ is denoted by $|T|$.
The empty string $\varepsilon$ is the string of length 0,
namely, $|\varepsilon| = 0$.
The $i$-th character of a string $T$ is denoted by
$T[i]$ for $1 \leq i \leq |T|$,
and the \emph{substring} of a string $T$ that begins at position $i$ and
ends at position $j$ is denoted by $T[i..j]$ for $1 \leq i \leq j \leq |T|$.
For convenience, let $T[i..j] = \varepsilon$ if $j < i$.
Substrings $T[1..j]$ and $T[i..|T|]$
are respectively called a \emph{prefix} and a \emph{suffix} of $T$.

A factorization of a non-empty string $T$
is a sequence $f_1, \ldots, f_x$ of non-empty substrings of $T$
such that $T = f_1 \cdots f_x$.
Each $f_i$ is called a \emph{factor}.
The \emph{size} of the factorization is the number $x$ of factors in the factorization.

A context-free grammar $\mathcal{G}$ which generates only a single string $T$
is called a \emph{grammar compression} for $T$.
The \emph{size} of $\mathcal{G}$ is the total length of the
right-hand sides of all the production rules in $\mathcal{G}$.
The \emph{height} of $\mathcal{G}$ is the height of the derivation tree of $\mathcal{G}$.

\subsection{Worst-case sensitivity of compressors and repetitiveness measures}

For a string compression algorithm $C$ and an input string $T$,
let $C(T)$ denote the size of the compressed representation of $T$
obtained by applying $C$ to $T$.
For convenience, we use the same notation
when $C$ is a string repetitiveness measure,
namely, $C(T)$ is the value of the measure $C$ for $T$.

Let us consider the following edit operations on strings:
character substitution ($\mathrm{sub}$), 
character insertion ($\mathrm{ins}$),
and character deletion ($\mathrm{del}$).
For two strings $T$ and $S$, let $\ed(T,S)$ denote
the \emph{edit distance} between $T$ and $S$, namely,
$\ed(T,S)$ is the minimum number of edit operations
that transform $T$ into $S$.

Our interest in this paper is: ``How much can the compression size
or the repetitiveness measure size change
when a single-character-wise edit operation is performed on a string?''
To answer this question,
for a given string length $n$,
we consider an arbitrarily fixed string $T$ of length $n$
and all strings $T'$ that can be obtained by applying a single edit operation to $T$,
that is, $\ed(T,T') = 1$.
We define the worst-case \emph{multiplicative sensitivity} of $C$
w.r.t. a substitution, insertion, and deletion as follows:
\begin{eqnarray*}
  \MSensSub(C,n) & = & \max_{T \in \Sigma^n}\{C(T')/C(T) : T' \in \Sigma^n, \ed(T,T') = 1\},\\
  \MSensIns(C,n) & = & \max_{T \in \Sigma^{n}}\{C(T')/C(T) : T' \in \Sigma^{n+1}, \ed(T,T') = 1\},\\
  \MSensDel(C,n) & = & \max_{T \in \Sigma^n}\{C(T')/C(T) : T' \in \Sigma^{n-1}, \ed(T,T') = 1\}.
\end{eqnarray*}

We also consider the worst-case \emph{additive sensitivity} of $C$
w.r.t. a substitution, insertion, and deletion, as follows:
\begin{eqnarray*}
  \ASensSub(C,n) & = & \max_{T \in \Sigma^n}\{C(T')-C(T) : T' \in \Sigma^n, \ed(T,T') = 1\},\\
  \ASensIns(C,n) & = & \max_{T \in \Sigma^n}\{C(T')-C(T) : T' \in \Sigma^{n+1}, \ed(T,T') = 1\},\\
  \ASensDel(C,n) & = & \max_{T \in \Sigma^n}\{C(T')-C(T) : T' \in \Sigma^{n-1}, \ed(T,T') = 1\}.
\end{eqnarray*}

We remark that, in general, $C(T')$ can be larger than $C(T)$
even when $T'$ is obtained by a character deletion from $T$ (i.e. $|T'| = n-1$).
Such strings $T$ are already known for
the Lempel-Ziv 77 factorization size $z$ when $T' = T[2..n]$~\cite{MitsuyaNIBT2021},
or for the smallest string attractor size $\gamma$ when $T' = T[1..n-1]$~\cite{MantaciRRRS21}.

The above remark implies that in general the multiplicative/additive sensitivity for insertions and deletions
may not be symmetric and therefore they need to be discussed separately for some $C$.
Note, on the other hand, that the maximum difference between $C(T')$ and $C(T)$
when $|T'| = n-1$ (deletion) and $C(T')-C(T) < 0$ is equivalent to $\ASensIns(C,n-1)$,
and symmetrically
the maximum difference of $C(T')$ and $C(T)$ when $|T'| = n+1$ (insertion) and $C(T')-C(T) < 0$ is equivalent to $\ASensDel(C,n+1)$,
with the roles of $T$ and $T'$ exchanged.
Similar arguments hold for the multiplicative sensitivity with insertions/deletions.
Consequently, it suffices to consider $\MSensIns(C,n)$, $\MSensDel(C,n)$, $\ASensIns(C,n)$, $\ASensDel(C,n)$ for insertions/deletions.

Consider two measures $\alpha$ and $\beta$.
An upper bound for the multiplicative sensitivity of $\beta$ can readily be derived in the some cases, as follows:
\begin{lemma} \label{lem:squeeze}
  Let $T$ be any string of length $n$ and let $T'$ be any string with $\ed(T,T') = 1$.
  If the following conditions:
  \begin{itemize}
    \item $\alpha(T')/\alpha(T) = O(1)$;
    \item $\alpha(T) \leq \beta(T)$;
    \item $\beta(T) = O(\alpha(T) \cdot f \cdot (n, \alpha(T)))$, where $f$ is a function such that for any constant $c$ there exists a constant $c'$ satisfying $f(n, c \cdot \alpha(T)) \leq c' \cdot f(n, \alpha(T))$.
  \end{itemize}
  all hold, then we have the following upper bounds (1), (2), and (3) for the sensitivity of $\beta$:
  \begin{enumerate}
  \item[(1)] $\MSensSub(\beta,n) = O(f(n,\alpha))$ and $\ASensSub(\beta, n) = O(\alpha \cdot f(n, \alpha))$;
  \item[(2)] $\MSensIns(\beta,n) = O(f(n,\alpha))$ and $\ASensIns(\beta, n) = O(\alpha \cdot f(n, \alpha))$;
  \item[(3)] $\MSensDel(\beta,n) = O(f(n,\alpha))$ and $\ASensDel(\beta, n) = O(\alpha \cdot f(n, \alpha))$.
  \end{enumerate}
\end{lemma}

\begin{proof}
Let $c = \alpha(T')/\alpha(T)$, where $c$ is a constant.
  Then we have
  \begin{eqnarray*}
    \frac{\beta(T')}{\beta(T)} & = & O\left(\frac{\alpha(T') \cdot f(n,\alpha(T'))}{\alpha(T)}\right) \\
    & = & O\left(\frac{\alpha(T') \cdot f(n,c \cdot \alpha(T))}{\alpha(T)}\right) \\
    & = & O\left(\frac{\alpha(T') \cdot c' \cdot f(n,\alpha(T))}{\alpha(T)}\right) \\
    & = & O(f(n, \alpha(T))).
  \end{eqnarray*}
  Also,
  \begin{eqnarray*}
    \beta(T')-\beta(T) & = & O(\alpha(T') \cdot f(n, \alpha(T')) - \alpha(T) \cdot f(n, \alpha(T))) \\
    & = & O(\alpha(T') \cdot f(n, c \cdot \alpha(T)) - \alpha(T) \cdot f(n, \alpha(T))) \\
    & = & O(\alpha(T') \cdot c' \cdot f(n, \alpha(T)) - \alpha(T) \cdot f(n, \alpha(T))) \\
    & = & O((c' \cdot \alpha(T') - \alpha(T)) \cdot f(n, \alpha(T))) \\
    & = & O((c' \cdot c \cdot \alpha(T) - \alpha(T)) \cdot f(n, \alpha(T))) \\
    & = & O(\alpha(T) \cdot f(n, \alpha(T))).
  \end{eqnarray*}
\end{proof}

The functions satisfying $f(n, c \cdot \alpha(T)) \leq c' \cdot f(n, \alpha(T))$ include functions $f$ which are polynomial, poly-logarithmic, or constant in terms of $\alpha(T)$.
 \section{Substring Complexity}
\label{sec:delta}

In this section, we consider the worst-case sensitivity
of the string repetitiveness measure $\delta$,
which is the \emph{substring complexity} of strings~\cite{KociumakaNP20}.
For any string $T$ of length $n$,
the substring complexity $\delta(T)$ is defined as
$\delta(T) = \max_{1 \leq k \leq n}\left(\Substr(T,k)/k\right)$,
where $\Substr(T,k)$ is the number of distinct substrings of length $k$ in $T$.
It is known that $\delta(T) \leq \gamma(T)$ holds for any $T$~\cite{KociumakaNP20}.

In what follows, we present tight upper and lower bounds
for the multiplicative sensitivity of $\delta$ for all cases
of substitutions, insertions, and deletions.
We also present the additive sensitivity of $\delta$.

\subsection{Lower bounds for the sensitivity of $\delta$}

\begin{theorem}
  The following lower bounds on
  the sensitivity of $\delta$ hold: \\
  \textbf{substitutions:} $\MSensSub(\delta,n) \geq 2$.
    $\ASensSub(\delta,n) \geq 1$. \\
  \textbf{insertions:} $\MSensIns(\delta,n) \geq 2$.
    $\ASensIns(\delta,n) \geq 1$. \\
  \textbf{deletions:} $\lim \inf_{n\to\infty}\MSensDel(\delta,n) \geq 1.5$.
    $\lim \inf_{n\to\infty}\ASensDel(\delta,n) \geq 1$.
\end{theorem}

\begin{proof}
  \textbf{substitutions}:
    Consider strings $T = \mathtt{a^\mathit{n}}$ and $T' = \mathtt{a^{\mathit{n}-1} b}$.
    Then $\delta(T)=1$ and $\delta(T')=2$ hold.
    Thus we get $\MSensSub(\delta,n) \geq 2$ and
    $\ASensSub(\delta,n) \geq 1$. 
    
  \textbf{insertions}:
    Consider strings $T = \mathtt{a^\mathit{n}}$ and $T' = \mathtt{a^\mathit{n} b}$.
    Then $\delta(T)=1$ and $\delta(T')=2$ hold.
    Thus we get $\MSensIns(\delta,n) \geq 2$ and
    $\ASensIns(\delta,n) \geq 1$. 
  
  \textbf{deletions}:
  Consider string $$T = \mathtt{(abb)}^m \mathtt{a (bba)}^{m+1} \mathtt{a}^{3m} (\mathtt{bba})^m$$ with a positive integer $m$.
  Let $n = 12m+4 = |T|$. 
    For the sake of exposition, let $w_1=\mathtt{(abb)^\mathit{m}}$, $w_2=\mathtt{(bba)^{\mathit{m}+1}}$, $w_3=\mathtt{a}^{3m}$, and $w_4=\mathtt{(bba)^\mathit{m}}$ such that $T = w_1 \mathtt{a} w_2 w_3 w_4$.
    To analyze $\delta(T)$, we consider $\Substr(T, k)$
    for four different groups of $k$, as follows:
    \begin{itemize}
    \item For $1 \leq k \leq 2$:
    Since $T$ is a binary string,
    $\max_{1 \leq k \leq 2}\left(\Substr(T,k)/k\right) = 2$.
    
    \item For $3 \leq k \leq 3m$:
      The prefix $w_1 \mathtt{a} w_2 = \mathtt{(abb)^\mathit{m} a (bba)^{\mathit{m}+1}}$ and the suffix $w_4=\mathtt{(bba)^\mathit{m}}$ contain three distinct substrings $\mathtt{(abb)}^{k/3}$, $(\mathtt{bba})^{\mathit{k}/3}$, and $(\mathtt{bab})^{\mathit{k}/3}$ for each length $k$, and the substring $w_3 = \mathtt{a}^{3\mathit{m}}$ contains a unique substring $\mathtt{a^{\mathit{k}}}$ for each length $k$.
    The remaining distinct substrings must contain the range
    $[6m+4,6m+5]$ or $[9m+4,9m+5]$,
    which are the left and right boundaries of $w_3$, respectively.
    There are $k-1$ distinct substrings containing $[6m+4,6m+5]$ of form:
    \[
    \begin{array}{ll}
      (\mathtt{bba})^{l_1}\mathtt{a}^{k-3l_1} & \mbox{for } 1 \leq l_1 \leq \lfloor (k-1)/3 \rfloor; \\
      \mathtt{a(bba)}^{l_2-1}\mathtt{a}^{k-3l_2+2} & \mbox{for } 1\leq l_2 \leq \lfloor (k+1)/3 \rfloor;\\
      \mathtt{ba(bba)}^{l_3-1}\mathtt{a}^{k-3l_3+1} & \mbox{for } 1\leq l_3 \leq \lfloor k/3 \rfloor.
    \end{array}
    \]
    Also, there are $k-1$ distinct substrings containing $[9m+4,9m+5]$ of form
    \[
    \mathtt{a}^{k-l_4}(\mathtt{bba})^{l_4/3}~~\mbox{ for } 1\leq l_4 \leq k-1.
    \]
    Notice however that the two substrings $\mathtt{a(bba)}^{l_2-1} \mathtt{a}^{k-3l_2+2} = \mathtt{a^{\mathit{k}}}$ with $l_2=1$ and $\mathtt{a}^{k-l_4}(\mathtt{bba})^{l_4/3} \\= \mathtt{(abb)}^{\mathit{k}/3}$ with $l_4=k-1$
have already been counted in the other positions in $T$,
    and thus these duplicates should be removed.
    Summing up all these,
    we obtain $\Substr(T,k) = 3+1+2(k-1)-2 = 2k$ for every $3 \leq k \leq 3m$,
    implying $\max_{3 \leq k \leq 3m}\left(\Substr(T,k)/k\right) = 2$.
    
    \item For $3m < k \leq n$:
    The prefix $w_1 \mathtt{a} w_2$ contains at most three distinct substrings for every $k$ and the substrings $w_3$ and $w_4$ contain no substrings of length $k > 3m$.
    The remaining distinct substrings must again contain the positions in $[6m+4,6m+5]$ or $[9m+4,9m+5]$. These substrings can also be described in a similar way to the previous case for $3 \leq k \leq 3m$, except for how we should remove duplicates. We have the two following sub-cases:
    \begin{itemize}
      \item For $k = 3m+1$: 
    Since $\mathtt{a^{\mathit{k}}} = \mathtt{a}^{3m+1}$ has no occurrences in $T$ but $\mathtt{(abb)^{\mathit{k}/3}}$ has other occurrences and it has already been counted,
    the number of such distinct substrings is at most $2(k-1)-1$.
      \item For $k > 3m+1$: There exists at least one substring which contains both $[6m+4,6m+5]$ and $[9m+4,9m+5]$.
        Therefore, the number of such distinct substrings is at most $2(k-1)-1$.
      \end{itemize}  
    Hence, $\Substr(T,k) \leq 3+2(k-1)-1 = 2k$ for every $3m < k \leq n$,
    which implies that  
    $\max_{3m < k \leq n}\left(\Substr(T,k)/k\right) \leq 2$.
    \end{itemize}
    Consequently, we have that $\delta(T)=2$.

    Consider the string $$T' = \mathtt{(abb)^\mathit{m} (bba)^{\mathit{m}+1} a^{3\mathit{m}} (bba)^\mathit{m}} = w_1 w_2 w_3 w_4$$ that can be obtained from $T$ by removing $T[3m+1] = \mathtt{a}$ between $w_1$ and $w_2$.
    We consider the number of distinct substrings of length $3m+1$ in $T'$:
    Because of the lengths of $w_j$ with $j \in \{ 1,2,3,4 \}$,
    each substring of length $3m+1$ is completely contained in $w_2$ or
    it contains some boundaries of $w_j$.
    \begin{itemize}
    \item The prefix
$w_1 (w_2[1..|w_2|-3]) = \mathtt{(abb)^\mathit{m} (bba)^\mathit{m}}$
      contains $3m$ distinct substrings of length $3m+1$.
    \item The substring $w_2$ contains 3 distinct substrings of length $3m+1$.
    \item The substring
$w_2[4..|w_2|] w_3=\mathtt{(bba)}^m \mathtt{a}^{3m}$
      contains $3m$ distinct substrings of length $3m+1$.
    \item The suffix $w_3 w_4 = \mathtt{a}^{3m} (\mathtt{bba})^\mathit{m}$ contains $3m-1$ distinct substrings of length $3m+1$ (note that $\mathtt{a (bba)^\mathit{m}}$ is a duplicate and is not counted here).
    \end{itemize}
    Hence,
    \[
    \delta(T') \geq \Substr(T,3m+1)/(3m+1) = \frac{9m+2}{3m+1} = 3-\frac{1}{3m+1}.
    \]
    Thus we obtain $\lim \inf_{n\to\infty}\MSensDel(\delta,n) \geq \lim \inf_{m\to\infty}((3-1/(3m+1))/2) \geq 1.5$ and \\ $\lim \inf_{n\to\infty}\ASensDel(\delta,n) \geq \liminf_{m\to\infty}((3-1/(3m+1))-2) = 1$.
\end{proof}

\subsection{Upper Bounds for the sensitivity of $\delta$}

\begin{theorem} \label{theo:upper_bound_delta}
  The following upper bounds on
  the sensitivity of $\delta$ hold: \\
  \textbf{substitutions:} $\MSensSub(\delta,n) \leq 2$.
    $\ASensSub(\delta,n) \leq 1$. \\
  \textbf{insertions:} $\MSensIns(\delta,n) \leq 2$.
    $\ASensIns(\delta,n) \leq 1$. \\
  \textbf{deletions:} $\lim \sup_{n\to\infty}\MSensDel(\delta,n) \leq 1.5$.
    $\lim \sup_{n\to\infty}\ASensDel(\delta,n) \leq 1$. 
\end{theorem}

\begin{proof}
  First we consider the additive sensitivity for $\delta$.
  For each $k$, the number of substrings of length $k$ that contains the edited position $i$ is clearly at most $k$.
  Therefore, after a substitution or insertion, at most $k$ new distinct substrings of length $k$ can appear in the string $T'$ after the modification.
  Also, after a deletion, at most $k-1$ new distinct substrings of length $k$ can appear in $T'$.
Hence, in the case of substitutions and insertions,
    $\delta(T') \leq \max_{1\leq k \leq n} ((\Substr(T,k)+k)/k) \leq \max_{1\leq k \leq n}(\Substr(T,k))/k)+\max_{1\leq k \leq n}(k/k)=\delta(T)+1$ holds.
    Also, in the case of deletions,
    $\delta(T') \leq \max_{1\leq k \leq n} ((\Substr(T,k)+k-1)/k) \leq \delta(T)+\max_{1\leq k \leq n}((k-1)/k)$ holds.
  Thus we obtain $\ASensSub(\delta,n) \leq 1$,
  $\ASensIns(\delta,n) \leq 1$, and $\lim \sup_{n\to\infty}\ASensDel(\delta,n) \leq \lim \sup_{k\to\infty} (k-1)/k = 1$.

  Next we consider the multiplicative sensitivity for $\delta$.
  Note that $\delta(T') \geq 1$ for any non-empty string $T'$,
  since $\Substr(T',1) \geq 1$.
  Combining this with the afore-mentioned additive sensitivity,
  we obtain $\MSensSub(\delta,n) \leq 2$ and $\MSensIns(\delta,n) \leq 2$.
  For the case of deletions,
  observe that $\delta(T) = 1$ only if $T$ is a unary string.
  However $\delta(T')$ cannot increase after a deletion
  since $T'$ is also a unary string.
  Thus we can restrict ourselves to the case where
  $T$ contains at least two distinct characters.
  Then, we have $\lim \sup_{n\to\infty}\MSensDel(\delta,n) \leq 1.5$,
  which is achieved when $\delta(T)=2$ and
  $\delta(T')=2+\frac{k-1}{k}$ with $k\to\infty$.
\end{proof}
 \section{String Attractors}
\label{sec:gamma}

In this section, we consider the worst-case sensitivity
of the string repetitiveness measure $\gamma$,
which is the size of the smallest string attractor~\cite{KempaP18}.
A string attractor $\Gamma(T)$ for a string $T$
is a set of positions in $T$ such that any substring $T$ has an occurrence containing a position in $\Gamma(T)$.
We denote the size of the smallest string attractor of $T$ by $\gamma(T)$.
It is known that $\gamma(T)$ is upper bounded by any of $\zorig(T)$, $r(T)$, $e(T)$ for any string $T$~\cite{KempaP18}.

In what follows, we present lower bounds
for the multiplicative sensitivity of $\gamma$ for all cases
of substitutions, insertions, and deletions.
We also present the additive sensitivity of $\gamma$.

\subsection{Lower bounds for the sensitivity of $\gamma$}

\begin{theorem}
  The following lower bounds on
  the sensitivity of $\gamma$ hold: \\
  \textbf{substitutions:} $\lim \inf_{n\to\infty}\MSensSub(\gamma,n) \geq 2$.
    $\ASensSub(\gamma,n) \geq \gamma-2$ and $\ASensSub(\gamma,n)= \Omega(\sqrt n)$. \\
  \textbf{insertions:} $\lim \inf_{n\to\infty}\MSensIns(\gamma,n) \geq 2$.
    $\ASensIns(\gamma,n) \geq \gamma-2$ and $\ASensIns(\gamma,n)= \Omega(\sqrt n)$. \\
  \textbf{deletions:} $\lim \inf_{n\to\infty}\MSensDel(\gamma,n) \geq 2$.
    $\ASensDel(\gamma,n) \geq \gamma-3$ and $\ASensDel(\gamma,n) = \Omega(\sqrt n)$.
\end{theorem}

\begin{proof}
  Consider string $T = \mathtt{a^\mathit{k} x a^{\mathit{k}+1} \#_1 a^{\mathit{k}-1} x a \#_2 a^{\mathit{k}-2} x a^2 \#_3 \cdots \#_{\mathit{k}} x a^\mathit{k}}$, where $\#_j$ for every $1 \leq j \leq k$ is a distinct character.
  The position where $\#_j$ for each $1 \leq j \leq k$ occurs has to be an element of any string attractor for $T$.
  Also, each of the intervals $[1,k+1]$ and $[k+2,2k+2]$ has to contain at least one element of any string attractor for $T$,
  since each of the substrings $T[1..k+1] = \mathtt{a^\mathit{k} x}$ and $T[k+2.. 2k+2] = \mathtt{a}^{\mathit{k}+1}$ occurs only once in $T$.
  Therefore, $\gamma(T)\geq k+2$ holds.
  Consider the set $S = \{ k+1, k+2, 2k+3, 3k+5,\ldots, k+1+k(k+2)\}$ of
  $k+2$ positions in $T$ which contains all the positions required above.
  Since each substring $\mathtt{a^{\mathit{k-j}} x a^{\mathit{j}}}$ of length $k+1$
  immediately preceded by $\#_j$~($1 \leq j \leq k$) occurs in the prefix $\mathtt{a^\mathit{k} x a^{\mathit{k}+1}}$ and contains the position $k+1$,
  $S$ is indeed a string attractor for $T$, we get $\gamma(T)=k+2$.
  In the following, we use this string $T$ for the analysis of
  lower bounds for the sensitivity of $\gamma$.

  \textbf{substitutions}:
  Let $T'$ be the string obtained by substituting
  the leftmost occurrence of $\mathtt{x}$ at position $k+1$ in $T$
  with character $\mathtt{b}$,
  yielding the new prefix $\mathtt{a^\mathit{k} b a^{\mathit{k}+1}}$
  right before $\#_1$.
  The size of the smallest string attractor for $T'$ is as follows:
    Each occurrence position of $\#_j$ for $1 \leq j \leq k$ still has to be an element of any string attractor for $T'$.
    Also, each of the intervals $[k+1]$ and $[k+2,2k+2]$ has to contain at least one element of any string attractor for $T'$.
    In addition, each of the intervals $[2k+4,3k+4], [3k+6,4k+6],\ldots,[k+2+k(k+2),2k+2+k(k+2)]$ which are the occurrences of substrings $\mathtt{a^{\mathit{k}-1} x a}, \mathtt{a^{\mathit{k}-2} x a^2}, \ldots, \mathtt{x a^\mathit{k}}$ has to contain one string attractor,
    since we have lost the prefix $\mathtt{a^\mathit{k} x a^{\mathit{k}+1}}$.
    Therefore, $\gamma(T')\geq 2k+2$ holds and the set $\{ k+1, k+2, 2k+3, 3k+5,\ldots, k+1+k(k+2), 2k+4, 3k+6, \ldots, k+2+k(k+2)\}$ of $2k+2$ positions in $T'$
    is a string attractor for $T'$, implying $\gamma(T')=2k+2$.
    Thus we get $\lim \inf_{n\to\infty}\MSensSub(\gamma,n) \geq \lim \inf_{k\to\infty} (2k+2)/(k+2) = 2$ and $\ASensSub(\gamma,n) \geq \gamma-2$.
    Since $n = k^2+4k+2$ and $\gamma(T)=k+2$,
    $\ASensSub(\gamma,n)= \Omega(\sqrt n)$ holds.
  
    \textbf{insertions}:
    Let $T'$ be the string obtained by
    inserting $\mathtt{b}$ between $T[k+1] = \mathtt{x}$ and $T[k+2] = \mathtt{a}$, yielding the new prefix $\mathtt{a^\mathit{k} xb a^{\mathit{k}+1}}$ right before $\#_1$.
    The size of the smallest string attractor for $T'$ is as follows,
    using a similar argument to the case of substitutions:
    Each occurrence position of $\#_j$ for $1 \leq j \leq k$ still has to be an element of any string attractor for $T'$.
    Also, each of the intervals $[k+2]$ and $[k+3,2k+3]$ have to contain at least one element of any string attractor for $T'$.
    In addition, each of the intervals $[2k+5,3k+5]$,
    $[3k+7,4k+7],\ldots,[k+3+k(k+2),2k+3+k(k+2)]$ which are the occurrences of substrings $\mathtt{(a^{\mathit{k}-1} x a)}, \mathtt{(a^{\mathit{k}-2} x a^2)}, \ldots, \mathtt{(x a^\mathit{k})}$ have to contain one string attractor.
    Therefore, $\gamma(T')\geq 2k+2$ holds and the set $\{ k+2, k+3, 2k+4, 3k+6,\ldots, k+2+k(k+2), 2k+5, 3k+7, \ldots, k+3+k(k+2)\}$ achieves $\gamma(T')=2k+2$.
    Thus we get $\lim \inf_{n\to\infty}\MSensIns(\gamma,n) \geq \lim \inf_{k\to\infty} (2k+2)/(k+2) = 2$,
    $\ASensIns(\gamma,n) \geq \gamma-3$, and $\ASensIns(\gamma,n)= \Omega(\sqrt n)$.
  
    \textbf{deletions}:
    Let $T'$ be the string obtained by
    deleting $T[k+1] = \mathtt{x}$ from $T$,
    yielding the new prefix $\mathtt{a^{2\mathit{k}+1}}$ right before $\#_1$.
    The size of the smallest string attractor for $T'$ is as follows,
    using a similar argument to the cases of insertions and substitutions:
    Each occurrence position of $\#_j$ for $1 \leq j \leq k$ still has to be an element of any string attractor for $T'$.
    Also, the interval $[1,2k+1]$ has to contain one element of any string attractor for $T'$.
    In addition, each of the intervals $[2k+3,3k+3], [3k+5,4k+5],\ldots,[k+1+k(k+2),2k+1+k(k+2)]$ has to contain one string attractor for $T'$.
    Therefore, $\gamma(T')\geq 2k+1$ holds and the set $\{ 1, 2k+2, 3k+4,\ldots, k+k(k+2), 2k+3, 3k+5, \ldots, k+1+k(k+2)\}$ achieves $\gamma(T')=2k+1$.
    Thus we get $\lim \inf_{n\to\infty}\MSensDel(\gamma,n) \geq \lim \inf_{k\to\infty} (2k+1)(k+2) = 2$, $\ASensDel(\gamma,n) \geq \gamma-3$, and $\ASensDel(\gamma,n)= \Omega(\sqrt n)$. 
\end{proof}
  
\subsection{Upper Bounds for the sensitivity of $\gamma$}

In this section, we present some upper bounds for the worst-case sensitivity of the smallest string attractor size $\gamma$.

We use the following known results:

\begin{theorem}[Lemma 3.7 of \cite{KempaP18}] \label{theo:z_gamma}
  For any string $T$, $\gamma(T) \leq \zsrss(T)$.
\end{theorem}

\begin{theorem}[Lemma 1 of \cite{KociumakaNP20}] \label{theo:z_delta}
  For any string $T$ of length $n$,
  $\zsrss(T) = O(\delta(T) \log (n / \delta(T)))$.
\end{theorem}

\begin{theorem}[Lemma 2 of \cite{KociumakaNP20}] \label{theo:delta_gamma}
  For any string $T$, $\gamma(T) \geq \delta(T)$.
\end{theorem}

We are ready to show our results:

\begin{corollary} \label{coro:upper_gamma}
  The following upper bounds on
  the sensitivity of $\gamma$ hold: \\
  \textbf{substitutions:} $\MSensSub(\gamma,n) = O(\log n)$.
    $\ASensSub(\gamma,n) = O(\delta \log n)$. \\
  \textbf{insertions:} $\MSensIns(\gamma,n) = O(\log n)$.
    $\ASensIns(\gamma,n) = O(\delta \log n)$. \\
  \textbf{deletions:} $\MSensDel(\gamma,n) = O(\log n)$.
    $\ASensDel(\gamma,n) = O(\delta \log n)$. \\
\end{corollary}

\begin{proof}
  Let $T$ be any string of length $n$,
  and let $T'$ be any string such that $\ed(T,T') = 1$.

  It follows from Theorem~\ref{theo:z_gamma} and Theorem~\ref{theo:z_delta} that
  $\gamma(T') \leq \zsrss(T') = O(\delta(T') \log n)$.
  Also, $\gamma(T) \geq \delta(T)$ by Theorem~\ref{theo:delta_gamma}
    and $\delta(T') = O(\delta(T))$ by Theorem~\ref{theo:upper_bound_delta}.
    Then, Lemma~\ref{lem:squeeze} leads that
    \[
    \frac{\gamma(T')}{\gamma(T)} = O\left(\frac{\delta(T') \log n}{\delta(T)}\right) = O(\log n).
    \]
  Similarly, 
  $\gamma(T')-\gamma(T) = O(\delta(T') \log n) \subseteq O(\delta(T) \log n)$ holds.
\end{proof}

\section{Run-Length Burrows-Wheeler Transform (RLBWT)}
\label{sec:rlbwt}

The \emph{Burrows-Wheeler transform} (\emph{BWT}) of a string $T$,
denoted $\BWT(T)$, is the string obtained by concatenating
the last characters of the lexicographically sorted suffixes of $T$.
The \emph{run-length BWT} (\emph{RLBWT}) of $T$ is
the run-length encoding of $\BWT(T)$
and $r(T)$ denotes its size, i.e., the number of maximal character runs in $\BWT(T)$.

For example, for string $T = \mathtt{abbaabababab}$,
$r(T) = 4$ since $\BWT(T) = \mathtt{babbbbbaaaaa}$ consists in four maximal character runs $\mathtt{b^1a^1b^5a^5}$.

\begin{theorem}[Theorem 1 of~\cite{GiulianiILPST21}] \label{theo:rlbwt_sofsem}
  There exists a family of strings $S$ such that $r(S) = 2$ and $r(S') = \Theta(\log n)$, where $n = |S|$ and $S'$ is a string obtained by prepending a character to $S$. The string $S$ is a reversed Fibonacci word. 
\end{theorem}

Theorem~\ref{theo:rlbwt_sofsem} immediately leads to the following lower bound
for the sensitivity of $r$:
\begin{corollary} \label{coro:rlbwt_lowerbounds}
    The following lower bound on the sensitivity of RLBWT with $|\Sigma| = 2$ hold: \\
    \textbf{insertions:} $\MSensIns(r,n) = \Omega(\log n)$. $\ASensIns(r,n) = \Omega(\log n)$.
\end{corollary}

To obtain a non-trivial upper bound for the sensitivity of $r$,
we can use the following known result:

\begin{theorem}[Theorem III.7 of~\cite{kempa2020resolution}] \label{theo_r_vs_delta}
  For any string $T$ of length $n$,
  $$r(T) = O\left(\delta(T) \max \left(1, \log \frac{n}{\delta(T) \log \delta(T)}\right) \log \delta(T) \right).$$
\end{theorem}

\begin{corollary} \label{coro:rlbwt_upperbounds}
The following upper bounds on the sensitivity of $r$ hold: \\
  \textbf{substitutions:} $\MSensSub(r,n) = O(\log n \log r)$. $\ASensSub(r,n) = O(r \log n \log r)$.\\
  \textbf{insertions:} $\MSensIns(r,n) = O(\log n \log r)$. $\ASensIns(r,n) = O(r \log n \log r)$.\\
  \textbf{deletions:} $\MSensDel(r,n) = O(\log n \log r)$. $\ASensDel(r,n) = O(r \log n \log r)$.
\end{corollary}

\begin{proof}
  For any string $T$, it is known that
  $\delta(T) \leq r(T)$~\cite{KempaP18,KociumakaNP20}.
  We also use a simplified and relaxed bound $r(T) = O(\delta(T) \log n \log \delta(T))$ from Theorem~\ref{theo_r_vs_delta}, which always holds and is sufficient for our purpose.

  Let $T'$ be any string with $\ed(T,T') = 1$.
  It follows from Theorem~\ref{theo:upper_bound_delta} that
  $\delta(T') \leq 2 \delta(T)$.
  Therefore, we obtain $r(T') = O(\delta(T') \log n \log \delta(T')) = O(\delta(T) \log n \log \delta(T)) = O(r(T) \log n \log r(T))$ by Lemma~\ref{lem:squeeze}. 
  This leads to the claimed upper bounds for the sensitivity for $r$.
\end{proof}

We remark that 
the lower bounds $\MSensIns(r,n) = \Omega(\log n)$ and $\ASensIns(r,n) = \Omega(\log n)$ from Theorem~\ref{theo:rlbwt_sofsem} and Corollary~\ref{coro:rlbwt_lowerbounds} are asymptotically tight when $r = O(1)$, since $\MSensIns(r,n) = O(\log n \log r) = O(\log n)$ and $\ASensIns(r,n) = \Omega(\log n)$ in this case. 
 \section{Bidirectional Scheme}
\label{sec:bs}

In this section, we consider the worst-case sensitivity of the size of \emph{bidirectional scheme}~\cite{StorerS82}.
A factorization $T = f_{1} \cdots f_{b}$ for a string $T$ of length $n$ is a bidirectional scheme of $T$
if each phrase $f_j=T[p_j..p_j+\ell_j-1]$ is
either a single character or corresponding to another substring $T[q_j..q_j+\ell_j-1]$ where $\ell_j = |f_j|$ such that $p_j \neq q_j$.
We denote $f_j$ either a single character or the pair $(q_j, \ell_j)$.
If $|f_j|=1$, then $f_j$ is called a ground phrase.
A bidirectional scheme $B$ for $T$ defines a function $F_B:[1..n]\cup\{ 0\}\rightarrow[1..n]\cup\{ 0\}$, where
\[
\begin{cases}
        F_B(p_j)=0, & \text{ if } f_j \text{ is a ground phrase,}\\
        F_B(p_j+k)=q_j+k, & \text{ if } f_j = (q_j, \ell_j) \text{ and } 0 \leq k < \ell_j,\\
        F_B(0)=0.
\end{cases}
\]
  Let $F_B^0(p_j)=p_j$ and $F_B^m(p_j)=F_B(F_B^{m-1}(p_j))$ for any $m\geq 1$.
  A bidirectional scheme $B$ is called \emph{valid} if $F_B$ has no cycles;
  namely, there exists an $m\geq 1$ such that $F_B^{m}(x) = 0$ for every $x \in [1..n]$.
The string $T$ can be reconstructed from the bidirectional scheme if and only if it is valid.
The \emph{size} of a valid bidirectional scheme $B$ is the number of phrases in $B$.
We denote by $b(T)$ the size of a valid bidirectional scheme for $T$
of the smallest size possible.

For example, for string $T = \mathtt{abaabababbbba}$,
$B$ shown below is a valid bidirectional scheme of the smallest size possible:
\begin{equation*}
  \begin{split}
    B=(4,3)(6,4)\mathtt{ab}(9,3)\mathtt{a},
  \end{split}
\end{equation*}
where its corresponding factorization is:
$$B=\mathtt{aba|abab|a|b|bbb|a|}.$$
Here we have $b(T) = 6$.

In what follows, we present upper and lower bounds
for the multiplicative/additive sensitivity of $b$.
It is noteworthy that	our upper and lower bounds for the multiplicative sensitivity of $b$ for substitutions and insertions are tight. 

\subsection{Lower bounds for the sensitivity of $b$}

\begin{theorem} \label{theo:l_bs_1}
  The following lower bounds on
  the sensitivity of $b$ hold:  \\
  \textbf{substitutions:} $\MSensSub(b,n) \geq 2$. \\
  \textbf{insertions:} $\MSensIns(b,n) \geq 2$.
\end{theorem}

\begin{proof}
  \textbf{substitutions}:
  Consider strings $T = \mathtt{a^\mathit{n}}$ and $T' = \mathtt{a^{\lceil\mathit{n}/2\rceil-1} b a^{\lfloor\mathit{n}/2\rfloor}}$.
  Then $b(T)=2$ and $b(T')=4$ hold.
  Thus we get $\MSensSub(b,n) \geq 2$. 
  
\textbf{insertions}:
  Consider strings $T = \mathtt{a^\mathit{n}}$ and $T' = \mathtt{a^{\lceil\mathit{n}/2\rceil} b a^{\lfloor\mathit{n}/2\rfloor}}$.
  Then $b(T)=2$ and $b(T')=4$ hold.
  Thus we get $\MSensIns(b,n) \geq 2$. 
\end{proof}

The family of strings used in Theorem~\ref{theo:l_bs_1} gives us tight lower bounds for multiplicative sensitivities.
However, this family of strings only provides us with weak lower bound $2$ for the additive sensitivity of $b$.
The following theorem will give us stronger lower bounds for the additive sensitivity for $b$.
We remark that this theorem also leads us to a non-trivial lower bound for the multiplicative sensitivity of $b$ in the case of deletions.

\begin{theorem} \label{theo:l_bs_2}
  The following lower bounds on
  the sensitivity of $b$ hold:  \\
  \textbf{substitutions:} $\ASensSub(b,n) \geq b/2-1$, and $\ASensSub(b,n)=\Omega(\sqrt n)$. \\
  \textbf{insertions:} $\ASensIns(b,n) \geq b/2-1$, and $\ASensIns(b,n)=\Omega(\sqrt n)$. \\
  \textbf{deletions:} $\lim \inf_{n\to\infty}\MSensDel(b,n) \geq 1.5$, $\ASensDel(b,n) \geq b/2-3$, and $\ASensDel(b,n)=\Omega(\sqrt n)$.
\end{theorem}

\begin{proof}
  Consider string $$T = \mathtt{a^\mathit{k} x a^{\mathit{k}+1} \#_1 a^{\mathit{k}} x a \#_2 a^{\mathit{k}-1} x a^2 \#_3 \cdots \#_{\mathit{k}} a x a^\mathit{k}},$$ where $\#_j$ for every $1 \leq j \leq k$ is a distinct character.
  One of the valid bidirectional schemes $B$ for $T$ is
  \begin{equation*}
    \begin{split}
      B=(k+2,k)\mathtt{xa}(k+2,k)\mathtt{\#_1}(1,k+2)\mathtt{\#_2}(2,k+2)\mathtt{\#_3}\cdots\mathtt{\#_{\mathit{k}}}(k,k+2).
    \end{split}
  \end{equation*}
  The corresponding factorization of the above bidirectional scheme is as follows:
  $$B=\mathtt{a^{\mathit{k}}|x|a|a^{\mathit{k}}|\#_1| a^{\mathit{k}} x a|\#_2|a^{\mathit{k}-1} x a^2|\#_3|\cdots|\#_{\mathit{k}}|a x a^{\mathit{k}}}|.$$
  The size of $B$ is $2k+4$ and thus $b(T)\leq 2k+4$.
  
  As for substitutions, let $T'$ be the string obtained by substituting
  the leftmost occurrence of $\mathtt{x}$ at position $k+1$ in $T$
  with a character $\mathtt{y}$ such that $\mathtt{y} \neq \mathtt{x}$, that is,
  $$T' = \mathtt{a^\mathit{k} y a^{\mathit{k}+1} \#_1 a^{\mathit{k}} x a \#_2 a^{\mathit{k}-1} x a^2 \#_3 \cdots \#_{\mathit{k}} a x a^\mathit{k}}.$$
  Then, one of the valid bidirectional schemes $B'$ of $T'$ is:
  \begin{equation*}
    \begin{split}
      B'=(k+2,k)\mathtt{ya}(k+2,k)\mathtt{\#_1}(1,k)\mathtt{xa}\mathtt{\#_2}(2k+5,k)(1,2)\mathtt{\#_3}\cdots\mathtt{\#_{\mathit{k}}}(3k+4,2)(1,k).
    \end{split}
  \end{equation*}
  Also, the corresponding factorization for $B$ is as follows:
  $$B'=\mathtt{a^{\mathit{k}}|y|a|a^{\mathit{k}}|\#_1| a^{\mathit{k}}|x|a|\#_2|a^{\mathit{k}-1} x| a^2|\#_3|\cdots|\#_{\mathit{k}}|a x |a^{\mathit{k}}}|.$$
    The size of $B'$ is $3k+5$.
    We show that $B'$ is a valid bidirectional scheme for $T'$ of the smallest size possible,
namely, $b(T') = 3k+5$.
    Since $\mathtt{y}$ and $\#_j$ for every $1 \leq j \leq k$ are unique characters in $T'$,
    they have to be ground phrases.
    Also, since each substring $\mathtt{a^{\mathit{k-j}+1} x a^{\mathit{j}}}$ of length $k+2$ for all $1 \leq j \leq k$ and $\mathtt{a}^{k+1}$ are unique in $T'$,
    each corresponding interval has to have at least one boundary of phrases.
    In addition, at least one occurrence of $\mathtt{x}$ has to be a ground phrase.
    Then, $b(T')=3k+5$ holds.
    Since $|T| = n = k^2+5k+2$, we have $k = \Theta(\sqrt{n})$.
    Hence, we get $\lim \inf_{n\to\infty}\MSensSub(b,n) \geq 1.5$ and
    $\ASensSub(b,n) \geq k+1 = b/2-1 =\Omega(\sqrt n)$.

    Moreover, by considering the case where the character $T[k+1]$ is deleted and the case where the character $\mathtt{y}$ is inserted between positions $k+1$ and $k+2$, we obtain Theorem~\ref{theo:l_bs_2}.  
\end{proof}

\subsection{Upper bounds for the sensitivity of $b$}

\begin{theorem} \label{theo:u_bs}
  The following upper bounds on
  the sensitivity of $b$ hold:  \\
  \textbf{substitutions:} $\limsup_{n \rightarrow \infty}\MSensSub(b,n) \leq 2$.
    $\ASensSub(b,n) \leq b+2$. \\
  \textbf{insertions:} $\MSensIns(b,n) \leq 2$.
    $\ASensIns(b,n) \leq b$. \\
  \textbf{deletions:} $\limsup_{n \rightarrow \infty}\MSensDel(b,n) \leq 2$.
    $\ASensDel(b,n) \leq b+1$.
\end{theorem}

\begin{proof}
  In the following, we consider the case that $T[i]=a$ is substituted by a character $\#$ that does not occur in $T$.
  The other cases of insertions, deletions, and substitutions with another character $b~(\neq a)$ occurring in $T$, can be proven similarly.
  We show how to construct a valid bidirectional scheme of $T'$ of the size $b' \geq b (T')$ by dividing each phrase of $B$ into some phrases,
  where $B$ is a valid bidirectional scheme for $T$ of the smallest size possible.
We categorize each phrase $f_j=T[p_j..p_j+\ell_j-1]$ of $B$ into one of the three following cases:
  \begin{itemize}
    \item[(1)] $i \in [p_j..p_j+\ell_j-1]$;
    \item[(2)] $i \notin [p_j..p_j+\ell_j-1]$ and $i \notin [q_j..q_j+\ell_j-1]$;
    \item[(3)] $i \notin [p_j..p_j+\ell_j-1]$ and $i \in [q_j..q_j+\ell_j-1]$.
  \end{itemize}
  
  \noindent \textbf{Case (1):} Let $T[p_j..p_j+\ell_j-1]= w_1 a w_2$ and $T'[p_j..p_j+\ell_j-1]=w_1 \# w_2$, where $a \in \Sigma$ and $w_1, w_2 \in \Sigma^*$.
  If $i \notin [q_j..q_j+\ell_j-1]$, 
  then the phrase $f_j$ is divided into three phrases $w_1=(q_j,|w_1|), \#, w_2=(q_j+|w_1|+1, |w_2|)$ in $T'$.
  See also the top of Figure~\ref{fig:bs}.
  Otherwise, i.e., if $i \in [q_j..q_j+\ell_j-1]$,
  intervals $[p_j..p_j+\ell_j-1]$ and $[q_j..q_j+\ell_j-1]$ are overlapping.
  We consider the case $p_j < q_j$. (Another case can be treated similarly.)
  Then $[q_j..q_j+|w_1|]$ contains the edited position $i$.
  Let $T[p_j..p_j+\ell_j-1]=w'_1 a w'_2 a w_2$, where $w'_1, w'_2 \in \Sigma^*$ and $q_j+|w'_1|=i$.
  We divide the phrase $f_j$ into at most five phrases $w'_1=(q_j,|w'_1|), a, w'_2=(q_j+|w'_1|+1, |w'_2|),\#, w_2=(q_j+|w_1|+1, |w_2|)$.
  See also the middle of Figure~\ref{fig:bs}.

  \noindent \textbf{Case (2):} No changes are made to the phrase $f_j$ in this case, since $f_j$ can continue to refer to the same reference.

  \noindent \textbf{Case (3):} Among all phrases in Case (3), let $f_k$ be the phrase whose ending position of the reference is the rightmost.
  Let $T[p_k..p_k+\ell_k-1]=u_1 a u_2$, where $u_1, u_2 \in \Sigma^*$ and $q_k+|u_1|=i$.
  Then we divide the phrase $f_k$ into at most three phrases $u_1=(q_k,|u_1|), a, u_2=(q_k+|u_1|+1, |u_2|)$ in $T'$.
  For the other phrases of Case(3), we divide $f_j=v_1 a v_2$, where $v_1, v_2 \in \Sigma^*$ and $q_j+|v_1|=i$, into at most two phrases $v_1 = (q_j,|v_1|)$ and $a v_2=(q_k+|u_1|, |v_2|+1)$.
  From the above operations, the character that referred to position $i$ in $T$ becomes a ground phrase or refers to position $q_k+|u_1|$, which is a ground phrase, in $T'$.
  The other substrings refer to the original reference positions or to a subinterval of $[q_k+|u_1|..q_k+|f_k|-1]$.
  The reference of the subinterval corresponds to the original reference of the substring.
  See also the bottom of Figure~\ref{fig:bs}.
  
  Then, the bidirectional scheme obtained from the above operations is ensured to be valid.
  The size of the bidirectional scheme $b'$ is maximized if exactly one phrase of Case (1) is divided into five phrases, and the remaining $b(T)-1$ phrases belong to Case (3).
  Since at most one of the $b(T)-1$ phrases of Case (3) can be divided into three phrases, and all the others can be divided into two phrases, $b'$ is at most $5+3+2(b(T)-2)=2b(T)+4$.
  Furthermore, if $T$ is a unary string, then $b(T)=2$ and the valid bidirectional scheme of size $4 (=2b(T))$ can be constructed easily.
  Otherwise, there are at least two ground phrases in $T$, and these phrases can not be divided into some phrases in $T'$.
  Then we get $b' \leq 2b(T)+2$ and Theorem~\ref{theo:u_bs}.
  \begin{figure}[htpb]
  \centering
          \includegraphics[width=90mm]{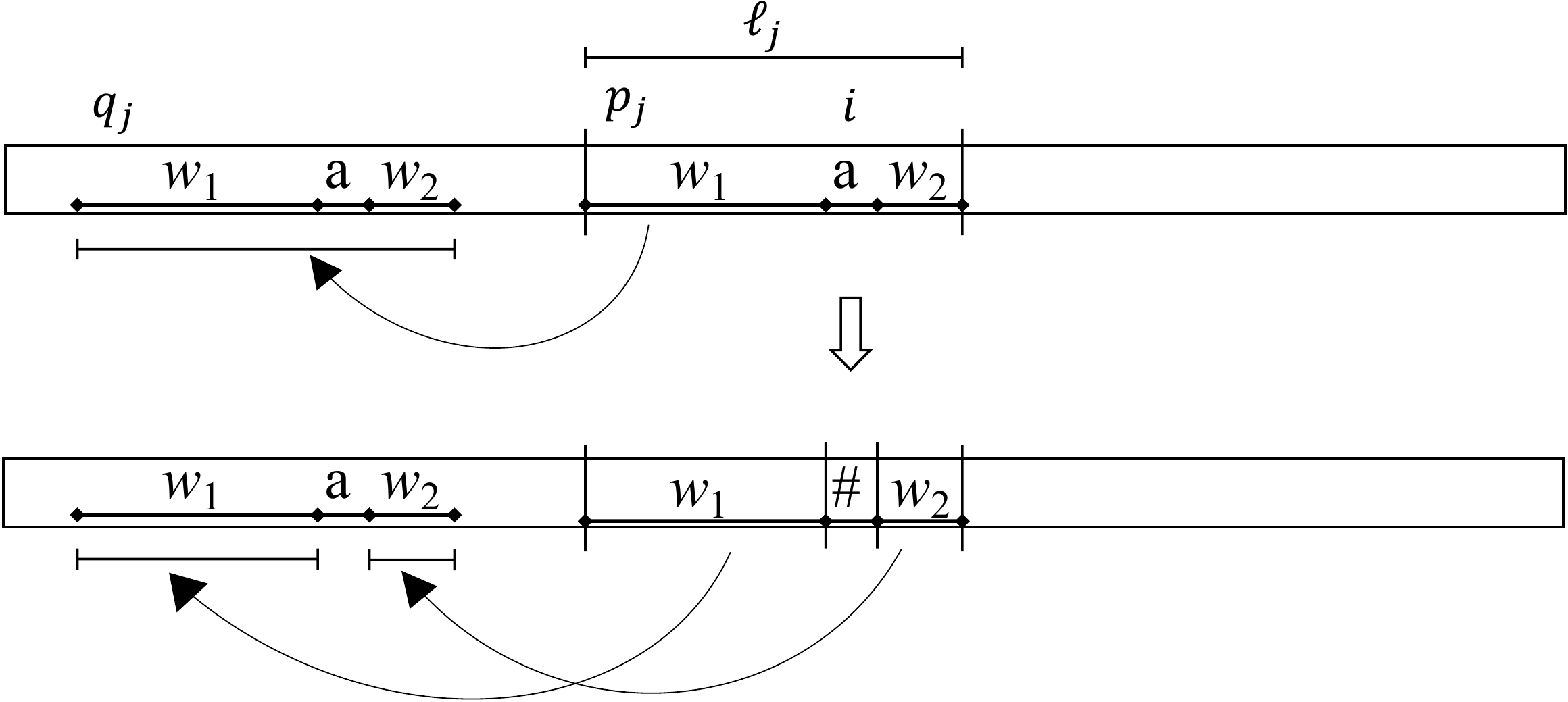}\\
          \subcaption{Subcase of Case (1): $i \in [p_j..p_j+\ell_j-1]$ and $i \notin [q_j..q_j+\ell_j-1]$.}
          \vspace{8mm}
          \includegraphics[width=90mm] {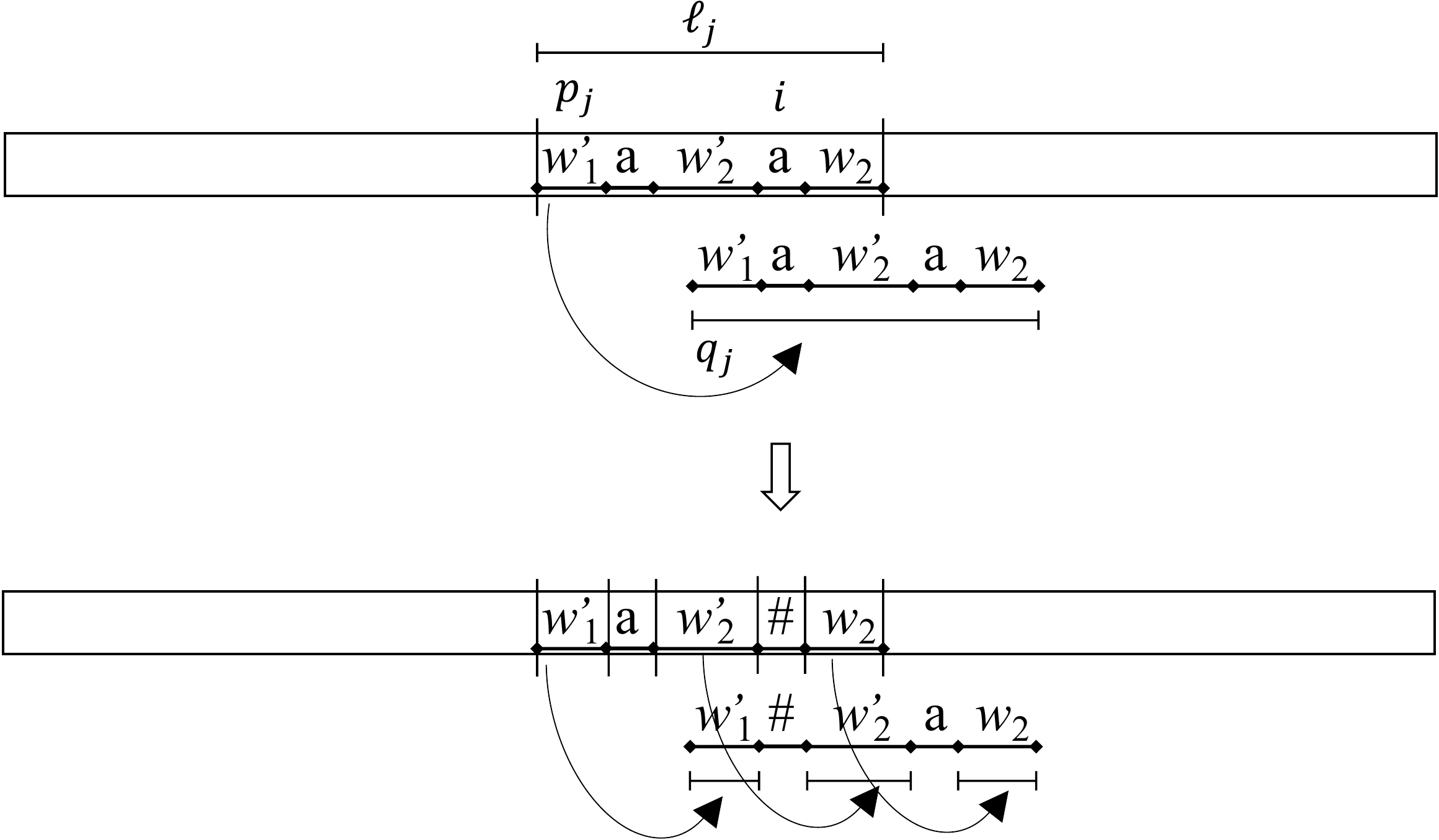}\\
          \subcaption{Subcase of Case (1): $i \in [p_j..p_j+\ell_j-1]$ and $i \in [q_j..q_j+\ell_j-1]$.}
          \vspace{8mm}
          \includegraphics[width=90mm]{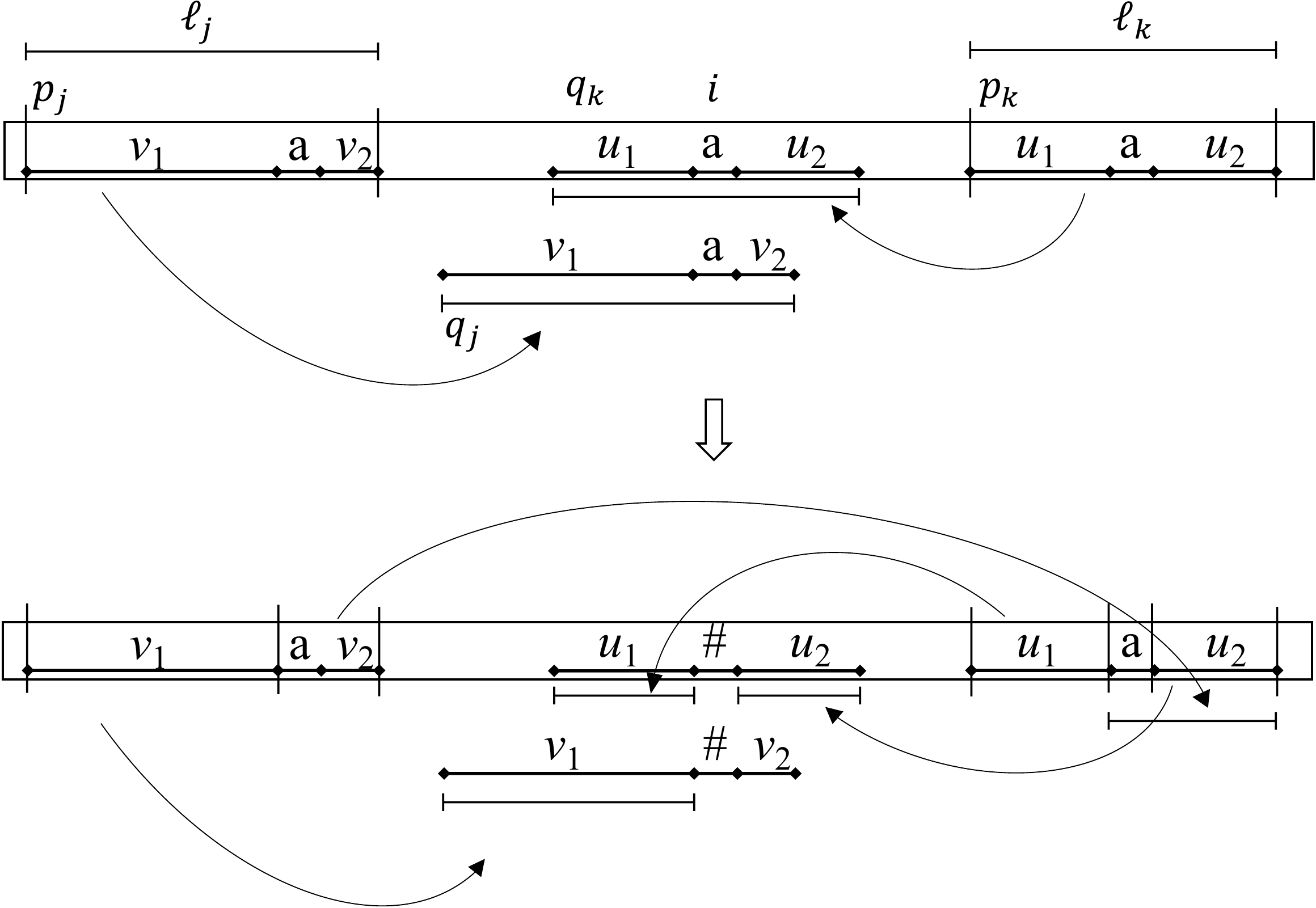}\\
          \subcaption{Case (3): $i \notin [p_j..p_j+\ell_j-1]$ and $i \in [q_j..q_j+\ell_j-1]$.}
     \caption{Illustration for changes of references in Case (1) and Case (3).}
     \label{fig:bs}
\end{figure}

\end{proof}
 \section{Lempel-Ziv 77 factorizations with/without self-references}
\label{sec:lz77}

In this section, we consider the worst-case sensitivity
of the \emph{Lempel-Ziv 77 factorizations} (\emph{LZ77})~\cite{LZ77}
with/without self-references.

For convenience, let $f_0 = \varepsilon$.
A factorization $f_{1} \cdots f_{z}$ for a string $T$ of length $n$
is the non self-referencing LZ77 factorization $\LZorig(T)$ of $T$
if for each $1 \leq i < z$ the factor $f_i$ is
the shortest prefix of $f_i \cdots f_z$ that does not occur
in $f_0 f_1 \cdots f_{i-1}$
(or alternatively $f_k[1..|f_k|-1]$ is the longest prefix of $f_i \cdots f_z$
that occurs in $f_0 f_1 \cdots f_{i-1}$).
Since $f_k[1..|f_k|-1]$ never overlaps with its previous occurrence,
it is called non self-referencing.
The last factor $f_z$ is the suffix of $T$ of length $n-|f_1 \cdots f_{z-1}|$
and it may have multiple occurrences in $f_1 \cdots f_z$.

A factorization $f_{1} \cdots f_{z}$ for a string $T$ of length $n$
is the self-referencing LZ77 factorization $\LZorigsr(T)$ of $T$
if for each $1 \leq i < z$ the factor $f_i$ is
the shortest prefix of $f_i \cdots f_z$ that occurs exactly once in
$f_1 \cdots f_i$ as a suffix
(or alternatively $f_k[1..|f_k|-1]$ is the longest prefix of $f_i \cdots f_z$
which has a previous occurrence beginning at a position in range
$[1..|f_1 \cdots f_{k-1}|]$).
Since $f_k[1..|f_k|-1]$ may overlap with its previous occurrence,
it is called self-referencing.
The last factor $f_z$ is the suffix of $T$ of length $n-|f_1 \cdots f_{z-1}|$
and it may have multiple occurrences in $f_1 \cdots f_z$.

If we use a common convention
that the string $T$ terminates with a unique character $\$$,
then the last factor $f_z$ satisfies the same properties as $f_1, \ldots, f_{z-1}$,
in both cases of (non) self-referencing LZ77 factorizations.

To avoid confusions, we use different notations to denote the sizes
of these factorizations.
For a string $T$
let $\zorig(T)$ and $\zsrorig(T)$ denote
the number $z$ of factors in $\LZorig(T)$ and $\LZorigsr(T)$, respectively.

For example, for string $T = \mathtt{abaabababababab\$}$,
\begin{eqnarray*}
\LZorig(T) & = & \mathtt{a|b|aa|bab|ababa|bab\$}|, \\
\LZorigsr(T) & = & \mathtt{a|b|aa|bab|abababab\$}|,
\end{eqnarray*}
where $|$ denotes the right-end of each factor in the factorizations.
Here we have $\zorig(T) = 6$ and $\zsrorig(T) = 5$.

In what follows, we present tight upper and lower bounds
for the multiplicative sensitivity of $\zorig$ and $\zsrorig$ for all cases
of substitutions, insertions, and deletions.
We also present the additive sensitivity of $\zorig$ and $\zsrorig$.

\subsection{Lower bounds for the sensitivity of $\zorig$}

\begin{theorem} \label{theo:lz77_lowerbounds}
  The following lower bounds on
  the sensitivity of non self-referencing LZ77 factorization hold:  \\
  \textbf{substitutions:} $\liminf_{n \rightarrow \infty}\MSensSub(\zorig,n) \geq 2$. $\ASensSub(\zorig,n) \geq \zorig-1$. \\
  \textbf{insertions:} $\liminf_{n \rightarrow \infty}\MSensIns(\zorig,n) \geq 2$. $\ASensIns(\zorig,n) \geq \zorig-1$. \\
  \textbf{deletions:} $\liminf_{n \rightarrow \infty}\MSensDel(\zorig,n) \geq 2$. $\ASensDel(\zorig,n) \geq \zorig-2$.
\end{theorem}

\begin{proof}
Let $p \geq 2$ and $\Sigma = \{\mathtt{0, 1, 2}\}$.
We use the following string $T$ for our analysis
in all cases of substitutions, insertions, and deletions.

Let $Q_1 = \mathtt{0}$ and $Q_k = Q_1 \cdots Q_{k-1} \mathtt{1}$ with $2 \leq k \leq p$.
Let 
\begin{eqnarray*}
  T & = & Q_1Q_2 \cdots Q_p \\
    & = & \mathtt{0} \cdot \mathtt{01} \cdot \mathtt{0011} \cdot \mathtt{00100111} \cdot \mathtt{0010011001001111} \cdots Q_p
\end{eqnarray*}
with $|T| = n = \Theta(2^p)$.
Since $Q_k[1..|Q_k|-1] = T[1..|Q_k|-1]$,
$Q_k[|Q_k|] = \mathtt{1}$, and $T[|Q_k|]=\mathtt{0}$
for $2 \leq k \leq p$,
each $Q_k$ forms a single factor in the non self-referencing LZ77 factorization of $T$. Namely,
\begin{eqnarray*}
  \LZorig(T) & = & Q_1|Q_2| \cdots |Q_p| \\
  & = & \mathtt{0} | \mathtt{01} | \mathtt{0011} | \mathtt{00100111} | \mathtt{0010011001001111}| \cdots |Q_p|
\end{eqnarray*}
with $\zorig(T) = p = \Theta(\log n)$.

\textbf{substitutions:}  
Consider the string 
\begin{eqnarray*}
  T' & = & \mathtt{2} \cdot T[2..n] \\
  & = & \mathtt{2} \cdot Q_2 \cdots Q_p \\
  & = & \mathtt{2} \cdot \mathtt{01} \cdot \mathtt{0011} \cdot \mathtt{00100111} \cdot \mathtt{0010011001001111} \cdots Q_p
\end{eqnarray*}
which can be obtained from $T$ by substituting the first $\mathtt{0}$ with $\mathtt{2}$.
Let us analyze the structure of the non self-referencing LZ77 factorization $\LZorig(T')$ of $T'$.
We prove by induction that $Q_k$ is divided into exactly two factors
for every $2 \leq k \leq p$ in $\LZorig(T')$.
$Q_2$ is factorized as $\mathtt{0|1|}$ in $\LZorig(T')$.
Suppose that $Q_{k-1}$ is divided into exactly two factors in $\LZorig(T')$,
which means that the next factor is a prefix of $Q_k \cdots Q_p$.
Since $T'[1] = \mathtt{2}$, each $Q_k[1..|Q_k|-1]$ cannot occur
as a prefix of $T'$.
The longest prefix of $Q_k = Q_1 \cdots Q_{k-1}\mathtt{1}$ that occurs
in $T'[1..|Q_1 \cdots Q_{k-1}|]$ is 
$Q_{k-1}[1..|Q_{k-1}|-1] = Q_1 \cdots Q_{k-2}$.
Thus, $Q_1 \cdots Q_{k-2}\mathtt{0}$ is the shortest prefix of
$T'[|Q_1 \cdots Q_{k-1}|+1..n] = Q_k \cdots Q_p$
that does not occur in $T'[1..|Q_1 \cdots Q_{k-1}|] = Q_1 \cdots Q_{k-1}$.
The remaining suffix of $Q_k$ is $Q_{k-1}[2..|Q_{k-1}|] \mathtt{1} = Q_2 \cdots Q_{k-2} \mathtt{1}\mathtt{1}$.
Since $Q_k$ has $\mathtt{01}^{k-1}$ as a suffix
and this is the leftmost occurrence of $\mathtt{1}^{k-1}$ in $T'$,
the next factor is this remaining suffix $Q_2 \cdots Q_{k-2} \mathtt{1}\mathtt{1}$ of $Q_k$.
Thus, the non self-referencing LZ77 factorization of $T'$ is
\[
 \LZorig(T') = \mathtt{2|0|1|00|11|0010|0111|00100110|01001111| \cdots }|Q_1 \cdots Q_{p-2}\mathtt{0} | Q_2 \cdots Q_{p-2}\mathtt{11} |
\]
with $\zorig(T') = 2p-1$, which leads to
$\liminf_{n \rightarrow \infty}\MSensSub(\zorig,n) \geq \liminf_{p \to \infty} ((2p-1)/p) = 2$, $\ASensSub(\zorig,n) \geq (2p-1)-p = p-1 = \zorig-1 = \Omega(\log n)$.

\textbf{insertions:}
Let $T'$ be the string obtained by inserting
$\mathtt{2}$ immediately after the first character $T[1] = \mathtt{0}$,
namely,
\begin{eqnarray*}
  T' & = & Q_1 \cdot \mathtt{2} \cdot Q_2 \cdots Q_p \\
  & = & \mathtt{0} \cdot \mathtt{2} \cdot \mathtt{01} \cdot \mathtt{0011} \cdot \mathtt{00100111} \cdot \mathtt{0010011001001111} \cdots Q_p.
\end{eqnarray*}
Then, by similar arguments to the case of substitutions,
we have
\[
 \LZorig(T') = \mathtt{0|2|01|00|11|0010|0111|00100110|01001111| \cdots }|Q_1 \cdots Q_{p-2}\mathtt{0} | Q_2 \cdots Q_{p-2}\mathtt{11} |
\]
with $\zorig(T') = 2p-1$,
which leads to $\liminf_{n \rightarrow \infty}\MSensIns(\zorig,n) \geq \liminf_{p \to \infty} ((2p-1)/p) = 2$, $\ASensIns(\zorig,n) \geq (2p-1)-p = p-1 = \zorig-1 = \Omega(\log n )$.

\textbf{deletions:}
Let $T'$ be the string obtained by deleting
the first character $T[1] = \mathtt{0}$,
namely
\begin{eqnarray*}
  T' & = & Q_2 \cdots Q_p \\
  & = & \mathtt{01} \cdot \mathtt{0011} \cdot \mathtt{00100111} \cdot \mathtt{0010011001001111} \cdots Q_p.
\end{eqnarray*}
Then, by similar arguments to the case of substitutions,
we have
\[
 \LZorig(T') = \mathtt{0|1|00|11|0010|0111|00100110|01001111| \cdots }|Q_1 \cdots Q_{p-2}\mathtt{0} | Q_2 \cdots Q_{p-2}\mathtt{11} |
\]
with $\zorig(T') = 2p-2$,
which leads to $\liminf_{n \rightarrow \infty}\MSensDel(\zorig,n) \geq \liminf_{p \to \infty} ((2p-2)/p) = 2$, $\ASensDel(\zorig,n) \geq (2p-2)-p = p-2 = \zorig-2 = \Omega(\log n )$.
\end{proof}

The strings $T$ and $T'$ used in Theorem~\ref{theo:lz77_lowerbounds}
give us optimal additive lower bounds in terms $\zorig$,
are highly compressible ($\zorig(T) = O(\log n)$)
and only use two or three distinct characters.
By using more characters,
we can obtain larger lower bounds for the additive sensitivity
for the size of the non self-referencing LZ77 factorizations $\LZorig$
in terms of the string length $n$, as follows:

\begin{theorem} \label{theo:lz77_lowerbounds_sqrtn}
  The following lower bounds on
  the sensitivity of non self-referencing LZ77 factorization $\LZorig$ hold: \\
  \textbf{substitutions:} $\ASensSub(\zorig,n) = \Omega(\sqrt{n})$. \\
  \textbf{insertions:} $\ASensIns(\zorig,n) = \Omega(\sqrt{n})$. \\
  \textbf{deletions:} $\ASensDel(\zorig,n) = \Omega(\sqrt{n})$.
\end{theorem}

\begin{proof}
  In~\ref{app:lz77_lowerbounds_sqrtn}.
\end{proof}

\subsection{Upper bounds for the sensitivity of $\zorig$} \label{sec:u_zorig}

\begin{theorem} \label{theo:u_zorig}
  The following upper bounds on
  the sensitivity of non self-referencing LZ77 factorization $\LZorig$ hold:  \\
  \textbf{substitutions:} $\limsup_{n \rightarrow \infty}\MSensSub(\zorig,n) \leq 2$.
    $\ASensSub(\zorig,n) \leq \zorig-1$. \\
  \textbf{insertions:} $\limsup_{n \rightarrow \infty}\MSensIns(\zorig,n) \leq 2$.
    $\ASensIns(\zorig,n) \leq \zorig-1$. \\
  \textbf{deletions:} $\limsup_{n \rightarrow \infty}\MSensDel(\zorig,n) \leq 2$.
    $\ASensDel(\zorig,n) \leq \zorig-2$.
\end{theorem}

\begin{proof}
  In the following, we consider the case that $T[i]=a$ is substituted by a character $\#$ that does not occur in $T$.
  The other cases of insertions, deletions, and substitutions with another character $b$~($\neq a$) occurring in $T$, can be proven similarly,
  which will be discussed at the end of the proof.

  We denote the factorizations as
  $\LZorig(T) = f_{1} \cdots f_{z}$ and $\LZorig(T') =f'_{1} \cdots f'_{z'}$.
  We denote the interval of factor $f_j$ (resp. $f'_j$) by $[p_j,q_j]$ (resp. $[p'_j, q'_j]$).

  Now we prove the following claim:
  \begin{claim}
    Each interval $[p_j,q_j]$ has at most two starting positions $p'_k$ and $p'_{k+1}$ of factors in $\LZorig(T')$ for some $1 \leq k < z'$.
  \end{claim}
  
  \begin{claimproof}
    There are the three following cases:
    \begin{itemize}
      \item[(1)] When the interval $[p_j, q_j]$ satisfies $q_j < i$:
      $f_j=f'_j$ holds for any such $j$.
      Therefore, in the interval $[p_j, q_j]$ there exists exactly one starting position $p'_j = p_j$ of a factor in $\LZorig(T')$.
      
      \item[(2)] When the interval $[p_j, q_j]$ satisfies $p_j \leq i \leq q_j$:
        Let $T[p_j..q_j] = w_{1} a w_{2} c$ and $T'[p_j..q_j] = w_1 \# w_2 c$,
        where $a, c, \# \in \Sigma$ and $w_1, w_2 \in \Sigma^*$.
      By definition, $w_1 a w_2$ has at least one previous occurrence in $f_1 \cdots f_{j-1}$.
      After the substitution, $w_1 \#$ becomes a factor $f'_j$ of $\LZorig(T')$ since $\#$ is a fresh character, and $w_2 c$ becomes a prefix of the next factor $f'_{j+1}$ in $\LZorig(T')$.
      This means that $p'_j=p_j$ and $q'_{j+1}\geq q_j$.
      Therefore, the interval $[p_j, q_j]$ has at most two starting positions $p'_j$ and $p'_{j+1}$ of factors in $\LZorig(T')$.
      
      \item[(3)] When the interval $[p_j, q_j]$ satisfies $i < p_j$:
        There are the two following sub-cases:
        \begin{itemize}
          \item[(3-A)] When $T[p_j..q_j-1]$ has a previous occurrence which does not contain the edited position $i$ in $T$:
          In this case, any suffix of $T[p_j..q_j-1]$ has a previous occurrence in $T'$.
          Therefore, $[p'_k, q'_k]$ with $p_j \leq p'_k$ satisfies $q'_k \geq q_j$.
          Hence, the interval $[p_j..q_j]$ has at most one starting position $p'_k$ of a factor in $\LZorig(T')$.
          
        \item[(3-B)] When all previous occurrences of $T[p_j..q_j-1]$ in $T$ contain the edited position $i$:
          Let $u_{1} a u_{2} d = T[p_j, q_j]$ with
          $a, d \in \Sigma$ and $u_1, u_2 \in \Sigma^*$.
          $u_1$ and $u_2$ have previous occurrences in $T'[1..p_j-1]$.
          Let $p'_k$ be the starting position of the leftmost factor of $\LZorig(T')$
          which begins in range $[p_j, q_j]$. 
          If $p'_k$ is in $u_2$, then $q'_k \geq q_k$
          and thus there is only one starting position of a factor of $\LZorig(T')$ in the interval $[p_j..q_j]$.
          Suppose $p'_k$ is in $u_1$.
          If $a$ has no previous occurrences
          (which happens when $T[i]$ was the only previous occurrence of $a$),
          then $T'[p_k+|u_1|]$ is the first occurrence of $a$
          and thus $q'_k = p_k+|u_1|$.
          Otherwise, $q'_k \geq p_k+|u_1|$.
          In either case, since $u_2$ has a previous occurrence,
          $q'_{k+1} \geq q_{k+1}$.
          Thus, there can exist at most two starting positions of factors of $\LZorig(T')$ in the interval $[p_j..q_j]$.

\end{itemize}
    \end{itemize}
  This completes the proof for the claim. 
  \end{claimproof}
  
  By the above claim, $\zorig(T') \leq 2 \zorig(T)$ holds for any string $T$ and any substitution operation.
  Since $f_1$ consists of a single character for any string and the interval $[1,1]$ cannot have two starting positions of factors in $\LZorig(T')$, $\zorig(T') \leq 2 \zorig(T)-1$ holds. This completes the proof for the case of substitution with $\#$.

  The above proof can be generalized to all the other cases,
  by replacing $\#$ in $T'$ as follows:
  \begin{itemize}
  \item $\# \leftarrow b$ for substitutions with character $b$ occurring in $T$, where we have $T'[p_j, q_j] = w_1 b w_2 c$ for Case (2);
  \item $\# \leftarrow T[i]\#$ for insertions with $\#$, where we have $T'[p_j, q_j] = w_1 T[i]\# w_2 c$ for Case (2);
  \item $\# \leftarrow T[i]b$ for insertions with character $b$ occurring in $T$, where we have $T'[p_j, q_j] = w_1 T[i]b w_2 c$ for Case (2);
  \item $\# \leftarrow \varepsilon$ for deletions, where we have $T'[p_j, q_j] = w_1w_2 c$ for Case (2).
  \end{itemize}
The analysis for Case (2) and Case (3) is analogous for all these cases.
Also, in the case of deletions, since $|f_2|\leq 2$ and the interval can have two starting positions of factors in $\LZorig(T')$ only when $f_1=T[1]$ is deleted, $\zorig(T') \leq 2 \zorig(T)-2$ holds.
\end{proof}

\subsection{Lower bounds for the sensitivity of $\zsrorig$}

\begin{theorem} \label{theo:LZsr_lowerbounds}
  The following lower bounds on
  the sensitivity of self-referencing LZ77 factorization $\LZorigsr$ with $|\Sigma| = 3$ hold: \\
  \textbf{substitutions:} $\MSensSub(\zsrorig, n) \geq 2$. $\ASensSub(\zsrorig,n) \geq \zsrorig$. \\
  \textbf{insertions:} $\MSensIns(\zsrorig, n) \geq 2$. $\ASensIns(\zsrorig,n) \geq \zsrorig$. \\
  \textbf{deletions:} $\liminf_{n \rightarrow \infty}\MSensDel(\zsrorig,n) \geq 2$. $\ASensDel(\zsrorig,n) \geq \zsrorig -2$.
\end{theorem}

\begin{proof}
\textbf{substitutions:}
Let $p \geq 2$ and $\Sigma = \{\mathtt{0,1,2}\}$.
We use the following string $T$ for our analysis.

Let $R_1 = \mathtt{00}$ and $R_k = R_1 \cdots R_{k-1} \mathtt{1}$
with $2 \leq k \leq p$.
Consider the following string $T$ of length $n = \Theta(2^p)$:
\begin{eqnarray*}
  T & = & R_1 \cdots R_p \\
  & = & \mathtt{00 \cdot 001 \cdot 000011 \cdot 000010000111} \cdots R_p 
\end{eqnarray*}
with $|T| = n = \Theta(2^p)$.
It immediately follows from the definition of $T$ that
the self-referencing LZ77 factorization of $T$ is
\begin{eqnarray*}
  \LZorigsr(T) & = & \mathtt{0|0001}|R_3|R_4| \cdots |R_p| \\
  & = & \mathtt{0|0001|000011|000010000111|} \cdots |R_p|
\end{eqnarray*}
with $\zsrorig(T) = p = \Theta(\log n)$.
Note that the second factor $\mathtt{0001}$ is self-referencing.

As for substitution, we consider the string 
\begin{eqnarray*}
  T' & = & T[1]\cdot \mathtt{2} \cdot T[3..n] \\
  & = & \mathtt{02 \cdot 001 \cdot 000011 \cdot 000010000111} \cdots R_p 
\end{eqnarray*}
which can be obtained from $T$ by substituting the second $\mathtt{0}$ with $\mathtt{2}$.
Let us analyze the structure of the self-referencing LZ77 factorization of $T'$.
The second factor $\mathtt{0001}$ in $\LZorigsr(T)$ becomes $\mathtt{2001}$
in the edited string $T'$,
and this is divided into exactly three factors as $\mathtt{2|00|1|}$ in $\LZorigsr(T')$
because $\mathtt{2}$ is a fresh character,
$\mathtt{00}$ is the shortest prefix of
$T[3..n] = \mathtt{001}R_3 \cdots R_p$ that does not occur in $T[1..2] = \mathtt{02}$,
and $\mathtt{1}$ is a fresh character.
Our claim is that each $R_k$ with $3 \leq k \leq p$
is halved into two factors
$R_k[1..|R_k|/2] = R_1 \cdots R_{k-2} \mathtt{0}$ and
$R_k[|R_k|/2+1..|R_k|] = \mathtt{0} R_2 \cdots R_{k-2} \mathtt{11}$
of equal length in $\LZorigsr(T')$.
Suppose that $R_{k-1}$ is factorized as
$R_{k-1}[1..|R_{k-1}|/2]~|~R_{k-1}[|R_{k-1}|/2+1..|R_{k-1}|]~|$ in $\LZorigsr(T')$,
which means that the next factor is a prefix of $R_k \cdots R_p$.
Since $R_k[1..2] = \mathtt{00}$ and $T'[1..2] = \mathtt{02}$,
$R_k[1..|R_k|-1]$ does not have a previous occurrence as a prefix of $T'$.
Since $R_k = R_1 \cdots R_{k-2}R_{k-1} \mathtt{1}$
and $R_1 \cdots R_{k-2} = R_{k-1}[1..|R_{k-1}|-1]$,
the longest prefix of $T'[|R_1 \cdots R_{k-1}|+1..n] = R_k \cdots R_p$
that has a previous occurrence beginning in range $[1..|R_1 \cdots R_{k-1}|]$ is $R_1 \cdots R_{k-2}$, which implies $R_1 \cdots R_{k-2} R_{k-1}[1] = R_1 \cdots R_{k-2}\mathtt{0}$ is the next factor in $\LZorigsr(T')$.
The remaining part of $R_k$ is
$R_{k-1}[2..|R_{k-1}|]\mathtt{1} = \mathtt{0} R_2 \cdots R_{k-2} \mathtt{11}$.
Since its prefix $\mathtt{0} R_2 \cdots R_{k-2} \mathtt{1}$ has a previous occurrence
and $\mathtt{0} R_2 \cdots R_{k-2} \mathtt{11}$ has a suffix $\mathtt{01}^{k-1}$
which is the leftmost occurrence of $\mathtt{1}^{k-1}$ in $T'$,
this remaining part $\mathtt{0} R_2 \cdots R_{k-2} \mathtt{11}$ becomes the next factor in $\LZorigsr(T')$.
Thus, the self-referencing LZ77 factorization of $T'$ is
\[
\LZorigsr(T') = \mathtt{0|2|00|1| 000|011| 000010|000111| \cdots }|R_1 \cdots R_{p-2}\mathtt{0|0}R_2 \cdots R_{p-2}\mathtt{11} |
\]
with $\zsrorig(T') = 2p$, which leads to
$\MSensSub(\zsrorig,n) \geq 2p/p = 2$
and $\ASensSub(\zsrorig,n) \geq 2p-p = p = \zsrorig = \Omega(\log n)$.

\textbf{insertions:}
We use the same string $T$ in the case of substitutions.
Let $T'$ be the string obtained by inserting
$\mathtt{2}$ immediately after $T[1] = \mathtt{0}$,
namely,
\begin{eqnarray*}
  T' & = &  \mathtt{0} \cdot \mathtt{2} \cdot \mathtt{0} \cdot R_2 \cdots R_p \\
  & = & \mathtt{0\cdot2 \cdot0 \cdot 001 \cdot 000011 \cdot 000010000111} \cdots R_p.
\end{eqnarray*}
Then, by similar arguments to the case of substitutions,
we have
\[
\LZorigsr(T') = \mathtt{0|2|00|01| 0000|11| 000010|000111| \cdots }|R_1 \cdots R_{p-2}\mathtt{0|0}R_2 \cdots R_{p-2}\mathtt{11} |
\]
with $\zsrorig(T') = 2p$, which leads to
$\MSensIns(\zsrorig,n) \geq 2p/p = 2$
and $\ASensIns(\zsrorig,n) \geq 2p-p = p = \zsrorig = \Omega(\log n)$.

\textbf{deletions:}
As for deletions, we use the same strings $T$ and $T'$ from Theorem~\ref{theo:lz77_lowerbounds}. 
This string and the deletion also achieve the same lower bound for the self-referencing LZ77 factorization in the case of deletions. 
Then, we obtain $\zsrorig(T) = p$, $\zsrorig(T') = 2p-2$, which leads to
$\liminf_{n \rightarrow \infty}\MSensDel(\zsrorig,n) \geq 2$
and $\ASensDel(\zsrorig,n) \geq \zsrorig -2 = \Omega(\log n)$.

\end{proof}

The strings $T$ and $T'$ used in Theorem~\ref{theo:LZsr_lowerbounds}
give us optimal additive lower bounds in terms $\zsrorig$,
are highly compressible ($\zsrorig(T) = O(\log n)$)
and only use two or three distinct characters.
By using more characters,
we can obtain larger lower bounds for the additive sensitivity
for the size of the self-referencing LZ77 factorizations
in terms of the string length $n$, as follows:

\begin{theorem} \label{theo:lz77sr_lowerbounds_sqrtn}
  The following lower bounds on
  the sensitivity of self-referencing LZ77 factorization $\LZorigsr$ hold: \\
  \textbf{substitutions:} $\ASensSub(\zsrorig,n) = \Omega(\sqrt{n})$. \\
  \textbf{insertions:} $\ASensIns(\zsrorig,n) = \Omega(\sqrt{n})$. \\
  \textbf{deletions:} $\ASensDel(\zsrorig,n) = \Omega(\sqrt{n})$.
\end{theorem}

\begin{proof}
  In~\ref{app:lz77sr_lowerbounds_sqrtn}.
\end{proof}

\subsection{Upper bounds for the sensitivity of $\zsrorig$}

\begin{theorem} \label{theo:u_zsrorig}
  The following upper bounds on
  the sensitivity of self-referencing LZ77 factorization $\LZorigsr$ hold:  \\
  \textbf{substitutions:} $\MSensSub(\zsrorig,n) \leq 2$.
    $\ASensSub(\zsrorig,n) \leq \zsrorig$. \\
  \textbf{insertions:} $\MSensIns(\zsrorig,n) \leq 2$.
    $\ASensIns(\zsrorig,n) \leq \zsrorig$. \\
  \textbf{deletions:} $\MSensDel(\zsrorig,n) \leq 2$.
    $\ASensDel(\zsrorig,n) \leq \zsrorig$.
\end{theorem}

\begin{proof}
  We use the same notations as in Theorem~\ref{theo:u_zorig} of Section~\ref{sec:u_zorig}.
We consider the case where $T[i]$ is substituted by a fresh character $\#$,
  as in the proof for Theorem~\ref{theo:u_zorig}.
  We prove the following claim:
  \begin{claim}
    Each interval $[p_j,q_j]$ has at most two starting positions $p'_k$ and $p'_{k+1}$ of factors in $\LZorigsr(T')$ for $1 \leq k < z'$, excluding the interval $[p_I, q_I]$ that contains the edited position $i$.
    The interval $[p_I,q_I]$ has at most three starting positions of factors in $\LZorigsr(T')$.
  \end{claim}
  
  \begin{claimproof}
    Cases (1) and (3) which correspond to the positions before and after $i$ can be shown by the same discussions in the case of non self-referencing LZ factorizations (Theorem~\ref{theo:u_zorig} in Section~\ref{sec:u_zorig}).
    Now we consider case (2):
    \begin{itemize}
      \item[(2)] The interval $[p_j, q_j]$ satisfies $p_j \leq i \leq q_j$ (namely, $f_j=f_I$):
      If $f_I$ is not self-referencing, then by the same argument to the proof for Theorem~\ref{theo:u_zorig} in Section~\ref{sec:u_zorig}, the interval has at most two starting positions of factors in $\LZorigsr(T')$.
      Now we consider the case that $f_I$ is self-referencing.
      For the string $w_{1} a w_{2} c = T[p_I..q_I]$, only the substrings of $w_{2}$ can have a self-referencing previous occurrence that contains the edited position $i$ in $T$.
      Therefore, $w_1$ has a previous occurrence in $T'$ not containing $i$,
      which means that $q'_k = i$ where $T'[i] = \#$ is a fresh character.
      For the $w_{2} c$ part, we can apply the same discussion of Case (3)
      in Theorem~\ref{theo:u_zorig} of Section~\ref{sec:u_zorig}.
      Therefore, the $w_{1}\#$ part of $T'[p_j, q_j] = w_1 \# w_2 c$ can have at most one starting position, and the $w_{2} c$ part can have at most two starting positions of a factor in $\LZorigsr(T')$. 
    \end{itemize}
  This completes the proof for the claim. 
  \end{claimproof}
  
  By the above claim, $\zsrorig(T') \leq 2 \zsrorig(T)+1$ holds for any string $T$ and any substitution.
  Since again $|f_1| = 1$, we get $\zsrorig(T') \leq 2 \zsrorig(T)$.

  Using the same character(s) as in the proof for Theorem~\ref{theo:u_zorig},
  we can generalize this proof to
  the other types of edit operations.
\end{proof}
 \section{Lempel-Ziv-Storer-Szymanski factorizations with/without\\ self-references}
\label{sec:lzss}

In this section, we consider the worst-case sensitivity
of the \emph{Lempel-Ziv-Storer-Szymanski factorizations} (\emph{LZSS})~\cite{StorerS82} with/without self-references,
a.k.a. C-factorizations~\cite{Crochemore84}.

Given a factorization $T = f_{1} \cdots f_{z}$ for a string $T$ of length $n$:
\begin{itemize}
\item
  it is the non self-referencing LZSS factorization $\LZSS(T)$ of $T$
  if for each $1 \leq i \leq z$ the factor $f_i$ is either
  the first occurrence of a character in $T$,
  or the longest prefix of $f_i \cdots f_z$ occurs in $f_1 \cdots f_{i-1}$.

\item
  it is the self-referencing LZSS factorization $\LZSSsr(T)$ of $T$
  if for each $1 \leq i \leq z$ the factor $f_i$ is either
  the first occurrence of a character in $T$,
  or the longest prefix of $f_i \cdots f_z$ occurs
  at least twice in $f_1 \cdots f_{i}$.
\end{itemize}

To avoid confusions, we use different notations to denote the sizes
of these factorizations.
For a string $T$
let $\zss(T)$ and $\zsrss(T)$ denote
the number $z$ of factors in the non self-referencing LZSS factorization
and in the self-referencing LZSS factorization
of $T$, respectively.

For example, for string $T = \mathtt{abaabababababab}$, we have
\begin{eqnarray*}
  \LZSS(T) & = & \mathtt{a|b|a|aba|ba|baba|bab|}, \\
  \LZSSsr(T) & = & \mathtt{a|b|a|aba|babababab|},
\end{eqnarray*}
where $|$ denotes the right-end of each factor in the factorizations.
Here we have $\zss(T) = 7$ and $\zsrss(T) = 5$.

\subsection{Lower bounds for the sensitivity of $\zss$}

\begin{theorem} \label{theo:lzss_lowerbounds_sqrtn}
  The following lower bounds on
  the sensitivity of non self-referencing LZSS factorization $\LZSS$ hold: \\
  \textbf{substitutions:} $\liminf_{n \rightarrow \infty}\MSensSub(\zss,n) \geq 3$.
  $\ASensSub(\zss,n) \geq 2\zss - \Theta(\sqrt{\zss})$ and\\ $\ASensSub(\zss,n) = \Omega(\sqrt{n})$. \\
  \textbf{insertions:} $\liminf_{n \rightarrow \infty}\MSensIns(\zss,n) \geq 2$.
  $\ASensIns(\zss,n) \geq \zss - \Theta(\sqrt{\zss})$ and\\ $\ASensIns(\zss,n) = \Omega(\sqrt{n})$. \\
  \textbf{deletions:} $\liminf_{n \rightarrow \infty}\MSensDel(\zss,n) \geq 3$.
  $\ASensDel(\zss,n) \geq 2\zss - \Theta(\sqrt{\zss})$ and\\ $\ASensDel(\zss,n) = \Omega(\sqrt{n})$.
\end{theorem}

\begin{proof}
  Let $\Sigma = \{\mathtt{0, 1, a_1, \ldots, a_{\mathit{p}}, b_1, \ldots , b_{\mathit{p}}} \}$.
  Let
  \begin{eqnarray*}
  Q_1 &=& \mathtt{(a_1\cdots a_{\mathit{p}}) (a_1\cdots a_{{\mathit{p}}-1}) \cdots (a_1a_2)(a_1)}, \\
  Q_2 &=&  \mathtt{(b_1)  (b_1 b_2)  \cdots (b_1 \cdots b_{\mathit{p}-1}) (b_1 \cdots b_{\mathit{p}})},
  \end{eqnarray*}
  and $m = |Q_1| = |Q_2| = p(p+1)/2$.
  Consider the following string:
    \begin{equation*}
      \begin{split}
        T =& (Q_1 \mathtt{a_11} Q_2) \cdot (\mathtt{a_11} Q_2[1]) \cdot (Q_1[m]\mathtt{a_11}Q_2[1..2]) \cdot (Q_1[m-1..m]\mathtt{a_11}Q_2[1..3])\\ & \cdots (Q_1[m-k+2..m] \mathtt{a_11}Q_2[1..k]) \cdots (Q_1[2..m] \mathtt{a_11}Q_2)\\
        =& (Q_1 \mathtt{a_11} Q_2) (\varepsilon \mathtt{a_11b_1}) (\mathtt{a_1a_11b_1b_1}) \cdots (Q_1[m-k+2..m] \mathtt{a_11}Q_2[1..k]) \cdots (Q_1[2..m] \mathtt{a_11}Q_2)
     \end{split}
    \end{equation*}
  with $1 \leq k \leq m$.

Let us analyze the structure of the non self-referencing LZSS factorization of $T$. $Q_1$ consists of $p$ characters $\mathtt{a_1, \ldots,a_\mathit{p}}$, and the prefix $\mathtt{a_1 \cdots a_\mathit{p}}$ of $Q_1$ forms $p$ factors of length 1.
  The remaining part of $Q_1$ is divided into $p-1$ factors as $\mathtt{(a_1 \cdots a_\mathit{k})}$ with $p-1 \geq k \geq 1$ because $\mathtt{(a_1 \cdots a_\mathit{k})a_1}$ does not occur before.
  Next, both $T[m+1]=\mathtt{a_1}$ and $T[m+2]=\mathtt{1}$ become a factor of length $1$. As for the prefix of $Q_2$, $\mathtt{b_1}$ is a fresh character and becomes a factor of length $1$. For each $\mathtt{(b_1 \cdots b_\mathit{k})}$ with $2\leq k \leq p$, $\mathtt{b_1 \cdots b_{\mathit{k}-1}}$ occurs previously, and
  $\mathtt{b_1 \cdots b_\mathit{k}}$ does not occur before. Therefore, each interval of $\mathtt{(b_1 \cdots b_\mathit{k})}$ has two factors as $\mathtt{b_1 \cdots b_{\mathit{k}-1}|b_\mathit{k}}|$. Then, there are $4p$ factors in the interval $[1..|Q_1\mathtt{a_1}\mathtt{1}Q_2|]$.
  The substring $T[|Q_1\mathtt{a_1}\mathtt{1}Q_2|..|T|]$ is the sequence of $m$ parts $(Q_1[m-k+2..m] \mathtt{a_11}Q_2[1..k])$ with $1 \leq k \leq m$. Each part becomes a factor because $(Q_1[m-k+2..m] \mathtt{a_11}Q_2[1..k])$ occurs at $T[m-k+2 .. m+k+2]$, and $(Q_1[m-k+2..m] \mathtt{a_11}Q_2[1..k])Q_1[m-k+1]$
  does not occur before.
  Therefore, the factorization of $T$ is:
  \begin{eqnarray*}
  \LZSS(T) = && Q_1 \mathtt{|a_1|1|} Q_2|(\varepsilon \mathtt{a_11b_1})|(\mathtt{a_1a_11b_1b_1})|\cdots\\&&|(Q_1[m-k+2..m] \mathtt{a_11}Q_2[1..k])| \cdots |(Q_1[2..m] \mathtt{a_11}Q_2)|,
  \end{eqnarray*}
  where
  \[
  \LZSS(Q_1) = \mathtt{a_1|\cdots |a_{\mathit{p}}| (a_1\cdots a_{{\mathit{p}}-1})| \cdots |(a_1a_2)|(a_1)|}
  \]
  and
  \[
  \LZSS(Q_2) =  \mathtt{b_1| b_1| b_2|  \cdots |b_1 \cdots b_{\mathit{p}-2}|b_{\mathit{p}-1}|b_1 \cdots b_{\mathit{p}-1}|b_{\mathit{p}}|}.
  \]
  Then $\zss(T) = 4p+(1/2)p(p+1)$ holds.
  
  \textbf{substitutions:}
  Let 
  \[
  T' = (Q_1 \mathtt{a_10} Q_2) (\varepsilon \mathtt{a_11b_1}) (\mathtt{a_1a_11b_1b_1}) \cdots (Q_1[m-k+2..m] \mathtt{a_11}Q_2[1..k]) \cdots (Q_1[2..m] \mathtt{a_11}Q_2)
  \]
  be the string obtained from $T$ by substituting the first $\mathtt{1}$ with $\mathtt{0}$.
  It is clear that the factorization of the interval $[1..|Q_1\mathtt{a_1}\mathtt{0}Q_2|]$ is unchanged, and there are $4p$ factors in.
  Next, $m$ factors $(Q_1[m-k+2..m] \mathtt{a_11}Q_2[1..k])$ with $1 \leq k \leq m$ lose the position they refer to.
  Then, each factor $Q_1[m-k+2..m] \mathtt{a_1}\mathtt{1}Q_2[1..k]$ is divided into three factors as $Q_1[m-k+2..m] \mathtt{a_1} | \mathtt{1} Q_2[1..k-1] | Q_2[k]|$ because of their previous occurrences.
  Therefore, the factorization of $T'$ is:
    \begin{equation*}
      \begin{split}
        \LZSS(T') = & Q_1 \mathtt{|a_1|0|} Q_2|\mathtt{a_1|1|b_1}|\mathtt{a_1a_1|1b_1|b_1}|\cdots |Q_1[m-k+2..m] \mathtt{a_1|1}Q_2[1..k-1]|Q_2[k]| \cdots\\ &|Q_1[2..m] \mathtt{a_1|1}Q_2[1..m-1]|Q_2[m]|,
     \end{split}
    \end{equation*}
    where
    \[
    \LZSS(Q_1) = \mathtt{a_1|\cdots |a_{\mathit{p}}| (a_1\cdots a_{{\mathit{p}}-1})| \cdots |(a_1a_2)|(a_1)|}
    \]
    and
    \[
    \LZSS(Q_2) =  \mathtt{b_1| b_1| b_2|  \cdots |b_1 \cdots b_{\mathit{p}-2}|b_{\mathit{p}-1}|b_1 \cdots b_{\mathit{p}-1}|b_{\mathit{p}}|}.
    \]
  Then, $\zss(T') = 4p+(3/2)p(p+1)$ holds.
  Also, 
  \begin{eqnarray*}
  |T| &=& p(p+1) + 2 +\sum_{k=1}^{(1/2)p(p+1)} (2k+1)   \\ 
  &=& p(p+1) + 2 +2 \sum_{k=1}^{\frac{p(p+1)}{2}}k + \frac{p(p+1)}{2}  \\ \nonumber
  &=& p(p+1) + 2 +\frac{p^2(p+1)^2}{4} + p(p+1)  = \Theta(p^4)
  \end{eqnarray*}
  holds.
  Hence, we obtain
  \begin{eqnarray*}
  \liminf_{n \rightarrow \infty} \MSensSub(\zss,n) & \geq & \liminf_{p \rightarrow \infty} \left( \frac{\left(4p+\frac{3(p(p+1))}{2}\right)}{\left(4p+\frac{p(p+1)}{2}\right)} \right) = 3, \\
  \ASensSub(\zss,n) & \geq & \left(4p+\frac{3(p(p+1))}{2}\right) - \left(4p+\frac{p(p+1)}{2}\right) =  p(p+1) \\
  &=& 2\zss - \Theta(\sqrt{\zss}) \in \Omega(\sqrt{n}).
  \end{eqnarray*}

  \textbf{insertions:}
  Let 
  \[
  T' = (Q_1 \mathtt{a_101} Q_2) (\varepsilon \mathtt{a_11b_1}) (\mathtt{a_1a_11b_1b_1}) \cdots (Q_1[m-k+2..m] \mathtt{a_11}Q_2[1..k]) \cdots (Q_1[2..m] \mathtt{a_11}Q_2)
  \]
  be the string obtained from $T$ by inserting $\mathtt{0}$ before the first $\mathtt{1}$.
  The non self-referencing LZSS factorization of $T'$ is:
  \begin{equation*}
    \begin{split}
      \LZSS(T') = \ & Q_1 \mathtt{|a_1|0|1|} Q_2|\mathtt{a_1|1b_1}|\mathtt{a_1a_1|1b_1b_1}|\cdots \\
      & |Q_1[m-k+2..m] \mathtt{a_1|1}Q_2[1..k]| \cdots |Q_1[2..m] \mathtt{a_1|1}Q_2|,
   \end{split}
  \end{equation*}
  where
  \[
  \LZSS(Q_1) = \mathtt{a_1|\cdots |a_{\mathit{p}}| (a_1\cdots a_{{\mathit{p}}-1})| \cdots |(a_1a_2)|(a_1)|}
  \]
  and
  \[
  \LZSS(Q_2) =  \mathtt{b_1| b_1| b_2|  \cdots |b_1 \cdots b_{\mathit{p}-2}|b_{\mathit{p}-1}|b_1 \cdots b_{\mathit{p}-1}|b_{\mathit{p}}|}.
  \]
  Then, $\zss(T') = 4p+p(p+1)$ holds.
  Hence, we obtain $\liminf_{n \rightarrow \infty}\MSensIns(\zss,n) \geq 2$, $\ASensIns(\zss,n) \geq \zss - \Theta(\sqrt{\zss})$, and $\ASensIns(\zss,n) = \Omega(\sqrt{n})$.
  
  \textbf{deletions:}
  As for deletions, by considering $T'$ obtained from $T$ by deleting the first $\mathtt{1}$, we get a similar decomposition to the case of substitutions.
  Thus, we also obtain $\liminf_{n \rightarrow \infty}\MSensDel(\zss,n) \geq 3$, $\ASensDel(\zss,n) \geq 2\zss - \Theta(\sqrt{\zss})$, and $\ASensDel(\zss,n) = \Omega(\sqrt{n})$.
\end{proof}

\subsection{Upper bounds for the sensitivity of $\zss$} \label{sec:u_zss}

\begin{theorem} \label{theo:u_zss}
  The following upper bounds on
  the sensitivity of non self-referencing LZSS factorization $\LZSS$ hold:  \\
  \textbf{substitutions:} $\limsup_{n \rightarrow \infty}\MSensSub(\zss,n) \leq 3$.
    $\ASensSub(\zss,n) \leq 2\zss-2$. \\
  \textbf{insertions:} $\MSensIns(\zss,n) \leq 2$.
    $\ASensIns(\zss,n) \leq \zss$. \\
  \textbf{deletions:} $\limsup_{n \rightarrow \infty}\MSensDel(\zss,n) \leq 3$.
    $\ASensDel(\zss,n) \leq 2\zss-3$.
\end{theorem}

\begin{proof}
  Let $\LZSS(T) = f_{1} \cdots f_{z}$ and $\LZSS(T') = f'_{1} \cdots f'_{z'}$.
  We denote the interval of the $j$th factor $f_j$ (resp. $f'_j$) by $[p_j,q_j]$ (resp. $[p'_j, q'_j]$), namely $T[p_j..q_j] = f_j$ and $T'[p'_j..q'_j] = f'_{j}$.
  Also, let $f_I$ be the factor of $\LZSS(T)$ whose interval $[p_I, q_I]$ contains the edited position $i$, namely $p_I \leq i \leq q_I$.
  
  \textbf{substitutions}:
    In the following, we consider the case that the $i$th character $T[i]=a$ is substituted by a fresh character $\#$ which does not occur in $T$. The other cases can be proven similarly.
    Now we show the following claim:
    \begin{claim} 
      After the substitution, each interval $[p_j,q_j]$ has at most three starting positions $p'_k$, $p'_{k+1}$, and $p'_{k+2}$ of factors in $\LZSS(T')$ for $1 \leq k \leq z'-2$.
    \end{claim}
  
    \begin{claimproof}
      There are the three following cases:
      \begin{itemize}
      \item[(i)] When the interval $[p_j, q_j]$ satisfies $q_j < i$:
        By the same argument to Case (1) for LZ77,
        the interval $[p_j, q_j]$ contains exactly one starting position $p'_j = p_j$.

        \item[(ii)] When the interval $[p_j, q_j]$ satisfies $p_j \leq i \leq q_j$ (namely, $f_j=f_I$):
        For the string $w_{j_1} a w_{j_2} = T[p_j..q_j]$, it is guaranteed that $w_1 a w_2$ has at least one occurrence in $f_1 \cdots f_{j-1}$.
        After the substitution which gives $T'[p_j..q_j] = w_1 \# w_2$, $w_1$ and $\#$ become factors as $f'_j$ and $f'_{j+1}$, and $w_2$ becomes the prefix of factor $f'_{j+2}$.
        This means that $p'_j=p_j$ and $q'_{j+2}\geq q_j$.
        Therefore, the interval $[p_j, q_j]$ contains at most three starting positions $p'_j$, $p'_{j+1}$ and $p'_{j+2}$ of factors in $\LZSS(T')$.
      
        \item[(iii)] When the interval $[p_j, q_j]$ satisfies $i < p_j$:
          We consider the two following sub-cases:
          \begin{itemize}
            \item[(iii-A)] When $T[p_j..q_j]$ has at least one occurrence which does not contain the edited position $i$ in $T$:
            Any suffix of $T[p_j..q_j]$ still has a previous occurrence in $T'$.
            Therefore, $[p'_k, q'_k]$ with $p_j \leq p'_k$ satisfies $q'_k \geq q_j$,
            meaning the interval $[p_j, q_j]$ contains at most one starting position $p'_k$ of a factor in $\LZSS(T')$.
          
          \item[(iii-B)] All occurrences of $T[p_j..q_j]$ in $T$ contain the edited position $i$:
            Let $u_{1} a u_{2} = T[p_j, q_j]$ with $a \in \Sigma$ and $u_1, u_2 \in \Sigma^*$.
            $u_1$ and $u_2$ have previous occurrences in $T'[1..p_j-1]$.
            Let $p'_k$ be the starting position of the leftmost factor of $\LZSS(T')$
            which begins in range $[p_j, q_j]$. 
            If $p'_k$ is in $u_2$, then $q'_k \geq q_k$
            and thus there is only one starting position of a factor of $\LZSS(T')$ in the interval $[p_j..q_j]$.
            Suppose $p'_k$ is in $u_1$.
            If $a$ has no previous occurrences
            (which happens when $T[i]$ was the only previous occurrence of $a$),
            then $T'[p_k+|u_1|]$ is the first occurrence of $a$ in $T'$
            and thus $q'_k = p_k+|u_1|-1$, $p'_{k+1} = q'_k + 1$ and $q'_{k+1} = p'_{k+1}+1$.
            Otherwise, $q'_k \geq p_k+|u_1|-1$,
            $p'_{k+1} \geq q'_k + 1$ and $q'_{k+1} \geq p'_{k+1}+1$.
            In either case, since $u_2$ has a previous occurrence, $q'_{k+2} \geq q_{k+1}$.
            Thus, there can exist at most three starting positions of factors of $\LZSS(T')$ in the interval $[p_j..q_j]$.

          \end{itemize}
      \end{itemize}
      This completes the proof for the claim.
      \end{claimproof}
  
      It follows from the above claim that $\zss(T') \leq 3 \zss(T)$ for any string $T$ and substitutions with $\#$.
      Since $|f_1| = 1$, $\zss(T') \leq 3 \zss(T)-2$ holds.
      Hence, we obtain $\limsup_{n \rightarrow \infty}\MSensSub(\zss,n) \leq 3$ and $\ASensSub(\zss,n) \leq 2\zss-2$.
      
      \textbf{insertions}:
        In the following, we consider the case that $\#$ is inserted to between positions $i-1$ and $i$.
        The other cases can be proven similarly.
        Now we show the following claim:
        \begin{claim}
          After the insertion, each interval $[p_j,q_j]$ contains at most two starting positions $p'_k$ and $p'_{k+1}$ of factors in $\LZSS(T')$ for $1 \leq k \leq z'-1$, excluding the interval $[p_I, q_I]$.
          Also, the interval $[p_I,q_I]$ contains at most three starting positions of factors in $\LZSS(T')$.
        \end{claim}
      
        \begin{claimproof}
          For Cases (i), (ii), and (iii-A), we can use the same discussions as in the case of substitutions.
          Now we consider Case (iii-B):
          \begin{itemize}
            \item[(iii-B)] When all occurrences of $T[p_j..q_j]$ in $T$ contain the edited position $i$:
            Let $u_{1} a u_{2} = T[p_j, q_j]$ with $a \in \Sigma$ and $w_1, w_2 \in \Sigma^*$. It is guaranteed that $w_{j_1}a$, and $w_{j_2}$ still have previous occurrences in $T'$.
            Therefore, each range of $w_{j_1}a$ and $w_{j_2}$ can contain at most one starting position of a factor in $\LZSS(T')$. 
          \end{itemize}    
        \end{claimproof}
      
        It follows from the above claim that $\zss(T') \leq 2 \zss(T)+1$ holds any string $T$ and insertions with $\#$.
        By using the same discussion as for $f_1$, we obtain $\zss(T') \leq 2 \zss(T)$ holds.
        Then we have $\MSensIns(\zss,n) \leq 2$ and $\ASensIns(\zss,n) \leq \zss$.
        
        \textbf{deletions}:
          In the following, we consider the case that $T[i]=a$ is deleted.
          Now we show the following claim:
          \begin{claim}
            After the deletion, each interval $[p_j,q_j]$ contains at most three starting positions $p'_k$, $p'_{k+1}$, and $p'_{k+2}$ of factors in $\LZSS(T')$ for $1 \leq k \leq z'-2$, excluding the interval $[p_I, q_I]$.
            The interval $[p_I,q_I]$ contains at most two starting positions of factors in $\LZSS(T')$.
          \end{claim}
        
          \begin{claimproof}
            For Cases (i) and (iii), we can use the same discussions as in the case of substitutions.
            Now we consider case (ii):
            \begin{itemize}
              \item[(ii)] When the interval $[p_j, q_j]$ satisfies $p_j \leq i \leq q_j$ (namely, $f_j=f_I$):
              Let $w_{1} a w_{2} = T[p_j..q_j]$ with $a \in \Sigma$ and $w_1, w_2 \in \Sigma^*$. It is guaranteed that $w_1 a w_2$ has at least one previous occurrence in $f_1 \cdots f_{j-1}$.
              Therefore, after the deletion of $a$, each range of $w_{1}$ and $w_{2}$ can contain at most one starting position of a factor in $\LZSS(T')$. 
            \end{itemize}    
          \end{claimproof}
        
          It follows from the above claim that $\zss(T') \leq 3 \zss(T)-1$ holds for any string $T$ and deletions.
          By using the same discussion as for $f_1$, $\zss(T') \leq 3 \zss(T)-3$ holds.
          Then we get\\ $\limsup_{n \rightarrow \infty}\MSensDel(\zss,n) \leq 3$ and $\ASensDel(\zss,n) \leq 2\zss-3$.
\end{proof}

\subsection{Lower bound for the sensitivity of $\zsrss$}

\begin{theorem} \label{theo:lzsssr_lowerbounds_sqrtn}
  The following lower bounds on
  the sensitivity of self-referencing LZSS factorization $\LZSSsr$ hold: \\
  \textbf{substitutions:} $\liminf_{n \rightarrow \infty}\MSensSub(\zsrss,n) \geq 3$.
  $\ASensSub(\zsrss,n) \geq 2\zsrss - \Theta(\sqrt{\zsrss})$ and\\ $\ASensSub(\zsrss,n) = \Omega(\sqrt{n})$. \\
  \textbf{insertions:} $\liminf_{n \rightarrow \infty}\MSensIns(\zsrss,n) \geq 2$.
  $\ASensIns(\zsrss,n) \geq \zsrss - \Theta(\sqrt{\zsrss})$ and\\ $\ASensIns(\zsrss,n) = \Omega(\sqrt{n})$. \\
  \textbf{deletions:} $\liminf_{n \rightarrow \infty}\MSensDel(\zsrss,n) \geq 3$.
  $\ASensDel(\zsrss,n) \geq 2\zsrss - \Theta(\sqrt{\zsrss})$ and\\ $\ASensDel(\zsrss,n) = \Omega(\sqrt{n})$.
\end{theorem}
\begin{proof}
  We use the same strings $T$ and $T'$ as in the proof for Theorem~\ref{theo:lzss_lowerbounds_sqrtn}
  which shows the lower bounds of the sensitivity of the non self-referencing LZSS.
  For the string $T$ and each edit operation,
  the self-referencing LZSS factorization is the same as the non self-referencing LZSS factorization.
  Hence, we obtain Theorem~\ref{theo:lzsssr_lowerbounds_sqrtn}.
\end{proof}

\subsection{Upper bounds for the sensitivity of $\zsrss$}

\begin{theorem} \label{theo:u_zsrss}
  The following upper bounds on
  the sensitivity of self-referencing LZSS factorization $\LZSSsr$ hold:  \\
  \textbf{substitutions:} $\MSensSub(\zsrss,n) \leq 3$.
    $\ASensSub(\zsrss,n) \leq 2\zsrss$. \\
  \textbf{insertions:} $\limsup_{n \rightarrow \infty}\MSensIns(\zsrss,n) \leq 2$.
    $\ASensIns(\zsrss,n) \leq \zsrss +1$. \\
  \textbf{deletions:} $\MSensDel(\zsrss,n) \leq 3$.
    $\ASensDel(\zsrss,n) \leq 2\zsrss -1$.
\end{theorem}

\begin{proof}
  We use the same notations as in Theorem~\ref{theo:u_zss} of Section~\ref{sec:u_zss}.
  As with the case of self-referencing LZ77 $\zsrorig$,
  only the interval $[p_I,q_I]$ that contains the edited position $i$ is effected
  in this case of self-referencing LZSS $\zsrss$.
  For the string $w_{1} a w_{2} = T[p_I..q_I]$,
  only $w_{2}$ can have a self-referencing previous occurrence
  that contains the edited position $i$.
  For each edit operation, by applying the discussion of Case (iii)
  in Theorem~\ref{theo:u_zss}
  to the range of $w_{2}$ in $[p_I..q_I]$, we obtain Theorem~\ref{theo:u_zsrss}.
\end{proof}

 \section{LZ-End factorizations}
\label{sec:lzend}

In this section, we consider the worst-case sensitivity
of the \emph{LZ-End factorizations}~\cite{KreftN13}.
This is an LZ77-like compressor
  such that each factor $f_i$ has a previous occurrence
  which corresponds to the ending position of a previous factor.
  This property allows for fast substring extraction in practice~\cite{KreftN13}.

A factorization $T = f_{1} \cdots f_{\zend}$ for a string $T$ of length $n$ is the LZ-End factorization $\LZEnd(T)$ of $T$
such that, for each $1 \leq i < \zend$, $f_i[1..|f_i|-1]$ is the longest prefix of $f_i \cdots f_{\zend}$ which has a previous occurrence in $f_1 \cdots f_{i-1}$ as a suffix of some string in $\{\varepsilon, f_1, f_1f_2, \ldots, f_1 \cdots f_{i-1}\}$.
The last factor $f_{\zend}$ is the suffix of $T$ of
length $n-|f_1 \cdots f_{\zend-1}|$.
Again, if we use a common convention
that the string $T$ terminates with a unique character $\$$,
then the last factor $f_{\zend}$ satisfies the same properties as $f_1, \ldots, f_{z-1}$,
in the cases of LZ-End factorizations.
Let $\zend(T)$ denote the number of factors in
the LZ-End factorization of string $T$.

For example, for string $T = \mathtt{abaabababababab\$}$,
\begin{eqnarray*}
T & = & \mathtt{a|b|aa|ba|bab|ababab\$}|,
\end{eqnarray*}
where $|$ denotes the right-end of each factor in the factorization.
Here we have $\zend(T) = 6$.

\subsection{Lower bounds for the sensitivity of $\zend$} \label{sec:l_zend}

\begin{theorem} \label{theo:lzend_lowerbounds}
  The following lower bounds on
  the sensitivity of $\zend$ hold: \\
  \textbf{substitutions:} $\liminf_{n \rightarrow \infty}\MSensSub(\zend,n) \geq 2$. 
    $\ASensSub(\zend,n) \geq \zend - \Theta(\sqrt{\zend})$ and \\ $\ASensSub(\zend,n) = \Omega(\sqrt{n})$. \\
  \textbf{insertions:} $\liminf_{n \rightarrow \infty}\MSensIns(\zend,n) \geq 2$.
    $\ASensIns(\zend,n) \geq \zend - \Theta(\sqrt{\zend})$ and \\ $\ASensIns(\zend,n) = \Omega(\sqrt{n})$. \\
  \textbf{deletions:} $\liminf_{n \rightarrow \infty}\MSensDel(\zend,n) \geq 2$.
    $\ASensDel(\zend,n) \geq \zend - \Theta(\sqrt{\zend})$ and \\ $\ASensDel(\zend,n) = \Omega(\sqrt{n})$.
\end{theorem}
\begin{proof}

  Let $\sigma_j$ denote the $i$th character in the alphabet $\Sigma$
  for $1 \leq j \leq |\Sigma|$.
  For a positive integer $p$,
  consider the string $Q = \sigma_1 \cdot \sigma_1\sigma_2 \cdots \sigma_1 \cdots \sigma_{\mathit{p}}$ of length $q = |Q| = p(p+1)/2 = \Theta(p^2)$.
Consider the string
\begin{eqnarray*}
  T & = & Q \cdot \sigma_1 \sigma_{\mathit{p}+1} \cdot Q[q]\sigma_1 \sigma_{\mathit{p}+1}\sigma_{\mathit{p}+2} \cdot Q[q-1..q]\sigma_1\sigma_{\mathit{p}+1}\sigma_{\mathit{p}+3} \cdots Q\sigma_1 \sigma_{\mathit{p}+1}\sigma_{\mathit{p+q}+1}
\end{eqnarray*}
with $|T| = \Theta(q^2)$.
As for the interval $[1,q]$ in $T$, $\LZEnd(Q)= f_1,\ldots,f_p$ such that $f_k = \sigma_1 \cdots \sigma_{\mathit{k}}$ for every $1\leq k \leq p$.
Since $f_p$ has no occurrences in $f_1 \cdots f_{p-1}$, the decomposition is not changed by appending any character to $Q$.
Hence, the next factor $f_{p+1}$ starts at position $q+1$.
Then $f_{p+1} = \sigma_1 \sigma_{\mathit{p}+1}$ holds, since $\sigma_{\mathit{p}+1}$ is a fresh character and $\sigma_1 = f_1$.
As for the remaining interval, we show $f_{p+1+j} = Q[q-j+1 .. q]\sigma_1\sigma_{\mathit{p}+1}\sigma_{\mathit{p+j}+1}$ holds for each $1\leq j \leq q$.
At first, for $j=1$, $f_{p+2}=Q[q]\sigma_1\sigma_{\mathit{p}+1}\sigma_{\mathit{p}+2}$ holds since $f_{p+2}$ starts with $Q[q]$, $Q[q]\sigma_1\sigma_{\mathit{p}+1}$ has an occurrence as a suffix of $T[1..q+2]$, and $\sigma_{\mathit{p}+2}$ is a fresh character in the prefix.
Next, we assume that $f_{p+1+j} = Q[q-j+1 .. q]\sigma_1\sigma_{\mathit{p}+1}\sigma_{\mathit{p+j}+1}$ holds with $1\leq j \leq k-1$ for some integer $k$.
Then we consider whether $f_{p+1+k} = Q[q-k+1 .. q]\sigma_1\sigma_{\mathit{p}+1}\sigma_{\mathit{p+k}+1}$ holds or not.
By the assumption, $f_{p+1+k}$ starts with $Q[q-k+1]$.
Also, $Q[q-k+1..q]\sigma_1\sigma_{\mathit{p}+1}$ has an occurrence as a suffix of $T[1..q+2]$, and $\sigma_{\mathit{p+k}+1}$ is a fresh character in the prefix.
Therefore, the assumption is also valid for $k$.
By the above argument, $f_{p+1+j} = Q[q-j+1 .. q]\sigma_1\sigma_{\mathit{p}+1}\sigma_{\mathit{p+j}+1}$ holds for each $1\leq j \leq q$ by induction.
Therefore,
\begin{eqnarray*}
  \LZEnd(T) = \sigma_1|\sigma_1\sigma_2| \cdots |\sigma_1 \cdots \sigma_{\mathit{p}}| \sigma_1 \sigma_{\mathit{p}+1} | Q[q] \sigma_1 \sigma_{\mathit{p}+1}\sigma_{\mathit{p}+2} |\cdots| Q\sigma_1 \sigma_{\mathit{p}+1}\sigma_{\mathit{p+q}+1}|
\end{eqnarray*}
with $\zend(T) = p+1+q = \Theta(q) = \Theta(\sqrt{n})$.

As for substitutions, consider the string
\begin{eqnarray*}
  T' & = & Q \cdot \# \sigma_{\mathit{p}+1} \cdot Q[q]\sigma_1 \sigma_{\mathit{p}+1}\sigma_{\mathit{p}+2} \cdot Q[q-1..q]\sigma_1\sigma_{\mathit{p}+1}\sigma_{\mathit{p}+3} \cdots Q\sigma_1 \sigma_{\mathit{p}+1}\sigma_{\mathit{p+q}+1}
\end{eqnarray*}
which can be obtained from $T$ by substituting $T[q+1]=\sigma_1$ with a character $\#$ which does not occur in $T$.
Let us analyze the structure of the $\LZEnd(T')$.
As mentioned above, $f_1,\ldots,f_p$ are not changed after the substitution.
The $(p+1)$th factor in $\LZEnd(T)$, namely, $\# \sigma_{\mathit{p}+1}$ is factorized as $\# | \sigma_{\mathit{p}+1}|$ in $\LZEnd(T')$
since both characters have no occurrence in $T'[1..q] = Q$.
Then the next factor starts with $Q[q]$.
Each of $Q[q]$ and $\mathtt{\sigma}_1$ have some occurrence as a suffix of a previous factor.
On the other hand, each of $Q[q]\mathtt{\sigma}_1$ and $\sigma_{\mathit{p}+1}\sigma_{\mathit{p}+2}$ have no occurrences previously.
Therefore, $Q[q]\sigma_1\sigma_{\mathit{p}+1}\sigma_{\mathit{p}+2}$ is factorized as $Q[q]\sigma_1 | \sigma_{\mathit{p}+1}\sigma_{\mathit{p}+2}|$ in $\LZEnd(T')$.
Similarly, by induction, each $(p+1+k)$th factor in $\LZEnd(T)$ for every $2 \leq k \leq q$, namely, $Q[q-k+1..q]\sigma_1\sigma_{\mathit{p}+1}\sigma_{\mathit{p+k}+1}$ is also factorized as $Q[q-k+1..q]\sigma_1 | \sigma_{\mathit{p}+1}\sigma_{\mathit{p+k}+1}|$.
Thus, the LZ-End factorization of $T'$ is
\[
 \LZEnd(T') = \sigma_1|\sigma_1 \sigma_2|\cdots |\sigma_1 \cdots \sigma_{\mathit{p}}| \# |\sigma_{\mathit{p}+1} | Q[q] \sigma_1 | \sigma_{\mathit{p}+1}\sigma_{\mathit{p}+2} |\cdots| Q\sigma_1 |\sigma_{\mathit{p}+1}\sigma_{\mathit{p+q}+1}|
\]
with $\zend(T') = p+2+2q$.
Recall $p = \Theta(\sqrt{q})$.
Hence we get $\liminf_{n \rightarrow \infty}\MSensSub(\zend,n) \geq \liminf_{q \to \infty} (p+2q+2)/(p+q+1) \geq 2$, $\ASensSub(\zend,n) \geq (p+2q+2)-(p+q+1) = \zend - \Theta(\sqrt{\zend})$, and $\ASensSub(\zend,n) = \Omega(\sqrt{n})$.

Also, as for deletions (resp. insertions), we get Theorem~\ref{theo:lzend_lowerbounds} by considering the case where the character $T[q+1]$ is deleted (resp. $\#$ is inserted between positions $q$ and $q+1$).
\end{proof}

\subsection{Upper bounds for the sensitivity of $\zend$} \label{sec:u_zend}

To show a non-trivial upper bound for the sensitivity of $\zend$,
we use the following known results:

\begin{theorem}[\cite{KreftN13}] \label{theo:zend_z}
  For any string $T$, $\zsrss(T) \leq \zend(T)$.
\end{theorem}

\begin{theorem}[Theorem 3.2 of \cite{KempaS22}] \label{theo:zend_delta}
  For any string $T$ of length $n$, $\zend(T) = O(\delta(T) \log^2 (n/\delta(T)))$.
\end{theorem}

From Theorems~\ref{theo:z_gamma},~\ref{theo:delta_gamma},~\ref{theo:zend_delta}, and~\ref{theo:zend_z}, we obtain the following result:
\begin{corollary} \label{coro:lzend_upperbounds}
  The following upper bounds on
  the sensitivity of $\zend$ hold: \\
  \textbf{substitutions:} $\MSensSub(\zend,n) = O(\log^2 (n/\delta))$. 
    $\ASensSub(\zend,n) = O(\zend \log^2 (n/\delta))$. \\
  \textbf{insertions:} $\MSensIns(\zend,n) = O(\log^2 (n/\delta))$.
    $\ASensIns(\zend,n) = O(\zend \log^2 (n/\delta))$. \\
  \textbf{deletions:} $\MSensDel(\zend,n) = O(\log^2 (n/\delta))$.
    $\ASensDel(\zend,n) = O(\zend \log^2 (n/\delta))$.
\end{corollary}
\begin{proof}

  For any string $T$, $\delta(T) \leq \zend(T)$ holds from Theorems~\ref{theo:delta_gamma},~\ref{theo:z_gamma}, and~\ref{theo:zend_z}.
  Let $T'$ be any string with $\ed(T,T') = 1$.
  It follows from Theorem~\ref{theo:upper_bound_delta} that
  $\delta(T') \leq 2 \delta(T)$.
  Now let $c$ be the constant value such that $\delta(T') = c \delta(T)$ holds.
  Then, $\log^2 (n/\delta(T')) = \log^2 n + \log^2 c \delta(T) - 2 \log n \log c \delta(T) = \log^2 n + \log^2 \delta(T) - 2 \log n \log \delta(T) + \log^2 c + 2 \log \delta(T) \log c -2 \log n \log c = O(\log^2 (n/\delta(T)))$.
  Following Lemma~\ref{lem:squeeze}, 
  we now obtain $\zend(T') = O(\delta(T') \log^2 (n/\delta(T'))) = O(\delta(T) \log^2 (n/\delta(T))) = O(\zend(T) \log^2 (n/\delta(T)))$, 
  which leads to the claimed upper bounds for the sensitivity for $\zend$.
\end{proof}

 \section{Lempel-Ziv 78 factorizations}
\label{sec:lz78}

In this section, we consider the worst-case sensitivity
of the \emph{Lempel-Ziv 78 factorizations} (\emph{LZ78})~\cite{LZ78}.

For convenience, let $f_0 = \varepsilon$.
A factorization $T = f_{1} \cdots f_{\zseveneight}$ for a string $T$ of length $n$ is the LZ78 factorization $\LZseveneight(T)$ of $T$
if for each $1 \leq i < \zseveneight$ the factor $f_i$ is
the longest prefix of $f_i \cdots f_{\zseveneight}$
such that $f_i[1..|f_i|-1] = f_j$ for some $0 \leq j < i$.
The last factor $f_{\zseveneight}$ is the suffix of $T$ of
length $n-|f_1 \cdots f_{\zseveneight-1}|$ and it may be equal to
some previous factor $f_j$~($1 \leq j < \zseveneight$).
Again, if we use a common convention
that the string $T$ terminates with a unique character $\$$,
then the last factor $f_{\zseveneight}$ can be defined
analogously to the previous factors.
Let $\zseveneight(T)$ denote the number of factors in
the LZ78 factorization of string $T$.

For example, for string $T = \mathtt{abaabababababab\$}$,
\[
\LZseveneight(T) = \mathtt{a|b|aa|ba|bab|ab|aba|b\$}|,
\]
where $|$ denotes the right-end of each factor in the factorization.
Here we have $\zseveneight(T) = 8$.

As for the sensitivity of LZ78, Lagarde and Perifel~\cite{LagardeP18} showed that $\MSensIns(\zseveneight,n) = \Omega(n^{1/4})$, $\ASensIns(\zseveneight,n) = \Omega(\zseveneight^{3/2})$, and $\ASensIns(\zseveneight,n) = \Omega(n/\log n)$ for insertions.
\footnote{In the restricted case of appending a character to the top of a string or deleting the first character of a string, they showed upper bounds that the ratio is $O(n^{1/4})$ and the increase is $O(\zseveneight^{3/2})$.}
In this section, we present lower bounds for the multiplicative/additive sensitivity of LZ78 for the remaining cases, i.e., for substitutions and deletions, by using a completely different string from~\cite{LagardeP18}.

\subsection{Lower bounds for the sensitivity of $\zseveneight$}

\begin{theorem} \label{theo:LZ78_lowerbound}
  The following lower bounds on
  the sensitivity of $\zseveneight$ hold: \\
  \textbf{substitutions:} $\MSensSub(\zseveneight,n) = \Omega(n^{1/4})$.
    $\ASensSub(\zseveneight,n) = \Omega(\zseveneight^{3/2})$ and $\ASensSub(\zseveneight,n) = \Omega(n^{3/4})$. \\
  \textbf{deletions:} $\MSensDel(\zseveneight,n) = \Omega(n^{1/4})$.
    $\ASensDel(\zseveneight,n) = \Omega(\zseveneight^{3/2})$ and $\ASensDel(\zseveneight,n) = \Omega(n^{3/4})$.
\end{theorem}

\begin{proof}
  Consider the string 
  \begin{equation*}
    T = \mathtt{(\sigma_{\mathit{k}+1}) \cdots (\sigma_{2\mathit{k}}) \cdot (\sigma_{1})\cdot (\sigma_{1}\sigma_{2})\cdots(\sigma_{1}\cdots \sigma_{\mathit{k}})\cdot (\sigma_1\cdots \sigma_{\mathit{y}_1} \cdot \sigma_{\mathit{k}+1}) \cdots (\sigma_1\cdots \sigma_{\mathit{y}_\mathit{k}} \cdot \sigma_{2\mathit{k}})},
  \end{equation*}
  where $\mathtt{\sigma_\mathit{i}}$ for every $1 \leq i \leq 2k$ is a distinct character and $y_j$ for every $1 \leq j \leq k$ satisfies the following property:
  $y_j$ is the maximum integer at most $k$ such that
  $2+j+\ell_j-1 \equiv y_j \pmod{\ell_{j}}$ where $\ell_j$ is an integer satisfying $(1/2)\ell_j(\ell_j-1)+1\leq j \leq (1/2)\ell_j(\ell_j+1)$.
  We remark that the parentheses $($ and $)$ in $T$ are shown only for
  the better visualization and exposition, and therefore they are not the characters in $T$.
  
  Let $n$ be the length of $T$.
  Since $k+ (1/2)k(k+1) < n < k+ (1/2)k(k+1) + k(k+1)$, $n \in \Theta(k^2)$ holds.
  In the LZ78 factorization of $T$,
  for each substring $(w)$,
  its suffix $w[1..|w|-1]$ has a previous occurrence as $(w[1..|w|-1])$,
  and $(w)$ is the leftmost occurrence of $w$ in the string $T$.
  Therefore, the LZ78 factorization of $T$ is
  \begin{equation*}
    \LZseveneight(T) = \mathtt{\sigma_{\mathit{k}+1}| \cdots |\sigma_{2\mathit{k}}|\sigma_{1}|\sigma_{1}\sigma_{2}|\cdots|\sigma_{1}\cdots \sigma_{\mathit{k}}|\sigma_1\cdots \sigma_{\mathit{y}_1} \cdot \sigma_{\mathit{k}+1}|\cdots |\sigma_1\cdots \sigma_{\mathit{y}_\mathit{k}} \cdot \sigma_{2\mathit{k}}|}
  \end{equation*}
  with $\zseveneight(T)=3k$.
  
  For our analysis of the sensitivity of $\zseveneight$ for substitutions,
  consider the string
  \begin{equation*}
    T' = \mathtt{(\sigma_{\mathit{k}+1}) \cdots (\sigma_{2\mathit{k}}) \cdot (\sigma_{1})\cdot (\sigma_{1}\sigma_{2})\cdots(\sigma_{1}\cdots \sigma_{\mathit{k}})\cdot (\#\sigma_2\cdots \sigma_{\mathit{y}_1} \cdot \sigma_{\mathit{k}+1}) \cdots (\sigma_1\cdots \sigma_{\mathit{y}_\mathit{k}} \cdot \sigma_{2\mathit{k}})},
  \end{equation*}
   which can be obtained from $T$ by substituting the first character $\sigma_1$ of the string in the $2k+1$th paring parentheses with a fresh character $\mathtt{\#}$, which does not occur in $T$.
Let us analyze the structure of the LZ78 factorization of $T'$.
   Clearly, the first $2k$ factors are unchanged after the substitution.
   Next, we consider $(\#\sigma_2\cdots \sigma_{\mathit{y}_1} \cdot \sigma_{\mathit{k}+1})$.
   First, the prefix $\mathtt{\#\sigma_2\cdots \sigma_{\mathit{y}_1}}$ is decomposed into $y_1$ factors of length $1$.
   The next factor is $\mathtt{\sigma_{\mathit{k}+1}\sigma_1}$ since $\mathtt{\sigma_{\mathit{k}+1}}$ has an occurrence as a previous factor and $\mathtt{\sigma_{\mathit{k}+1}\sigma_1}$ has no occurrences as a previous factor.
   Now we show in each interval of the $2k+j$th paring parentheses for $2\leq j\leq k$ (i.e., the interval of $\mathtt{\sigma_1 \cdots \sigma_{\mathit{y_j}} \sigma_{\mathit{k+j}}}$) there appear the right-ends $|$ of factors in $\LZseveneight(T')$ as follows:
   \begin{equation}
     \mathtt{\sigma_1|\sigma_2\cdots \sigma_{\mathit{j}+1}| \sigma_{\mathit{j}+2} \cdots \sigma_{\mathit{j+\ell_j}+1}|\cdots|\sigma_{\mathit{y_j-\ell_j}+1} \cdots\sigma_{\mathit{y_j}}|\sigma_{\mathit{k+j}}}. \label{eqn:lz78}
   \end{equation}
   Namely, the interval is decomposed into $3+((y_j-j-\ell_j-1)/\ell_j)+1=3+(y_j-j-1)/\ell_j$ pieces $d_1, \ldots, d_{3+(y_j-j-1)/\ell_j}$, where $|d_1|=|d_{3+(y_j-j-1)/\ell_j}| = 1$, $|d_2|=j$, and each of the others is of length $\ell_j$.
   At first, we show the partition (\ref{eqn:lz78}) is valid for $j=2$.
   As mentioned above, there is a factor $\mathtt{\sigma_{\mathit{k}+1}\sigma_1}$ constructed with the immediately preceded character and the first character of the interval of $\mathtt{\sigma_1 \cdots \sigma_{\mathit{y_2}} \sigma_{\mathit{k}+2}}$. 
   And then, $\mathtt{\sigma_2 \cdots \sigma_{\mathit{y_2}}}$ is decomposed into $\mathtt{\sigma_2\sigma_3|\sigma_4 \sigma_5|\cdots |\sigma_{\mathit{y_2}-1} \sigma_{\mathit{y_2}}|}$ since $y_2$ is the maximum odd value less than or equal to $k$ and each $\sigma_{2i}$ for $1\leq i \leq (y_2-1)/2$ is the longest prefix as some previous factor.
   Since $\ell_2 = 2$ holds,
   the partitions of the interval are $\mathtt{\sigma_1|\sigma_2\sigma_3|\sigma_4 \sigma_5|\cdots |\sigma_{\mathit{y}_2-1} \sigma_{\mathit{y}_2}|\sigma_{\mathit{k}+2}}$, and this satisfies the partition (\ref{eqn:lz78}).
   Next, we assume the partition (\ref{eqn:lz78}) is valid for $j\leq h-1$ for some integer $h$,
   and we consider whether the partition (\ref{eqn:lz78}) is valid or not for $j=h$.
   From the assumption and the same discussion as the above,
   there is a factor $\mathtt{\sigma_{\mathit{k+h}-1}\sigma_1}$ constructed with the immediately preceded character and the first character of the interval of $\mathtt{\sigma_1 \cdots \sigma_{\mathit{y_h}} \sigma_{\mathit{k+h}}}$. 
Since the set of previous factors starting with the character $\mathtt{\sigma_2}$ is $\{\mathtt{(\sigma_2), (\sigma_2 \sigma_3), \ldots, (\sigma_2 \cdots \sigma_\mathit{h})}\}$, the next factor becomes $\mathtt{\sigma_2 \cdots \sigma_{\mathit{h}+1}}$.
   In addition, it is guaranteed that the set of previous factors starting with the character $\mathtt{\sigma_\mathit{i}}$ for every $h+2 \leq i \leq y_h-\ell_{h}+1$ is equal to $\{\mathtt{(\sigma_{\mathit{i}}), (\sigma_{\mathit{i}} \sigma_{\mathit{i}+1}), \ldots, (\sigma_{\mathit{i}} \cdots \sigma_\mathit{i+\ell_j}-\mathrm{2}}\}$.
   Since $y_h$ can be described as $h+\ell_h+1+t\ell_h$ for some integer $t$, the decomposition of the interval becomes $\mathtt{\sigma_1|\sigma_2\cdots \sigma_{\mathit{h}+1}| \sigma_{\mathit{h}+2} \cdots \sigma_{\mathit{h+\ell_h}+1}|\cdots|\sigma_{\mathit{y_h-\ell_h}+1} \cdots\sigma_{\mathit{y_h}}|\sigma_{\mathit{k+h}}}$,
   and this satisfies the partition (\ref{eqn:lz78}).
   From the above, the partition (\ref{eqn:lz78}) is valid for $2\leq j \leq k$ by induction.
   
   Thus, the LZ78 factorization of $T'$ is
   \begin{equation*}
     \begin{split}
       \LZseveneight(T') = \mathtt{\sigma_{\mathit{k}+1}| \cdots |\sigma_{2\mathit{k}}|\sigma_{1}|\sigma_{1}\sigma_{2}|\cdots|\sigma_{1}\cdots \sigma_{\mathit{k}}|\#|\sigma_2|\cdots|\sigma_{\mathit{y}_1}| \sigma_{\mathit{k}+1}\sigma_1|\sigma_2 \sigma_3|\cdots|\sigma_{\mathit{y}_2-1} \sigma_{\mathit{y}_2}|} \\
      \mathtt{\sigma_{\mathit{k}+2} \sigma_1|\cdots| \sigma_2\cdots \sigma_{\mathit{j}+1}| \sigma_{\mathit{j}+2} \cdots \sigma_{\mathit{j+\ell_j}+1}|\cdots|\sigma_{\mathit{y_j-\ell_j}+1} \cdots \sigma_{\mathit{y_j}}| \sigma_{\mathit{k+j}} \sigma_1|\cdots|\cdots \sigma_{\mathit{y}_\mathit{k}}|\sigma_{2\mathit{k}}|}.
    \end{split}
  \end{equation*}
  See also Figure~\ref{fig:lz78} for a concrete example.
  
  The size of $\LZseveneight(T')$ is $\zseveneight(T')=2k+y_1+\sum_{j=2}^{k}(2+(y_j-j-1)/\ell_j)+1=5k-1+\sum_{j=2}^{k}((y_j-j-1)/\ell_j)$.
  $\sum_{j=2}^{k}((y_j-j-1)/\ell_j)$ is the total number of factors of length $\ell_j$ for $1\leq j \leq k$.
  Now we consider the number of factors of length in $L \in \{\ell_1, \ldots, \ell_k\}$.
  For all $j$ such that $(1/2)L(L-1)+1\leq j \leq (1/2)L(L+1)$,
  the total number of factors of length in $L$ is $\sum_{j=j_{\min}^{L}}^{j_{\max}^{L}}((y_j-j-1)/L)$, where $j_{\min}^{L}=(1/2)L(L-1)+1$ and $j_{\max}^{L}=(1/2)L(L+1)$.
  From the definition of $y_j$,
  $\sum_j y_j = k+(k-1)+\cdots+(k-L+1)$ holds.
  Therefore, $\sum_{j=j_{\min}^{L}}^{j_{\max}^{L}}((y_j-j-1)/L)=(k+(k-1)+\cdots+(k-L+1)-j_{\min}^{L}-(j_{\min}^{L}+1) - \cdots -j_{\max}^{L} -L)/L=(L(k-j_{\max}^{L})-L)/L=k-(1/2)L(L+1)-1$.
  Let $\ell_k = m$.
  Then the total number of factors of length $\ell_j$ for $1\leq j \leq k$ is $\sum_{L=2}^{m} (k-(1/2)L(L+1)-1)=(m-1)k-(1/12)m(m+1)(2m+1)-(1/4)m(m+1)-m+2$.
  Consider the case of $m=\sqrt{k}$, then $\zseveneight(T') \in \Omega(k\sqrt{k})$.
  Thus we obtain $\MSensSub(\zseveneight,n) = \Omega(n^{1/4})$,
  $\ASensSub(\zseveneight,n) = \Omega(\zseveneight^{3/2})$, and $\ASensSub(\zseveneight,n) = \Omega(n^{3/4})$. 
  
  As for deletions, by considering $T'$ obtained from $T$ by deleting the first character of the $2k+1$th factor in $\LZseveneight(T)$, we obtain a similar decomposition as the above.
  Thus, $\MSensDel(\zseveneight,n) = \Omega(n^{1/4})$, $\ASensDel(\zseveneight,n) = \Omega(\zseveneight^{3/2})$, and $\ASensDel(\zseveneight,n) = \Omega(n^{3/4})$ also hold.
\end{proof}

  \begin{figure}[th]
    \centerline{
      \includegraphics[scale=0.6]{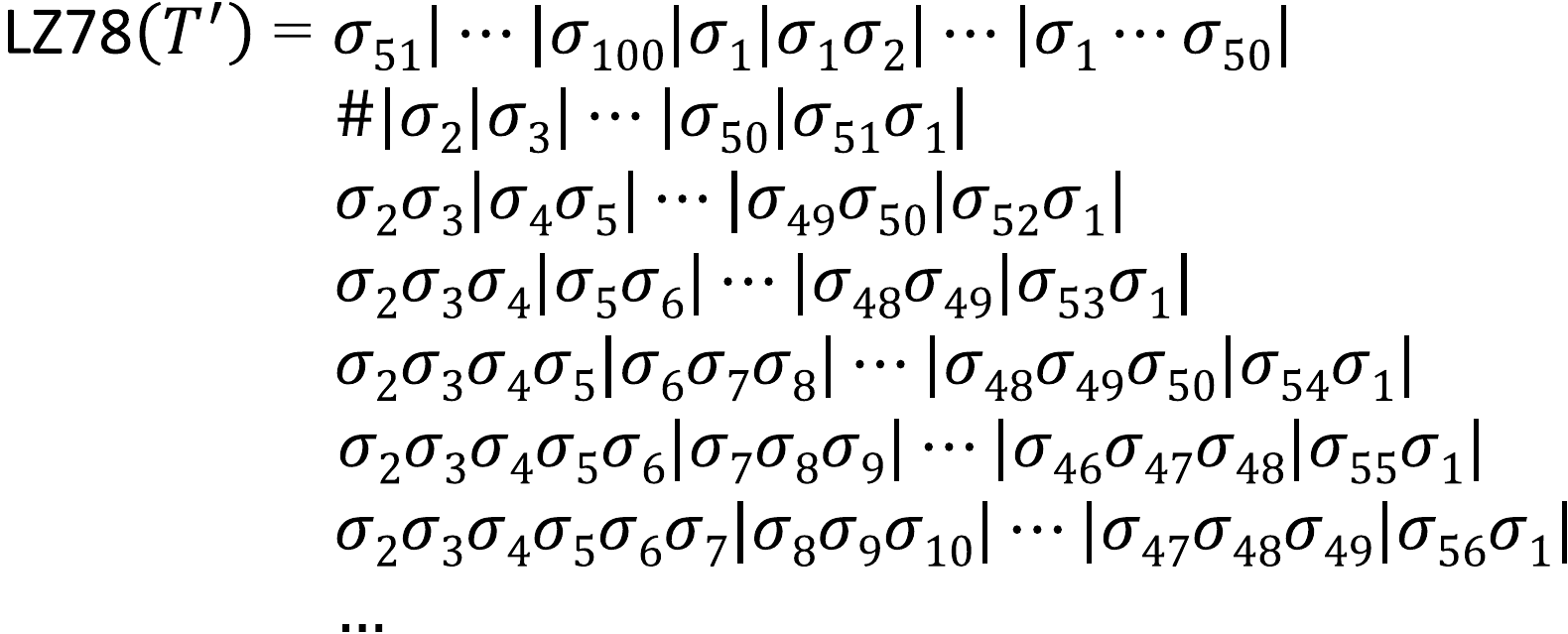}
    }
    \caption{Illustration for $\LZseveneight(T')$ for the string $T'$ of Theorem~\ref{theo:LZ78_lowerbound} with $k=50$.}
    \label{fig:lz78}
  \end{figure}
  
We remark that our string also achieves $\MSensIns(\zseveneight,n) = \Omega(n^{1/4})$, $\ASensIns(\zseveneight,n) = \Omega(\zseveneight^{3/2})$, and $\ASensIns(\zseveneight,n) = \Omega(n^{3/4})$ for insertions, if we consider the string $T'$ obtained from $T$ by inserting $\mathtt{\#}$ between the first and second characters of the $2k+1$th factor of $\LZseveneight(T)$.

In Section~\ref{sec:grammar}, we will present an $O((n / \log n)^{\frac{2}{3}})$ upper bound for the multiplicative sensitivity for LZ78.
 \section{Smallest grammars and approximation grammars} \label{sec:grammar}

In this section, we consider the sensitivity of the smallest grammar size $g^*$
and several grammars whose sizes satisfy some approximation ratios to $g^*$.

\subsection{Smallest grammar}

In this section (and also in the following sections),
we consider \emph{grammar-based} compressors for input string $T$.

It is known that the problem of computing
the size $g^*(T)$ of the smallest grammar only generating $T$ is NP-hard~\cite{StorerS82,CharikarLLPPSS05}.
It is also known that $\zss(T)$ is a lower bound
of the size of \emph{any} grammar generating $T$,
namely, $\zss(T) \leq g^*(T)$ holds for any string $T$~\cite{Rytter03,CharikarLLPPSS05}.

We have the following upper bounds for the sensitivity of $g^*(T)$: 
\begin{theorem} \label{theo:smallest_grammar_upper}
  The following upper bounds on the sensitivity of $g^*$ hold: \\
  \textbf{substitutions:} $\MSensSub(g^*,n) \leq 2$. $\ASensSub(g^*,n) \leq g^*$.\\
  \textbf{insertions:} $\MSensIns(g^*,n) \leq 2$. $\ASensIns(g^*,n) \leq g^*$.\\
  \textbf{deletions:} $\MSensDel(g^*,n) \leq 2$. $\ASensDel(g^*,n) \leq g^*$.
\end{theorem}
\begin{proof}
Let $T$ be any string of length $n$,
and let $\mathcal{G}^*(T)$ be a grammar of size $g^*(T)$ that only generates $T$.

We describe the case of substitutions.
Let $T'$ be the string that can be obtained by substituting a character $c$ for the $i$th character $T[i]$ of $T$, where $c \neq T[i]$.
Let $X$ be a non-terminal of $\mathcal{G}^*(T)$ in the path $P$ from the root to the leaf for the $i$th character in the derivation tree of $\mathcal{G}^*(T)$.
Let $X \rightarrow Y_1 \cdots Y_k$ be the production from $X$,
and let $Y_j$~($1 \leq j \leq k$) be the non-terminal that is the child of $X$ in the path $P$.
Then, we introduce a new non-terminal $X'$ and a new production $X' \rightarrow Y_1 \cdots Y_{j-1} Y'_{j} Y_{j+1} \cdots Y_{k}$,
where $Y'_{j}$ will be the new non-terminal at the next depth in the path $P$.
By applying this operation in a top-down manner on $P$, we can obtain a grammar $\mathcal{G}(T')$ of size $g(T') \leq 2g^*(T)$ that generates $T'$.
Since $g^*(T') \leq g(T')$, we have the claimed bounds.
The cases with insertions and deletions are analogous.
\end{proof}

\subsection{Practical grammars}

Since computing a smallest grammar of size $g^*(T)$ is NP-hard,
a number of practical grammar-based compressors have been proposed,
including RePair~\cite{LarssonM99}, Longest-Match~\cite{KiefferY00},
Greedy~\cite{ApostolicoL00}, Sequential~\cite{YangK00}\footnote{Sequential is an improved version of Sequitur~\cite{NevillManningW97}.},
and LZ78~\cite{LZ78}\footnote{The LZ78 factorization can also be seen as a grammar.}.
Charikar et al.~\cite{CharikarLLPPSS05} analyzed the approximation ratios
of these grammar compressors to the smallest grammar.
Let $\grepair$, $\glongest$, $\ggreedy$, $\gseq$, $\zseveneight$ denote the sizes of
the aforementioned compressors, respectively.
It is known that for any $g \in \{\grepair, \glongest, \ggreedy, \zseveneight \}$
$g(T) = O(g^*(T) (n/ \log n)^{\frac{2}{3}})$ holds,
and $\gseq = O(g^*(T) (n / \log n)^{\frac{3}{4}})$ holds~\cite{CharikarLLPPSS05}.
By combining these results with Lemma~\ref{lem:squeeze} and Theorem~\ref{theo:smallest_grammar_upper},
we obtain the following bounds:

\begin{corollary}
  The following upper bounds for the sensitivity of $g \in \{\grepair, \glongest, \ggreedy, \zseveneight \}$ hold:\\
  \textbf{substitutions:} $\MSensSub(g,n) = O((n/ \log n)^{\frac{2}{3}})$. $\ASensSub(g,n) = O(g^* \cdot (n/ \log n)^{\frac{2}{3}})$. \\
  \textbf{insertions:} $\MSensIns(g,n) = O((n/ \log n)^{\frac{2}{3}})$. $\ASensIns(g,n) = O(g^* \cdot (n/ \log n)^{\frac{2}{3}})$. \\
  \textbf{deletions:} $\MSensDel(g,n) = O((n/ \log n)^{\frac{2}{3}})$. $\ASensDel(g,n) = O(g^* \cdot (n/ \log n)^{\frac{2}{3}})$. 
\end{corollary}

\begin{corollary}
  The following upper bounds for the sensitivity of $\gseq$ hold:\\
  \textbf{substitutions:} $\MSensSub(\gseq,n) = O((n/ \log n)^{\frac{3}{4}})$. $\ASensSub(\gseq,n) = O(g^* \cdot (n/ \log n)^{\frac{3}{4}})$. \\
  \textbf{insertions:} $\MSensIns(\gseq,n) = O((n/ \log n)^{\frac{3}{4}})$. $\ASensIns(\gseq,n) = O(g^* \cdot (n/ \log n)^{\frac{3}{4}})$. \\
  \textbf{deletions:} $\MSensDel(\gseq,n) = O((n/ \log n)^{\frac{3}{4}})$. $\ASensDel(\gseq,n) = O(g^* \cdot (n/ \log n)^{\frac{3}{4}})$. 
\end{corollary}

\subsection{Approximation grammars}

There also exist (better) approximation algorithms in terms of the smallest grammar size $g^*$.

It is known that $\alpha$-balanced grammar compressor~\cite{CharikarLLPPSS05}, the AVL-grammar compressor~\cite{Rytter03},
and the really-simple grammar compressor~\cite{Jez16}
all achieve $O(\log (n/g^*))$-approximation ratios to $g^*$.
Let $\galpha$, $\gavl$, and $\gsimple$ denote the sizes of these compressors,
respectively.
Namely, for every $g \in \{\galpha, \gavl, \gsimple\}$,
$g = O(g^* \log (n / g^*))$ holds.
Since $\log (n/g^*)$ satisfies the conditions for the function $f(n,g^*)$ in Lemma~\ref{lem:squeeze},
and since $g^*$ satisfies the conditions Lemma~\ref{lem:squeeze} by Theorem~\ref{theo:smallest_grammar_upper}, we obtain the following:

\begin{corollary} \label{theo:approximation_grammars}
  The following upper bounds for the sensitivity of $g \in \{\galpha, \gavl, \gsimple\}$ hold:\\
  \textbf{substitutions:} $\MSensSub(g,n) = O(\log (n/ g^*))$. $\ASensSub(g,n) = O(g^* \log (n/g^*))$. \\
  \textbf{insertions:} $\MSensIns(g,n) = O(\log (n / g^*))$. $\ASensIns(g,n) = O(g^* \log (n / g^*))$. \\
  \textbf{deletions:} $\MSensDel(g,n) = O(\log (n / g^*))$. $\ASensDel(g,n) = O(g^* \log (n / g^*))$. 
\end{corollary}

\section{Grammar compression by induced sorting (GCIS)}
\label{sec:gcis}

In this section, we consider the worst-case sensitivity
of the \emph{grammar compression by induced sorting} (\emph{GCIS})~\cite{NunesLGAN18,NunesLGAN20}.
GCIS is based on the idea from the famous SAIS algorithm~\cite{NongZC11}
that builds the suffix array of an input string in linear time.
Recently, it is shown that GCIS has a locally consistent parsing property
similar to the ESP-index~\cite{MaruyamaNKS13} and the SE-index~\cite{NishimotoIIBT20},
and grammar-based indexing structures
based on GCIS have been proposed~\cite{AkagiKNIBT21,Diaz-DominguezN21}.

Let $T$ be the string of length $n$ over an integer alphabet $\Sigma = \{ \mathtt{1}, \ldots, \sigma \}$.
Let $\Pi = \{\sigma+1, \ldots, \sigma+|\Pi|\}$
be the set of non-terminal symbols.
For strings $x, y$ over $\Sigma$ or $\Pi$,
we write $x \prec y$ iff $x$ is lexicographically smaller than $y$. 

First we explain how the GCIS algorithm constructs its grammar from the input string.
For any text position $1 \leq i \leq |T|$,
position $i$ is of type L if $T[i..|T|]$ is lexicographically larger than $T[i+1...|T|]$, and it is of type S otherwise.
For any $2 < i < |T|$, we call position $i$ an \emph{LMS} (\emph{LeftMost S}) position if $i$ is of type S and $i-1$ is of type L.
For convenience, we append a special character $\$$ to $T$
which does not occur elsewhere in $T$,
and assume that positions $1$ and $|T\$|$ are LMS positions.

Let $i_1, \ldots, i_{z+1}$ be the sequence of the LMS positions in $T$ sorted in increasing order.
Let $D_j = T[i_j..i_{j+1}-1]$ for any $1 \leq j \leq z$.
When $z \geq 2$, then $T = D_1, \ldots, D_z$ is called the GCIS-parsing of $T$.

Next, we create new non-terminal symbols $R_1, \ldots, R_z$ such that $R_i = 1+\sigma + |\{ D_j : D_j \prec D_i :1 \leq j \leq z \}|$ for each $i$.
Intuitively, we pick the least unused character from $\Pi$
and assign it to $R_i$.
Then, $G_1 = R_1 \cdots R_z$ is called the GCIS-string of $T$.
Let $\mathcal{G}_1$ the set of all $z$ symbols in $G_1$, and $P_1= \{R_i \to D_i : 1 \leq i \leq z\}$ is the set of production rules. Let $\mathcal{D}_1 = \{ D_1, \ldots, D_z \}$ be the set of all distinct factors.
Let $G_0 = T$, then we define GCIS recursively, as follows:

\begin{definition}
\label{def_gcis}
For $k \geq 0 $, let the sequence $i_1, i_2, \ldots i_{z_k+1}$ be all LMS positions sorted in increasing order, and $D_j = G_k[i_j \ldots i_{j+1}-1]$ for any $1 \leq j \leq z_k$. $G_k = D_1, D_2, \ldots, D_{z_k}$ is the GCIS-parsing of $G_k$.
For all $i$ in $1 \leq i \leq z_k$, we define $R$ to satisfy : 
\[
R_i = |\{D_j : D_j \prec D_i :1 \leq j \leq z_k)\}|+\sum_{t=1}^{k-1} |P_t|+ \sigma + 1. \]
Then, $G_{k+1} = R_1 \ldots R_{z_k}$ is the GCIS-string of $G_k$. $\mathcal{G}_{k+1}$ is the set of non-terminals, $P_k= \{R_i \to D_i :1 \leq i \leq z_k\}$ is the set of production rules. $\mathcal{D}_k = \{ D_1, \ldots, D_{z_k} \}$ is the set of all distinct factors in the GCIS-parsing of $G_k$.
\end{definition}
Again, each $R_i$ is chosen to be the least unused character from $\Pi$.
$G_{k+1}$ is not defined if there are no LMS positions in $G_k[2..|G_k|]$. Then,
the GCIS grammar of $T$ is $(\Sigma, \bigcup_{t=1}^{k}{\mathcal{G}_t}, \bigcup_{t=1}^{k-1}{P_t},G_k)$.
$T$ is derived from the recursive application of the rules $\bigcup_{t=1}^{k-1}{P_t}$, which is the third argument, to the fourth argument $G_k$, which is the start string, until there are no non-terminal characters, which is in the second argument $\bigcup_{t=1}^{k}{\mathcal{G}_t} = \Pi$, in the string.  
Let $r = k$ be the height of GCIS, in other words how many times we applied this GCIS method recursively to $T$.
Let $\gis(T)$ be the size of GCIS grammar of $T$. Then, if $r=0$, $\gis(T) = |T|$, and if $r \geq 1$, $\gis(T) = \Vert \mathcal{D}_1 \Vert + \cdots + \Vert \mathcal{D}_r \Vert + G_r$, where $\Vert S \Vert$ for a set of strings denotes
the total length of the strings in $S$.

Figure~\ref{gcis_fig1} shows an example on how GCIS is constructed
from an input string.

\begin{figure}[tbh]
 \centering
 \includegraphics[keepaspectratio, scale=0.5]
      {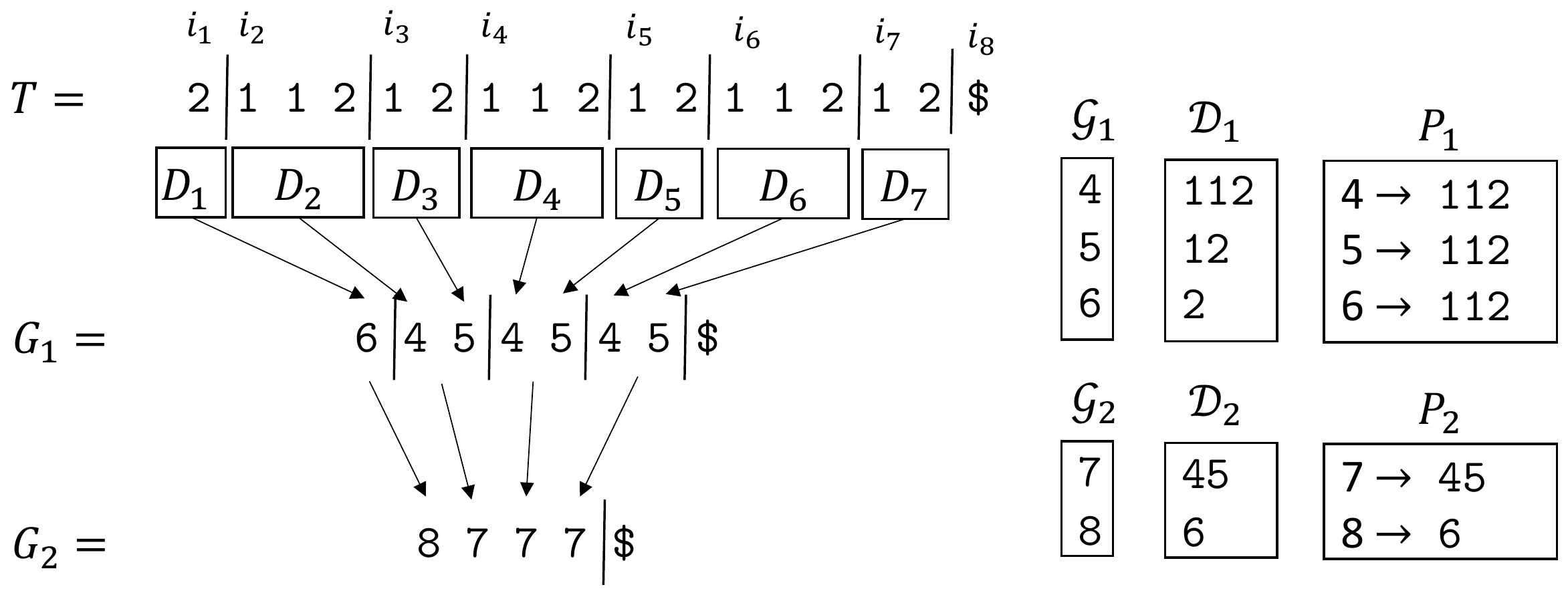}
 \caption{Construction of GCIS from string  $T = G_0 = \mathtt{2112121121211212\$}$. In this case, there are $8$ LMS positions $i_1, \ldots, i_8$ in $T$ and 7 factors $D_1, \ldots, D_7$. $\mathcal{D}_1 = \mathtt{\{ 112,12,2 \} }$ is the set of distinct factors of the GCIS-parsing for $T$, and $G_1 = R_1 \cdots R_7 = \mathtt{6454545} $ is the GCIS-string of $T$. Recursively, $\mathcal{D}_2 = \mathtt{\{ 6,45 \}}, G_2 = \mathtt{8777} $, and the start string of the GCIS for $T$ is $G_2$ because the number of factors of the GCIS-parsing of $G_2$ is 1 (excluding $\$$), in other words there are no LMS positions in $G_2[2.. |G_2|]$. The size of the GCIS grammar of $T$ is $\gis(T) = \Vert\mathcal{D}_1\Vert + \Vert\mathcal{D}_2\Vert + |G_2| = 6+3+4 = 13. $}
 \label{gcis_fig1}
\end{figure}

From now on, we consider to perform an edit operation
to the input string $T$
and will consider how the GCIS changes after the edit.

\begin{definition}
\label{def_gcis2}
Let $S$ and $S'$ be strings.
If $S'$ is obtained from $S$ by deleting the substring of length $a$ starting from a position $c$ in $S$ and by inserting a string of length $b$ to the same position $c$, then we write $F(S,S') = (a,b)$.
\end{definition}

Our single-character edit operation performed to $T$
can be described as $F(T,T') = (1,1)$ for substitution,
$F(T,T') = (0,1)$ for insertion,
and $F(T,T') = (1,0)$ for deletion.
We will use this notation $F$ to the GCIS-strings for $T$ and $T'$,
in which case $a,b$ can be larger than $1$.
Still, we will prove that $a,b$ are small constants
for the GCIS-strings.

As with the definitions for $T$,
$T'=D'_1, \ldots, D'_{z'}$ is the GCIS-parsing of $T'$,
$G'_1 = R'_1 \cdots R'_{z'}$ is the GCIS-string of $T'$,
$\mathcal{G'}_1$ is the set of non-terminals for $T'$,
$\mathcal{D}_1 = \{ D'_1, \ldots, D'_{z'} \}$
is the set of all distinct factors of the GCIS-parsing of $T'$,
$P'_1= \{R_i \to D_i :1 \leq i \leq z'\}$ is the set of production rules. 
Let $G'_0 = T'$, then we can recursively define $G'_1, G'_2 \ldots, G_{r'}$ similarly to $T$, where $r'$ is the height of the GCIS for $T'$.

\subsection{Upper bounds for the sensitivity of $\gis$}

This section presents the following upper bounds for
the sensitivity of GCIS.

\begin{theorem}
  The following upper bounds on
  the sensitivity of GCIS hold: \\
  \textbf{substitutions:} $\MSensSub(\gis,n) \leq 4$. $\ASensSub(\gis, n) \leq 3\gis$. \\
  \textbf{insertions:} $\MSensIns(\gis,n) \leq 4$. $\ASensIns(\gis, n) \leq 3\gis$. \\
  \textbf{deletions:} $\MSensDel(\gis,n) \leq 4$. $\ASensDel(\gis, n) \leq 3\gis$. 
\end{theorem}

We will prove this theorem as follows: We unify substitutions, insertions, and deletions by using the $F$ function in Definition~\ref{def_gcis2}. First, we prove that edit operations do not affect the size of the GCIS grammar. Second, we divide the size of GCIS grammar $\gis(T)$ into $\Vert\mathcal{D}_1\Vert$ and $\gis(G_1)$, and prove that $\Vert \mathcal{D'}_1 \Vert \leq 4 \Vert\mathcal{D}_1\Vert + O(1)$. Then, $\gis(T') = \Vert\mathcal{D'}_1\Vert + \gis(G'_1) \leq 4\Vert\mathcal{D}_1\Vert + \gis(G'_1) + O(1)$ holds. 
The essence is to find the two special strings $\hat{G_1}$ and $\hat{G'_1}$ which satisfy: 
\begin{itemize}
  \item $\hat{G'_1}$ can be obtained from $\hat{G_1}$ by some substitutions, insertions, and deletions.
  \item $\gis(G_1) = \gis(\hat{G_1})$, and $\gis(G'_1) = \gis(\hat{G'_1})$.
 \end{itemize}
 Then, we can apply the method to each height. The extra additive $O(1)$ factor can be charged to the process of the GCIS compression, which is to be proved in Lemma~\ref{gcis_under4plus}. Finally, we will obtain $\gis(T') \leq 4\gis(T)$.

\begin{lemma}[\cite{NongZC11}]
\label{sl_define}
The type of $T[k]$ is S if $T[k] \prec T[k+1]$ and L if $T[k] \succ T[k+1]$. If $T[k] = T[k+1]$, the type of $T[k]$ equals to the type of $T[k+1]$.
\end{lemma}

Let $\rank_T[i]$ be the lexicographical rank
of the character $T[i]$ at position $i$ in $T$.
Let $\hat{T}$ be any string of length $|\hat{T}| = |T|$
such that $\rank_{\hat{T}}[i] = \rank_T[i]$ for every $1 \leq i \leq |T|$.

\begin{lemma}
\label{gcis_identically_equal}
Let $G_1$ and $\hat{G_1}$ denote the GCIS-strings of $T$ and $\hat{T}$, respectively.
Then $\hat{G_1}$ is the string that can be obtained by replacing the characters in $G_1$ without changing the ranks of any characters in $G_1$, and $\gis(\hat{T}) = \gis(T)$.
\end{lemma}
\begin{proof}
  The lemma immediately follows from Lemma~\ref{sl_define}
  and that $\rank_{\hat{T}}[i] = \rank_{T}[i]$ for every $1 \leq i \leq |T|$.
\end{proof}
Figure~\ref{gcis_fig3} shows a concrete example
for Lemma~\ref{gcis_identically_equal}.

\begin{figure}[tbh]
 \centering
 \includegraphics[keepaspectratio, scale=0.5]
      {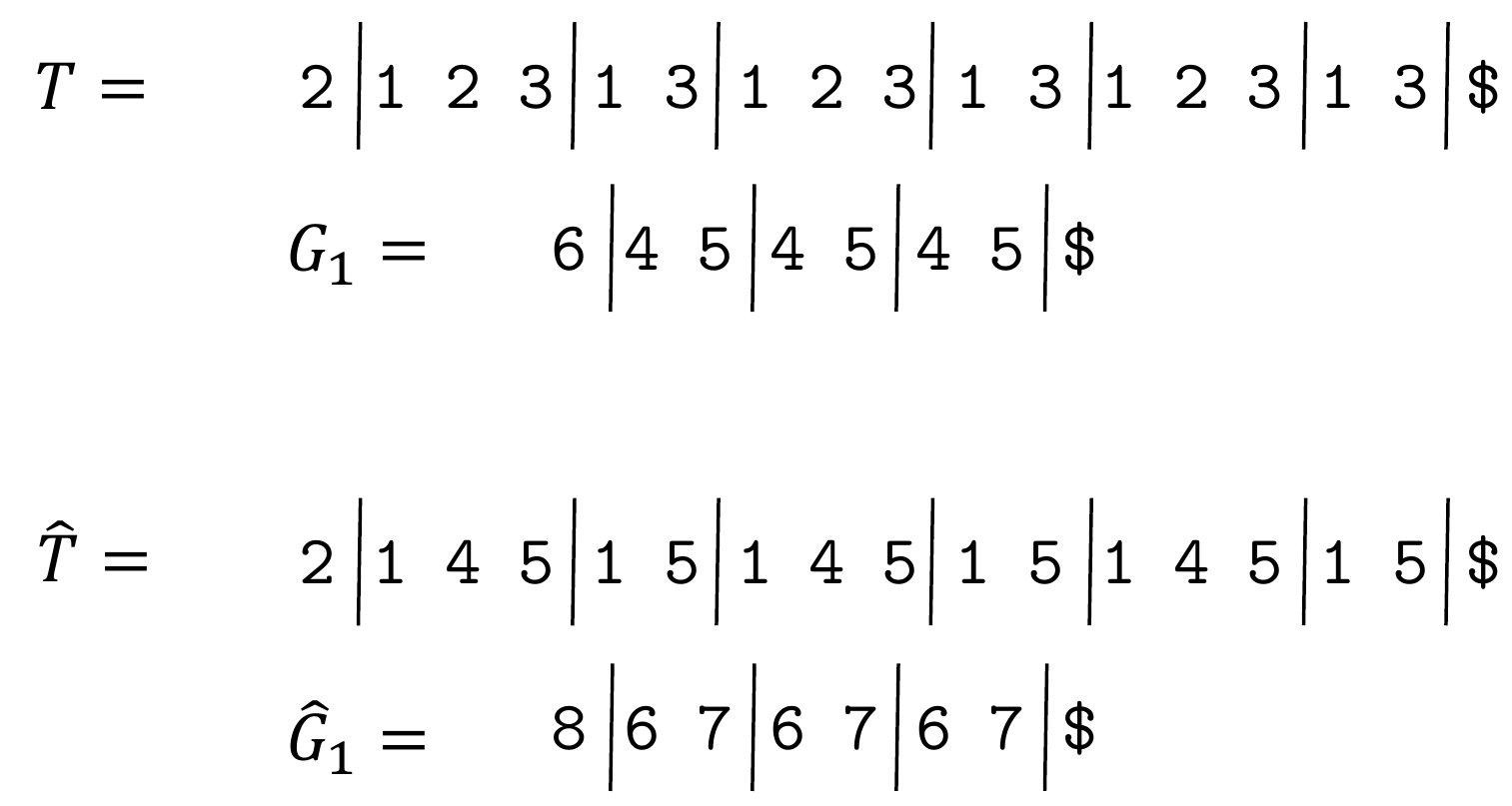}
 \caption{Two strings $T$ and $\hat{T}$ that can be obtained by replacing some characters in $T$ without changing the relative order of any characters, result in the same number of factors in the GCIS-parsing, and each length exactly matches in both of the strings. Therefore, $\Vert \mathcal{D}_1 \Vert =  \Vert\mathcal{\hat{D}}_1\Vert$ and $|G_1| = |\hat{G}_1|$ holds, and $\hat{G}_1$ is recursively the string that can be obtained by replacing some characters in $G_1$ without changing the relative order of any characters. Therefore, we can consider the size of GCIS of $T$ using such a string $\hat{T}$ instead of $T$ itself.}
 \label{gcis_fig3}
\end{figure}

A natural consequence of Lemma~\ref{gcis_identically_equal} is that
edit operations which do not change the relative order
of the characters in $T$ do not affect the size of the grammar.

From now on, we analyze how the size of the GCIS of the string $T$
can increase after the edit operation in the string $T'$.
In the following lemmas, let $1 \leq h \leq r$,
where $r$ is the height of the GCIS grammar for $T$.

\begin{lemma}
\label{gcis_changex}
If $F(G_h,G'_h) = (x,y)$, then
 $| \mathcal{D}_{h+1} \setminus {\mathcal{D'}}_{h+1}| \leq 2 + \lceil (x+1)/2 \rceil $.
\end{lemma}
\begin{proof}
First, let $c$ be the position such that $G'_h$ can be obtained from $G_h$ by deleting a substring of length $x$ from position $c$ and inserting a substring of length $y$ to position $c$. Let $z$ and $z'$ be the numbers of factors in the GCIS-parsing of $G_h$ and $G'_h$, respectively.

Considering $k$ where $i_k \leq c < i_{k+1}$ in $G_h$, the LMS positions $i_1, \ldots, i_{k-1}$ are also the LMS positions in $G'_h$, and for all $j$ where $1 \leq j \leq k-2$, $D_j = D'_j$ holds.
Similarly, for $l$ where $ i_{z-l-1} \leq c+x < i_{z-l} $ in $G_h$, the positions $i_{z-l}, \ldots, i_{z}$ and corresponding positions $i_{z'-l}, \ldots, i_{z'}$ in $G'_h$ are also LMS positions. Therefore, for $l \leq j \leq z$ and $j' = j +(z'-z)$, $D_j = D'_{j'}$. Note that $i_{z-l-1} - i_{k+1} < x$. Since $|D_j| \geq 2, |D'_j| \geq 2$ with $2 \leq j \leq z$, we obtain $| \mathcal{D}_{h+1}  \setminus {\mathcal{D'}}_{h+1}| \leq |\{ D_{k-1} \ldots D_{z-l-1} \} | \leq 2 + \lceil (x+1)/2 \rceil $.
\end{proof}

\begin{lemma}
\label{gcis_order4}
If $F(G_h,G'_h) = (x,y)$, $\Vert\mathcal{D'}_h\Vert \leq 4\Vert\mathcal{D}_h\Vert - x + y$.
\end{lemma}
\begin{proof}
Considering $k$ where $i_k \leq c < i_{k+1}$ in $G_h$ and $l$ where $ i_{z-l-1} \leq c+x < i_{z-l} $, the total length of new factors to be added in $\mathcal{D'}_h$, is at most $i_{z-l}-x+y - i_{k-2} \leq 3\Vert\mathcal{D}_h\Vert- x-y$.
\end{proof}

\begin{lemma}
\label{gcis_changez}
If $F(G_h,G'_h) = (x,y)$, $|G'_{h+1}| \leq |G_{h+1}| + 1 + \lfloor y/2 \rfloor$.
\end{lemma}
\begin{proof}
Assume $|G'_{h+1}| >  |G_{h+1}| + 1 + \lfloor y/2 \rfloor $. In other words, there are at least $2 + \lfloor y/2 \rfloor$ positions which are not LMS positions in $G_h$ but are LMS positions in $G'_h$.
Let $i$ be the right-most position where $G_h[i] \neq G_h[c]$ and $i < c$. For $k$ of $1 \leq k \leq i$, $G_h[k]$ and $G'_h[k]$ are of the same type.

For all $k$ with $c < k \leq |G_h|-x$, $G_h[k+x]$ and $G'_h[k+y]$ are of the same type.
Therefore, only $G_h[i..c+x+1]$ and $G'_h[i..c+y+1]$ can introduce new factors. 
Note that $G_h[i+1..c-1]$ and $G'_h[i+1..c-1]$ are of the same types by Lemma~\ref{sl_define},
There are $y+2$ positions $i+1, c, c+1, \ldots, c+y$ that can be new LMS positions in $G'_h$.
Since any LMS position must be the left-most position of consecutive type S positions, two possible positions adjacent each other cannot be LMS positions at the same time. Therefore, it is impossible to create $2 + \lfloor y/2 \rfloor$ new LMS positions if $y$ is even. If $y$ is odd, we can make $2 + \lfloor y/2 \rfloor$ new LMS positions in $G'_h$ by selecting $i+1, c+1, c+3,  \ldots, c+y$. However, in that case, $G_h[i+1]$ must be of type L, and $G_h[c+x]$ must be of type S since $G_h[c+x]$ and $G'_h[c+y]$ are of the same type. Then, there is at least an LMS position between $T[c]$ and $T[c+y]$, and we cannot create $2 + \lfloor y/2 \rfloor$ new LMS positions in $G'_h$, or there is at least an LMS position in $G_h[c \ldots c+x]$ that is not in $G'_h$. Therefore, whichever $y$ is even or odd, $|G'_{h+1}| \leq |G_{h+1}| + 1 + \lfloor y/2 \rfloor$ holds.
See Figure~\ref{gcis_fig2} for illustration.
\end{proof}

\begin{figure}[tbh]
 \centering
 \includegraphics[keepaspectratio, scale=0.5]
      {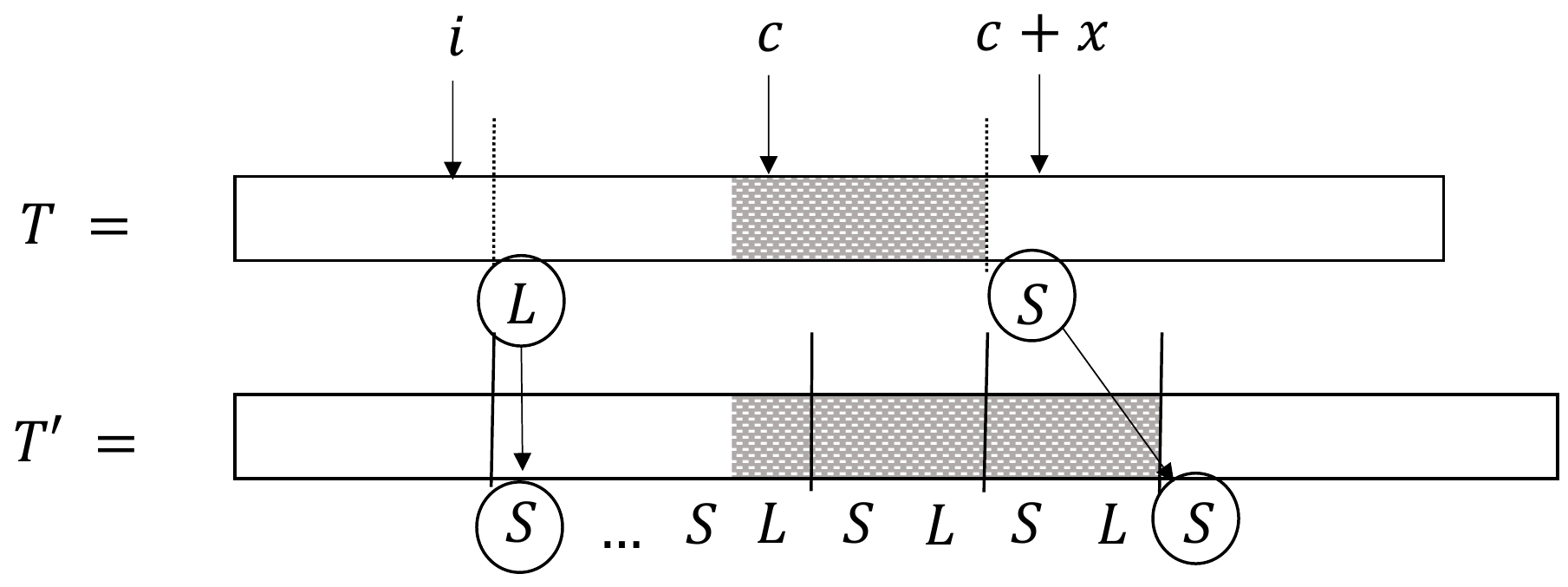}
 \caption{The case that there are $2 + \lfloor y/2 \rfloor$ new LMS positions in $T'$.  The position $i+1$ in $T$ must be of type L to turn position $i+1$ to a new LMS position in $T'$. The position $c+x$ in $T$ must be of type S to turn position $c+y$ to a new LMS position in $T'$. Then, there is at least an LMS position in $T[i+1..c+x]$. }
 \label{gcis_fig2}
\end{figure}

\begin{lemma}
\label{gcis_changexysub}
If $F(G_h,G'_h) = (x,y)$, 
let $a = | \mathcal{D}_{h+1} \setminus {\mathcal{D'}}_{h+1}|$, $b = | \mathcal{D'}_{h+1} \setminus {\mathcal{D}}_{h+1}|$. Then
$a \leq 2 + \lceil (x+1)/2 \rceil $, 
$b \leq 2 + \lceil (y+1)/2 \rceil $, 
$a+b \leq 4+\lfloor (x+y)/2 \rfloor$ hold.
\end{lemma}
\begin{proof}
We immediately get $a \leq 2 + \lceil (x+1)/2 \rceil $, $b \leq 2 + \lceil (x+1)/2 \rceil $ by a direct application of Lemma~\ref{gcis_changex}.
Assume $y \mod 2 = 1, b \leq 2 + \lceil (y+1)/2 \rceil $. Then, Lemma~\ref{gcis_changez} shows that there is only one possible combination of new $b$ LMS positions $i+1, c+1, c+3,  \ldots, c+y$ in $G'_h$. For that, neither $i+1$ nor $c+x$ can be  LMS positions in $G_h$ in this case since they must be new LMS positions in $G'_h$. Therefore, $a \leq 2 + \lceil (x+1)/2 \rceil - 1$ since there are no possible combination of $a+1$ LMS positions in $G_h$.
Assume $x \bmod 2 = 1$ and $a \leq 2 + \lceil (x+1)/2 \rceil $. Then, Lemma~\ref{gcis_changez} shows that there is only one possible combination of $a$ disappearing LMS positions $i+1, c+1, c+3,  \ldots, c+x$ in $G_h$. For that, neither $i+1$ not $c+y$ can be LMS positions in $G'_h$ in this case since they must be disappearing LMS positions in $G_h$. Therefore, $a \leq 2 + \lceil (x+1)/2 \rceil - 1$ since there are no possible combination of new $b+1$ LMS positions in $G'_h$.
\end{proof}

\begin{lemma}
\label{gcis_changexy}
If $F(G_h,G'_h) = (x,y)$, there are two strings $\hat{G}_{h+1}, \hat{G'}_{h+1}$ such that $\hat{G}_{h+1},\hat{G}'_{h+1}$ can be obtained by replacing some characters in $G_{h+1}, G'_{h+1}$ without changing the relative order of any characters in $G_{h+1},G'_{h+1}$, respectively, and $F(\hat{G}_{h+1},\hat{G}'_{h+1}) = (a,b)$, where 
$a \leq 2 + \lceil (x+1)/2 \rceil $, $b \leq 2 + \lceil (y+1)/2 \rceil $, and $a+b \leq 4+\lfloor (x+y)/2 \rfloor$.
\end{lemma}
\begin{proof}
Assume $G_h = D_1, \ldots, D_z$ and $G'_h = D'_1, \ldots, D'_{z'}$ are the GCIS-parsings of $G_h$ and $G'_h$, respectively. By Lemma~\ref{gcis_changex}, there are at most $j = 2 + \lceil (x+1)/2 \rceil $ consecutive factors $D_i, \ldots, D_{i+j-1}$ in $ \mathcal{D}_{h+1} \setminus {\mathcal{D'}}_{h+1}$, and at most $\hat{j} = 2 + \lceil (y+1)/2 \rceil $ consecutive factors $D_{\hat{i}}, \ldots, D_{\hat{i}+\hat{j}-1}$ in $ \mathcal{D'}_{h+1} \setminus {\mathcal{D}}_{h+1}$ and $D_k = D'_k$ for all $1 \leq k \leq \max(i,\hat{i})$, and $D_{z-k} = D'_{z'-k}$ for all $0 \leq k \leq \max(z-i-j-1,z'-\hat{i}-\hat{j}-1)$. By Lemma~\ref{gcis_changexysub}, $j + \hat{j} \leq 4+\lfloor (x+y)/2 \rfloor $. Let
\begin{eqnarray*}
\hat{S_p} & = & |\{D_s : D_s \prec D_p (1 \leq s \leq z)\}|+|\{D'_s : D'_s \prec D_p (1 \leq s \leq z')\}|,  \\
\hat{S'_p} & = & |\{D_s : D_s \prec D'_p (1 \leq s \leq z)\}|+|\{D'_s : D'_s \prec D'_p (1 \leq s \leq z')\}|.  
\end{eqnarray*}
Then, the string $\hat{G}_{h+1} = \hat{S}_1 \cdots \hat{S}_z$ can be obtained from $G_{h+1}$ by replacing some characters in $G_{h+1}, G'_{h+1}$ without changing the relative order of any characters, and $\hat{G'_{h+1}}$ as well. In addition, $F(\hat{G}_{h+1},\hat{G}'_{h+1}) = (j,\hat{j})$ holds because $\hat{R}_k = \hat{R}'_k$ for all $1 \leq k \leq \max(i,\hat{i})$, and $\hat{R}_{z-k} = \hat{R}'_{z'-k}$ for all $0 \leq k \leq \max(z-i-j-1, z'-\hat{i}-\hat{j}-1)$. See Figure~\ref{gcis_fig5}
. \end{proof}

\begin{figure}[tbh]
 \centering
 \includegraphics[keepaspectratio, scale=0.5]
      {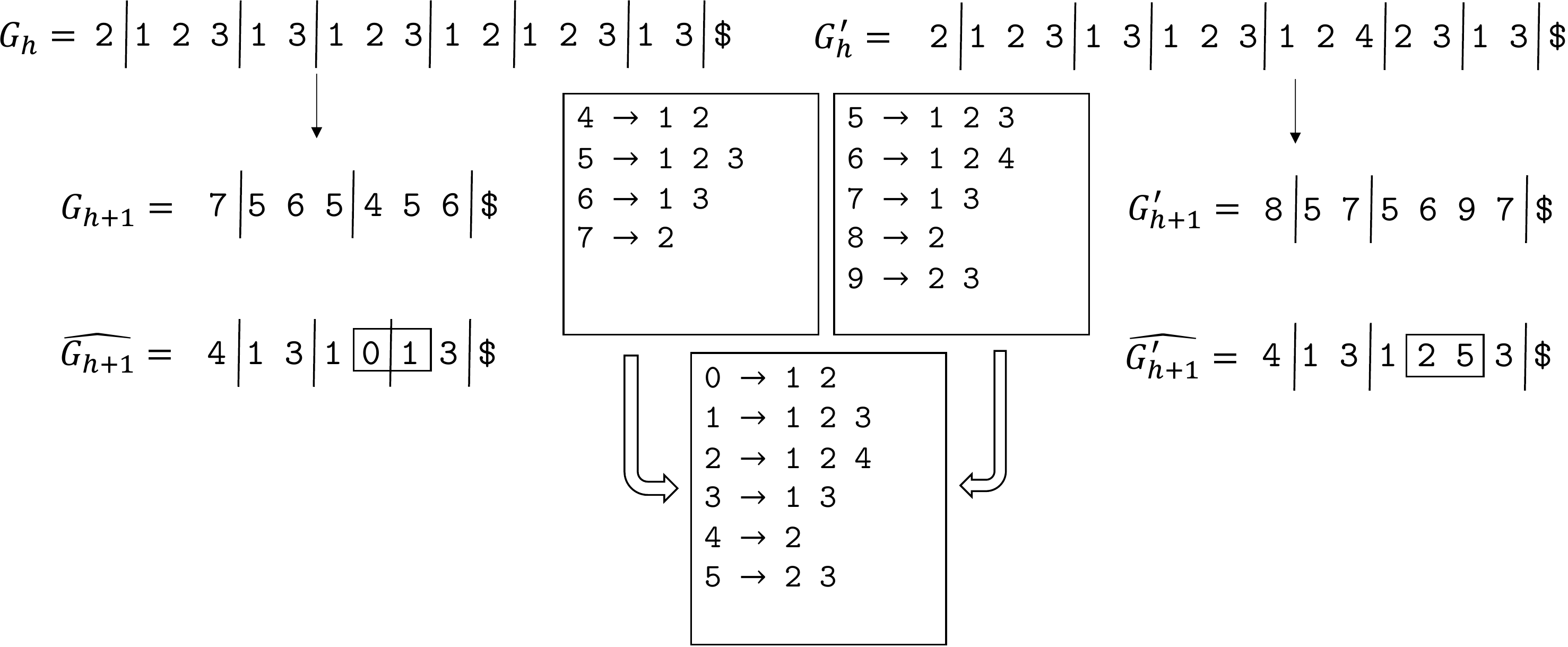}
      \caption{Examples of $\hat{G_{h+1}}$ and $\hat{G}'_{h+1}$ for strings $G_h$ and $G'_h$, where $G'_h$ can be obtained from $G_h$ by substituting a $\mathtt{1}$ with a $\mathtt{4}$. The box under the 2 other boxes shows the ``common'' productions common to $\hat{G_{h+1}}$ and $\hat{G'_{h+1}}$, e.g., by applying $\mathtt{1} \to \mathtt{123}$ to the occurrences of $\mathtt{1}$ in $\hat{G}_{h+1}$ and $\hat{G}'_{h+1}$,
        we obtain the corresponding substrings $\mathtt{123}$ in $G_h$ and $G'_h$. Since the common number is assigned to each equal factors in $G_h$ and $G'_h$, the size of the symmetric difference of $\hat{G}_{h+1}$ and $\hat{G}'_{h+1}$ equals to the number of factors changed by substitution, insertion, or deletion from $G_h$ to $G'_h$, which is four in this example.
      }
 \label{gcis_fig5}
\end{figure}

\begin{lemma}
\label{gcis_colloxy}
If $F(T,T') \in \{ (1,1), (1,0), (0,1)\}$, then there are two strings $\hat{G}_1,\hat{G}'_1$ such that $\hat{G}_1,\hat{G}'_1$ can be obtained by replacing some characters in $G_1,G'_1$ without changing the relative order of any characters in $G_1,G'_1$, respectively, and $F(\hat{G_1},\hat{G'_1}) = (a,b)$, where 
$a \leq 4 $, $b \leq 4 $, $a+b \leq 7$.
\end{lemma}
\begin{proof}
Immediately follows from Lemma~\ref{gcis_changexy}.
\end{proof}

\begin{lemma}
\label{gcis_top4plus}
If $F(G_h,G'_h) = (x,y)$ and $r = h+1$, then $\gis(G'_h) \leq \gis(G_h) + 2(1+ y-|G_{h+1}|)$.
\end{lemma}
\begin{proof}
By construction of GCIS, $\gis(G'_{h+1}) \leq 2|G'_{h+1}|$.
Remembering $||\mathcal{D'}_h|| \leq 4||\mathcal{D}_h|| - x + y$,
\begin{eqnarray*}
\gis(G'_h) &=& ||\mathcal{D'}_h|| + \gis(G'_{h+1}) \\
&\leq& 4||\mathcal{D}_{h+1}|| - x + y + 2|G'_{h+1}| \\
&\leq& 4||\mathcal{D}_{h+1}|| - x + y + 2(|G_{h+1}|+1 + \lfloor y/2 \rfloor) \\
&=& 4||\mathcal{D}_{h+1}|| - x + y + 2|G_{h+1}|+2 + y \\
&=& 4 ( ||\mathcal{D}_{h+1}||+ |G_{h+1}| ) - 2|G_{h+1}| + 2 - x + 2y \\
&=& 4\gis(G_h) - 2|G_{h+1}| + 2 - x + 2y \\
&\leq& 4\gis(G_h) + 2(1+ y-|G_{h+1}|). \\
\end{eqnarray*}
\end{proof}

\begin{lemma}
\label{gcis_top4}
If $F(T,T') = (x,y)$, $x \leq 1$, $y \leq 1$, and
$r = 1$, then $\gis(T') \leq 4\gis(T)$ holds.
\end{lemma}
\begin{proof}
By Lemma~\ref{gcis_top4plus},
\begin{eqnarray*}
\gis(T') &=& \Vert\mathcal{D'}_1\Vert + \gis(G'_1) \\
&\leq& 4\Vert\mathcal{D}_1\Vert - x + y + 2|G'_1| \\
&\leq& 4\Vert\mathcal{D}_1\Vert - x + y + 2(|G_{1}|+1 + \lfloor y/2 \rfloor) \\
&\leq& 4\Vert\mathcal{D}_1\Vert + 2|G_{1}| + 3 \\
&=& 4\Vert\mathcal{D}_1\Vert + 4|G_{1}| + 3 - 2|G_{1}| \\
&=& 4\gis(T) + 3 - 2|G_{1}|.
\end{eqnarray*}
Therefore, if 
$|G_1| \geq 2$, then $\gis(T') \leq 4\gis(T)$. If $|G_1| = 1 $, which is a special case, then $\gis(T') \leq  \Vert\mathcal{D}_1\Vert- x + y + 2|G'_1| \leq 4\Vert\mathcal{D}_1\Vert + 4|G_1| = 4\gis(T)$.
\end{proof}

\begin{lemma}
\label{gcis_under4plus}
If $F(G_h,G'_h) = (x,y)$, $r \geq h+2$, $| \mathcal{D}_h | \geq 2$ and $\Vert\mathcal{D'}_h\Vert \leq 4(\Vert\mathcal{D}_h\Vert-2) - x + y$.
\end{lemma}
\begin{proof}
If $| \mathcal{D}_h | = 1$, then $G_{h+1}$ must be a unary string, and therefore no $G_{h+2}$ is constructed. If $| \mathcal{D}_h | = 2$ and there is a factor of length 1 in $\mathcal{D}_h$, then $G_{h+1}$ is still a unary string except for the first position, and therefore no $G_{h+2}$ is constructed. Therefore, $G_{h+2}$ is constructed only if $| \mathcal{D}_h | \geq 2$ and there are at least two factors of length at least $2$, and hence $\Vert\mathcal{D'}_h\Vert \leq 4(\Vert\mathcal{D}_h\Vert-2) - x + y$ holds.
\end{proof}

If $F(T,T') = (x, y) \in \{(1,1), (1,0), (0,1)\}$, then $\Vert \mathcal{D}_h \Vert - 4(\Vert \mathcal{D}_h \Vert -2) - x + y = 8 + x - y \geq 7$.
It means that $\mathcal{D'}_1$ can afford to 7 character room to charge. Lemma~\ref{gcis_4all} shows that we can use the room to charge the extra additive factor of $7$ in $|G'_r|$, and leads us to the desired upper bound $\gis(T') \leq 4\gis(T)$, as follows:

\begin{lemma}
\label{gcis_4all}
If $F(T,T') \in \{ (1,1), (1,0), (0,1)\}$, then $\gis(T') \leq 4\gis(T)$.
\end{lemma}
\begin{proof}
By Lemma~\ref{gcis_top4}, the lemma holds when $r=1$.
If $r \geq 2$, Lemma~\ref{gcis_colloxy} shows that there are two strings $\hat{G_1}, \hat{G'_1}$, and $F(\hat{G_1}, \hat{G'_1}) = (a,b)$, where $a\leq4$, $b\leq 4$, $a+b \leq 7$ and $\gis(\hat{G_1})=\gis({G_1})$, $\gis(\hat{G'_1})=\gis(G'_1)$ by Lemma~\ref{gcis_identically_equal}. Additionally, Lemma~\ref{gcis_under4plus} shows that $\Vert\mathcal{D'}_1\Vert = \Vert\mathcal{\hat{D'}}_1\Vert \leq 4\Vert\mathcal{D}_1\Vert-7 = 4\Vert\mathcal{\hat{D}}_1\Vert-7$. Since $F(\hat{G_1}, \hat{G'_1}) = (a,b)$, there are also two strings $\hat{G_h}, \hat{G'_h}$, and $F(\hat{G_h},\hat{G'_h}) = (a,b)$, where $a\leq4$, $b\leq 4$, $a+b \leq 7$ and $\gis(\hat{G_h})=\gis({G_h}), \gis(\hat{G'_h})=\gis(G'_h)$ with $2 \leq h \leq r-1$ by Lemma~\ref{gcis_identically_equal}. Furthermore, Lemma~\ref{gcis_under4plus} shows that $\Vert\mathcal{D'}_1\Vert = \Vert\mathcal{\hat{D'}}_1\Vert \leq 4\Vert\mathcal{\hat{D}}_1\Vert-4 = 4\Vert\mathcal{D}_1\Vert$.
Noting that $\gis(\hat{G'_r}) \leq 2|\hat{G'_r}| = 2|G_r| = 2\gis(G_r)$ and $|G_k| \geq 1$, we obtain:
\begin{eqnarray*}
\gis(T') &=& \Vert\mathcal{D'}_1\Vert + \gis(G'_1) \\
&\leq& (4\Vert\mathcal{D}_1\Vert - 7 ) + \gis(\hat{G'_1}) \\ 
&\leq& ( 4\Vert\mathcal{D}_1\Vert - 7 ) + (4\Vert\mathcal{D}_2\Vert - 4) + \gis(\hat{G'_2}) \\ 
&\leq& \sum^{r-1}_{t=1} (4\Vert\mathcal{D}_t\Vert) - 7 + \gis(\hat{G'_r}) \\ 
&\leq& \sum^{r-1}_{t=1} ( 4\Vert\mathcal{D}_t\Vert) - 7  + 2|G_r| + 8 \\ 
&=& \sum^{r-1}_{t=1} ( 4\Vert\mathcal{D}_t\Vert) + 2|G_r| +1 \\ 
&<& \sum^{r-1}_{t=1} ( 4\Vert\mathcal{D}_t\Vert) + 4|G_r| \\
&=& 4\gis(T).
\end{eqnarray*}
\end{proof}

\subsection{Lower bounds for the sensitivity of $\gis$}

\begin{theorem}
  The following lower bounds on
  the sensitivity of GCIS hold: \\
  \textbf{substitutions:} $\liminf_{n \rightarrow \infty}\MSensSub(\gis,n) \geq 4$.  $\ASensSub(\gis,n) \geq 3\gis - 13 = \Omega(n)$. \\
  \textbf{insertions:} $\liminf_{n \rightarrow \infty}\MSensIns(\gis,n) \geq 4$.
    $\ASensIns(\gis,n)  \geq 3\gis - 24 = \Omega(n)$. \\
  \textbf{deletions:} $\liminf_{n \rightarrow \infty}\MSensDel(\gis,n) \geq 4$. 
    $\ASensDel(\gis,n)  \geq 3\gis - 29 = \Omega(n) $.
\end{theorem}

\begin{proof}

Assume $p > 1$. 

\textbf{substitutions:}
Consider the following string of length $n = 4p + 4 \in \Theta(p)$: 
$$T = \mathtt{2^{\mathit{p}}3 2^{\mathit{p}}3 2^{\mathit{p}}3 2^{\mathit{p}}3}$$
By the construction of the GCIS grammar of $T$, we obtain $ \mathcal{D}_1= \{ \mathtt{2^{\mathit{p}}3} \}$, $G_1 = \mathtt{4444}$, $\gis(T) = \Vert \mathcal{D}_1 \Vert + |G_1| = p+5$. The following string
\[
T' = \mathtt{2}^{\mathit{p}}3 \mathtt{2}^{\mathit{p}}3 \mathtt{2}^{\mathit{p}}\mathtt{1} \mathtt{2}^{\mathit{p}}\mathtt{3}
\]
can be obtained from $T$ by substituting the third $\mathtt{3}$ with $\mathtt{1}$.
By the construction of GCIS grammar of $T$, we obtain $G'_1 = \mathtt{564}, \mathcal{D'}_1= \{ \mathtt{12^{\mathit{p}}3, 2^{\mathit{p}}3, 2^{\mathit{p}}32^{\mathit{p}}} \}, \gis(T') = \Vert \mathcal{D'}_1 \Vert + |G'_1| = 4p+7$, which leads to $\liminf_{n \rightarrow \infty}\MSensSub(\gis,n) \geq  \liminf_{p \rightarrow \infty}(4p + 7) / (p + 5) = 4$, and $\ASensSub(\gis,n) =  (4p + 7) - (p + 5) = 3p+2 = 3\gis -13 = \Omega(n)$. 

\textbf{insertions:}
Consider the following string of length $n = 8p + 12 \in \Theta(p)$: 
$$T = \mathtt{(12)^{\mathit{p}}122(12)^{\mathit{p}}122(12)^{\mathit{p}}122(12)^{\mathit{p}}122}$$
By the construction of the GCIS grammar, we obtain $ \mathcal{D}_1 = \{ \mathtt{12,122} \}$, $G_1 = \mathtt{3^{\mathit{p}}4 3^{\mathit{p}}4 3^{\mathit{p}}4 3^{\mathit{p}}4}$, $\mathcal{D}_2= \{ \mathtt{3^{\mathit{p}}4} \}$, $G_2 = \mathtt{5555}$,  $\gis(T) = \Vert \mathcal{D}_1 \Vert+\Vert\mathcal{D}_2\Vert + |G_2| = 5+(p+1)+4 = p+10$.

The string
$$T' = \mathtt{(12)^{\mathit{p}}122(12)^{\mathit{p}}122(12)^{\mathit{p}}1122(12)^{\mathit{p}}122}$$
can be obtained from $T$ by inserting $\mathtt{0}$ to just before the third $\mathtt{11}$.
By the construction of GCIS grammar of $T'$, we obtain $\mathcal{D}_1 = \{ \mathtt{112, 12,1122} \}, G'_1 = \mathtt{4^{\mathit{p}}5 4^{\mathit{p}}5 4^{\mathit{p}}3 4^{\mathit{p}}5}, G'_2 = \mathtt{786},  \mathcal{D'}_2= \{ \mathtt{34^{\mathit{p}}5, 4^{\mathit{p}}5, 4^{\mathit{p}}54^{\mathit{p}}} \},\\ \gis(T') = \Vert \mathcal{D'}_1 \Vert+\Vert\mathcal{D'}_2\Vert + |G'_2| = 9+(4p+4)+3 = 4p+16$, which leads to $\liminf_{n \rightarrow \infty}\MSensIns(\gis,n) \geq \liminf_{p \rightarrow \infty}(4p+16)/(4p + 10) = 4$, and $\ASensIns(\gis,n) =  (4p + 16) - (p + 10) = 3p+6 = 3\gis - 24 = \Omega(n)$.

\textbf{deletions:}
Consider the following string of length $n = 12p+12 \in \Theta(p)$: 
$$T = \mathtt{(122)^{\mathit{p}}132(122)^{\mathit{p}}132(122)^{\mathit{p}}132(122)^{\mathit{p}}132}$$
By the construction of the GCIS grammar of $T$, we obtain $\mathcal{D}_1 = \{ \mathtt{122,132} \}$, $G_1 = \mathtt{4^{\mathit{p}}5 4^{\mathit{p}}5 4^{\mathit{p}}5 4^{\mathit{p}}5}$, $\mathcal{D}_2= \{ \mathtt{4^{\mathit{p}}5} \}$, $G_2 = \mathtt{6666}$, $\gis(T) = \Vert \mathcal{D}_1 \Vert+\Vert\mathcal{D}_2\Vert + |G_2| = 6+(p+1)+4 = p+11$.

The string
$$T' = \mathtt{(122)^{\mathit{p}}132(122)^{\mathit{p}}132(122)^{\mathit{p}}12(122)^{\mathit{p}}132}$$
can be obtained from $T$ by deleting the third $\mathtt{3}$.
By the construction of GCIS grammar of $T'$, we obtain $ \mathcal{D}_1 = \{ \mathtt{01, 011,021} \}, G'_1 = \mathtt{5^{\mathit{p}}6 5^{\mathit{p}}6 5^{\mathit{p}}4 5^{\mathit{p}}6}, \mathcal{D'}_2= \{ \mathtt{45^{\mathit{p}}6, 5^{\mathit{p}}6, 5^{\mathit{p}}65^{\mathit{p}}} \}, G'_2 = \mathtt{897}, \gis(T') = \Vert \mathcal{D'}_1 \Vert+\Vert\mathcal{D'}_2\Vert + |G'_2| = 8+(4p+4)+3 = 4p+15$, which leads to $\liminf_{n \rightarrow \infty}\MSensDel(\gis,n) \geq \liminf_{p \rightarrow \infty}(4p + 15)/(p + 11) = 4$, and $\ASensDel(\gis,n) =  (4p + 15) - (p + 11) = 3p+4 = 3\gis -29 = \Omega(n)$. 
\end{proof}

\section{Bisection}
\label{sec:bisection}

In this section, we consider the worst-case sensitivity
of the compression algorithm \emph{Bisection}~\cite{NelsonKC95}
which is a kind of grammar-based
compression that has a tight connection to BDDs.

Given a string $T$ of length $n$,
the bisection algorithm builds a grammar generating $T$ as follows.
We consider a binary tree $\mathcal{T}$ whose root
corresponds to $T$.
The left and right children of the root correspond to
$T_1 = T[1..2^j]$ and $T_2 = T[2^j+1..n]$, respectively,
where $j$ is the largest integer such that $2^j < n$.
We apply the same rule to $T_1$ and to $T_2$ recursively,
until obtaining single characters which are the leaves of $\mathcal{T}$.
After $\mathcal{T}$ is built,
we assign a label (non-terminal) to each node of $\mathcal{T}$.
If there are multiple nodes such that the leaves of their subtrees
are the same substrings of $T$,
we label the same non-terminal to all these nodes.
The labeled tree $\mathcal{T}$ is the derivation tree
of the bisection grammar for $T$.
We denote by $\gbsc(T)$ the size of the bisection grammar for $T$.
Recall that $\Sigma$ is the alphabet.

Let us briefly consider the case of unary alphabet $\Sigma_1 = \{a\}$.
Let $h(\mathcal{T})$ denote the height
of the derivation tree $\mathcal{T}$ for $T = a^n$.
After obtaining $T' = a^{n+1}$ for insertion or 
$T' = a^{n-1}$ for deletion, at most
$h(\mathcal{T})-1$ new productions are added
(note that $X \rightarrow a$ exists both for $T$ and for $T'$).
Thus the additive sensitivity of Bisection for unary alphabets
is at most $h(\mathcal{T})-1$.
This bound is almost tight, e.g. deleting a single $a$ from $T = a^{2^k}$
adds new $k-2 = h(\mathcal{T})-2$ non-terminals to
the existing $k = h(\mathcal{T})$ non-terminals
(note that the production $X \rightarrow a$ remains and the existing root of $\mathcal{T}$
is replaced with the new one).
The multiplicative sensitivity for Bisection is thus asymptotically $2 = |\Sigma_1|+1$.

In what follows, let us consider the case
of multi-character alphabets, where at least one of $T$ and $T'$ contains
two or more distinct characters.

\subsection{Lower bounds for the sensitivity of $\gbsc$}

\begin{theorem}
  The following lower bounds on
  the sensitivity of $\gbsc$ hold: \\
  \textbf{substitutions:} $\liminf_{n \rightarrow \infty}\MSensSub(\gbsc,n) \geq 2$.
  $\ASensSub(\gbsc,n) \geq \gbsc-4$ and $\ASensSub(\gbsc,n) \geq 2 \log_2 n - 4$. \\
  \textbf{insertions:} $\liminf_{n \rightarrow \infty}\MSensIns(\gbsc,n) \geq |\Sigma|$.
  $\ASensIns(\gbsc,n) \in \Omega(|\Sigma| \gbsc) $ and $\ASensIns(\gbsc,n) \in \Omega \left(|\Sigma|^2 \log \frac{n}{|\Sigma|}  \right) $. \\
  \textbf{deletions:} $\liminf_{n \rightarrow \infty}\MSensDel(\gbsc,n) \geq |\Sigma|$.
    $\ASensDel(\gbsc,n) \in \Omega(|\Sigma| \gbsc) $ and $\ASensDel(\gbsc,n) \in$ \\ $\Omega \left( |\Sigma|^2 \log \frac{n}{|\Sigma|} \right) $.
\end{theorem}

\begin{proof}
  \textbf{substitutions:} Consider a unary string $T = \mathtt{a}^n$ with $n = 2^k$.
  The set of productions for $T$ is
  \[
  \begin{array}{ll}
  X_1 = \mathtt{a} & \mbox{(generating $\mathtt{a}$)}, \\
  X_2 = X_1 X_1 & \mbox{(generating $\mathtt{aa}$)}, \\
  X_3 = X_2 X_2 & \mbox{(generating $\mathtt{aaaa}$)}, \\
  \ldots & \\
  X_{k} = X_{k-1} X_{k-1} & \mbox{(generating $\mathtt{a}^{2^{k}}$)},
  \end{array}
  \]
  with $\gbsc(T) = 2k-1$.
Let $T' = \mathtt{a}^{n-1}\mathtt{b}$ that can be obtained by
  replacing the last $\mathtt{a}$ in $T$ with $\mathtt{b}$.
  The set of productions for $T'$ is
  \[
  \begin{array}{ll}
  X_1 = \mathtt{a} & \mbox{(generating $\mathtt{a}$)}, \\
  X_2 = X_1 X_1 & \mbox{(generating $\mathtt{aa}$)}, \\
  X_3 = X_2 X_2 & \mbox{(generating $\mathtt{aaaa}$)}, \\
  \ldots & \\
  X_{k-1} = X_{k-2} X_{k-2} & \mbox{(generating $\mathtt{a}^{2^{k-1}}$)}, \\
  Y_{1} = \mathtt{b} & \mbox{(generating $\mathtt{b}$)}, \\
  Y_2 = X_1 Y_1 & \mbox{(generating $\mathtt{ab}$)}, \\
  Y_3 = X_2 Y_2 & \mbox{(generating $\mathtt{aaab}$)}, \\
  \ldots & \\
  Y_{k} = X_{k-1} Y_{k-1} & \mbox{(generating $\mathtt{a}^{2^{k}-1}\mathtt{b}$)}
  \end{array}
  \]
  with $\gbsc(T') = 2k-1 + 2(k-1)-1 = 4k-4$.
  Thus $\liminf_{n \rightarrow \infty}\MSensSub(\gbsc,n) \geq \liminf_{k \rightarrow \infty}\frac{4k-4}{2k-1} \geq 2$.
  Also, $\ASensSub(\gbsc,n) \geq (4k-4)-(2k-1) = 2k-5 = \gbsc(T)-4$ and $\ASensSub(\gbsc,n) \geq 2 \log_2 n - 4$ as $k = \log_2 n$.

  \textbf{deletions:} Assume that $|\Sigma| = 2^i$ with a positive integer $i \geq 1$. Let $Q$ be a string that contains $t = |\Sigma|^2$ distinct bigrams and $|Q| = |\Sigma|^2+1$.
Let $Q' = Q[2..|Q|]$.
Let $\sigma_i$ denote the lexicographically $i$th character in $\Sigma$.  
We consider the string
\begin{eqnarray*}
  T &=& Q'[1]^{2^p} \cdots Q'[|Q'|]^{2^p}.
\end{eqnarray*}
Note that $p = \log (n / \sigma)$.
  The set of productions for $T$ from depth $1$ to $p$ is:
\begin{eqnarray*}
  X_i &\to& \sigma_i\sigma_i ~~~(1 \leq i \leq p), \\
  X_{p|\Sigma|+i} &\to& X_{(p-1)|\Sigma|+i}X_{(p-1)|\Sigma|+i} ~~~(1 \leq i \leq |\Sigma|, 2 \leq k \leq p). \\
\end{eqnarray*}
Thus, the derivation tree $\mathcal{T}$ has $p|\Sigma|$ internal nodes with distinct labels. Additionally, after height $|\Sigma|$, the string consists of $t-1$ distinct bigrams, and there is no run of length $2$. Then the derivation tree $\mathcal{T}$ has $t-1$ internal nodes with distinct labels in height above $p$. Finally, $\gbsc(T) = p|\Sigma| + t-1$.

We consider the string $T'$ where $T[1]$ is removed,
namely,
\begin{eqnarray*}
  T' &=& T[2..|T|] = Q'[1]^{2^p-1}Q'[2]^{2^p} \cdots Q'[|Q'|]^{2^p}.
\end{eqnarray*}
  The set of productions for $T'$ of height $1$ is:
\begin{eqnarray*}
  X_{(i-1)|\Sigma|+j} &\to& \sigma_i\sigma_j  ~~~(1 \leq i \leq |\Sigma|, 1 \leq j \leq |\Sigma|).
\end{eqnarray*}
Thus, the derivation tree $\mathcal{T'}$ for string $T'$ has $t = |\Sigma|^2$ internal nodes with distinct labels at height one.
Because of this, the number of internal nodes of the derivation tree $\mathcal{T'}$ in each height $2 \leq p' \leq p$ is also at least $t = |\Sigma|^2$. After that, the string of height $p$ consists of $t$ distinct bigrams, and there is no run of length 2, which is the same condition of $T$. Then the derivation tree $\mathcal{T}$ has additional $t-1$ internal nodes with distinct labels in height above $p$. 
Finally, $\gbsc(T') = tp + t$. Then, we obtain:

\begin{eqnarray*}
\MSensDel(\gbsc,n) &\geq& \lim_{n \rightarrow \infty} \frac{tp + t}{p|\Sigma| + t-1} = \lim_{p \rightarrow \infty} \frac{tp + t}{p|\Sigma| + t-1} = \frac{t}{|\Sigma|} \geq |\Sigma|, \\
\ASensDel(\gbsc,n) &\geq& (tp + t) - (p|\Sigma| + t-1) = (t - |\Sigma|)p+1 \in \Omega(|\Sigma|^2 p),\\
\end{eqnarray*}
where $\Omega(|\Sigma|^2 p) = \Omega\left(|\Sigma|^2 \log \frac{n}{|\Sigma|}\right)$ and $\Omega(|\Sigma|^2 p) = \Omega(|\Sigma| \gbsc(T))$.

\textbf{insertions:}
We use the same string $T$ as in the case of deletions.
We consider the string $T'$ that is obtained by prepending $Q[1]$ to $T$,
namely,
\begin{eqnarray*}
  T' &=& Q[1]T = Q[1]Q'[1]^{2^p} \cdots Q'[|Q'|]^{2^p}.
\end{eqnarray*}
  The set of productions for $T'$ of height $1$ is:
\begin{eqnarray*}
  X_{(i-1)|\Sigma|+j} &\to& \sigma_i\sigma_j.  ~~~(1 \leq i \leq |\Sigma|, 1 \leq j \leq |\Sigma|) \\
  X_{|\Sigma|^2+1} &\to& Q[1].
\end{eqnarray*}
Thus, the derivation tree $\mathcal{T'}$ has $t+1$ internal nodes with distinct labels at height one.
Because of this, the number of internal nodes of derivation tree $\mathcal{T'}$ of each height $2 \leq p' \leq p$ is also at least $t = |\Sigma|^2$ nodes. After that, the string of height $p$ consists of $t$ distinct bigrams, and there is no run of length 2, which is the same condition of $T$. Then derivation tree $\mathcal{T}$ has additional $t-1$ internal nodes with distinct labels in height above $p$. 
Finally, $\gbsc(T') = (t+1)p + t$. Then, we obtain:

\begin{eqnarray*}
\MSensIns(\gbsc,n) &\geq& \lim_{n \rightarrow \infty} \frac{(t+1)p + t}{p|\Sigma| + t-1} = \lim_{p \rightarrow \infty} \frac{(t+1)p + t}{p|\Sigma| + t-1} = \frac{(t+1)}{|\Sigma| } \geq |\Sigma|, \\
\ASensIns(\gbsc,n) &\geq& \left(((t+1)p + t) - ( p|\Sigma| + t-1) \right) = (t+1-|\Sigma|)p + 1 \in \Omega(|\Sigma|^2 p), 
\end{eqnarray*}
where $\Omega(|\Sigma|^2 p) = \Omega\left(|\Sigma|^2 \log \frac{n}{|\Sigma|}\right)$ and $\Omega(|\Sigma|^2 p) = \Omega(|\Sigma| \gbsc(T))$.
\end{proof}

We show a concrete example of how the derivation tree of Bisection changes by an insertion in Figure~\ref{bisec_example}.

\begin{figure}
\includegraphics[width=\linewidth]{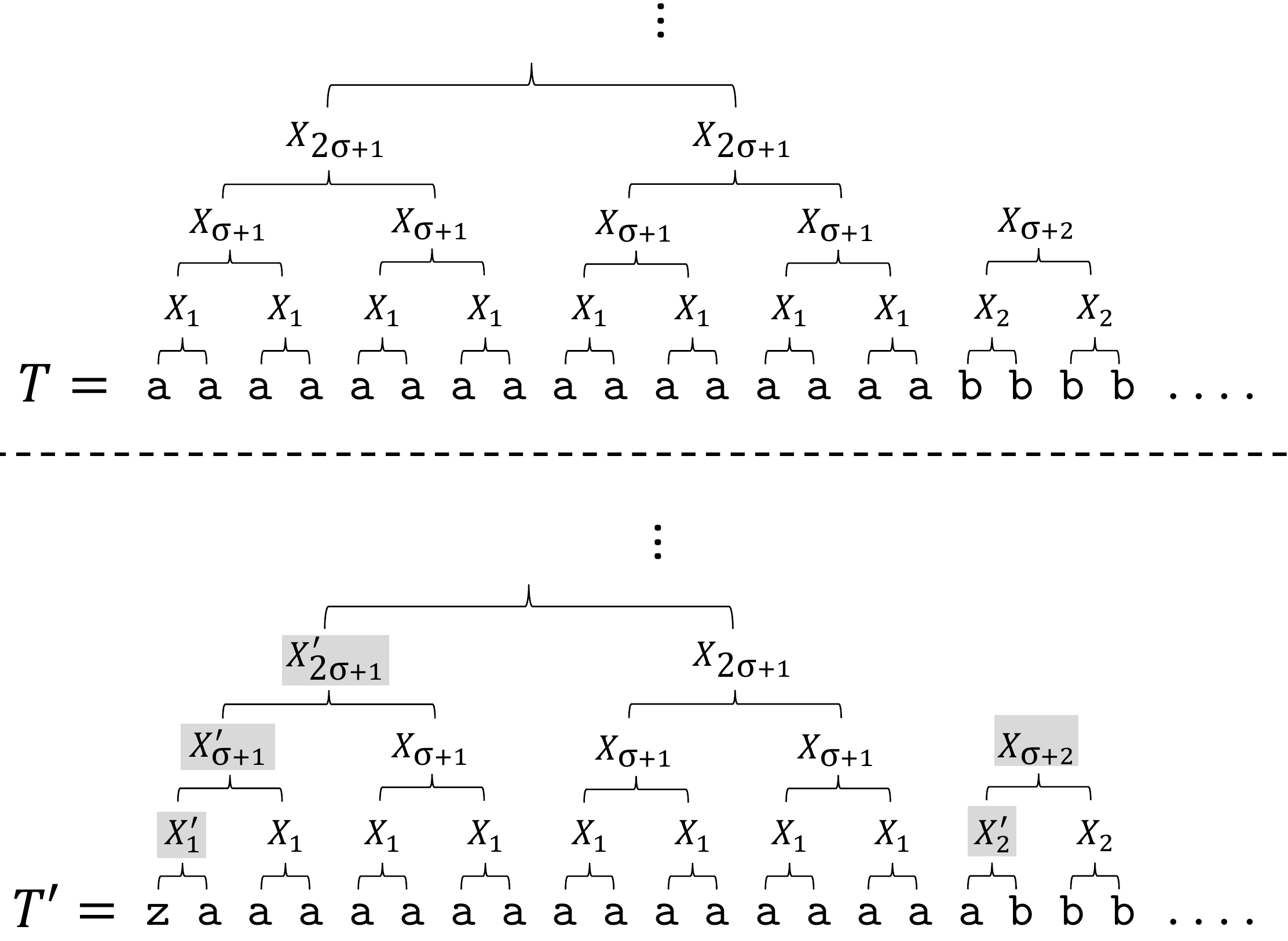}
\caption{An example of insertion for Bisection, where $p = 4$ and $\sigma = |\Sigma|$ in this figure. There are nodes $X_1, X_{\sigma+1}, X_{2\sigma+1}, X_{3\sigma+1}$ in the leftmost path in the derivation tree of $T = \mathtt{a}^{2^4}\mathtt{b}^{2^4} \mathtt{b}^{2^4} \cdots$ (upper). After a $\texttt{z}$ is prepended to $T$~(yielding $T'$), new internal nodes $X'_{1}, X'_{\sigma+1}, X'_{2\sigma+1}, X'_{3\sigma+1}$ that correspond to $\mathtt{za}, \mathtt{za}^3, \mathtt{za}^{7}, \mathtt{za}^{15}$ occur in the derivation tree for $T'$ (lower). This propagates to the other $\sigma-1$ bigrams $\mathtt{ab}$, $\mathtt{bc}$, \ldots, which consist of distinct characters.}
       \label{bisec_example}
\end{figure}

\subsection{Upper bounds for the sensitivity of $\gbsc$}

\begin{theorem}
  The following upper bounds on
  the sensitivity of $\gbsc$ hold: \\
  \textbf{substitutions:} $\MSensSub(\gbsc,n) \leq 2$. $\ASensSub(\gbsc,n) \leq 2\lceil \log_2 n \rceil \leq 2\gbsc$. \\
  \textbf{insertions:} $\MSensIns(\gbsc, n) \leq |\Sigma|+1$. $\ASensIns(\gbsc, n) \leq |\Sigma| \gbsc$.\\
  \textbf{deletions:} $\MSensDel(\gbsc, n) \leq |\Sigma|+1$. $\ASensDel(\gbsc, n) \leq |\Sigma| \gbsc$. 
\end{theorem}

\begin{proof}
  \textbf{substitutions:}
  Let $i$ be the position where we substitute the character $T[i]$.
  We consider the path $P$ from the root of $\mathcal{T}$
  to the $i$th leaf of $\mathcal{T}$ that corresponds to $T[i]$.
  We only need to change the labels of the nodes in the path $P$,
  since any other nodes do not contain the $i$th leaf.
  Since $\mathcal{T}$ is a balanced binary tree,
  the height $h$ of $\mathcal{T}$ is $\lceil \log_2 n \rceil$
  and hence $|P| \leq h = \lceil \log_2 n \rceil$.
  Since $h \leq \gbsc$, we get 
  $\MSensSub(\gbsc, n) \leq 2$.
  Since each non-terminal is in the Chomsky normal form
  and since $\lceil \log_2 n \rceil \leq \gbsc$,
  $\ASensSub(\gbsc, n) \leq 2\lceil \log_2 n \rceil \leq 2\gbsc$.

  \textbf{insertions:}
  Let $i$ be the position where we insert a new character $a$ to $T$,
  and let $\mathcal{T}$ and $\mathcal{T}'$
  be the derivation trees for the strings $T$ and $T'$
  before and after the insertion, respectively.
For any node $v$ in the derivation tree $\mathcal{T}$,
  let $\mathcal{T}(v)$ denote the subtree rooted at $v$.
  Let $\ell(v)$ and $r(v)$ denote the text positions 
  that respectively correspond to the leftmost and rightmost leaves in $\mathcal{T}(v)$.
We use the same analysis for the
  left children of the nodes in the path $P$ from the root to the new $i$th leaf
  which corresponds to the inserted character $a$.
Let $v'$ denote a node in $\mathcal{T}'$.
  From now on let us focus on the subtrees $\mathcal{T}'(v')$
  of $\mathcal{T}'$ such that $\ell(v') > i$
  and $v'$ is not in the rightmost path from the root of $\mathcal{T'}$.
  Let $\str(v')$ denote the string that is derived from the non-terminal for $v'$,
  and let $v$ be the node in $\mathcal{T}$ which corresponds to $v'$.
  Observe that $\str(v') = T'[\ell(v')..r(v')] = T[\ell(v)-1..r(v)-1]$,
  namely, $\str(v')$ has been shifted by one position in the string
  due to the new character $a$ inserted at position $i$.
  Since $T[\ell(v)..r(v)]$ is represented by the node $v$ in $\mathcal{T}$,
  there exist at most $\gbsc$ distinct substrings of $T$
  that can be the ``seed'' of the strings represented by the nodes $v'$ of $\mathcal{T}'$ with $\ell(v') > i$.
  Since the number of left-contexts of each $T[\ell(v)..r(v)]$ is at most $|\Sigma|$,
  there can be at most $|\Sigma|$ distinct shifts from the seed $T[\ell(v)..r(v)]$.
  Since the rightmost paths from the roots of $\mathcal{T}$ and $\mathcal{T}'$
  are all distinct except the root,
  and since inserting the character can increase
  the length of the rightmost path by at most 1,
overall, we have that
  \begin{equation}
    \gbsc(T') \leq |\Sigma| \gbsc(T) + \lceil \log_2 n \rceil + 1 \leq |\Sigma| \gbsc(T) + h(T)+1, \label{eqn:bisection_upperbound}
  \end{equation}
  where $h(T)$ is the height of $\mathcal{T}$.
  For the case of multi-character alphabets $\gbsc(T) \geq h(T)+1$ holds,
  and hence $\gbsc(T') \leq (|\Sigma| + 1)\gbsc(T)$ follows from formula (\ref{eqn:bisection_upperbound}).
  Hence we get $\MSensIns(\gbsc, n) \leq |\Sigma|+1$
  and $\ASensIns(\gbsc, n) \leq |\Sigma| \gbsc$.

  \textbf{deletions:} By similar arguments to the case of insertions,
  we get $\MSensDel(\gbsc, n) \leq |\Sigma|+1$
  and $\ASensDel(\gbsc, n) \leq |\Sigma| \gbsc$.

\end{proof}
 \section{Compact Directed Acyclic Word Graphs (CDAWGs)}
\label{sec:cdawg}

In this section, we consider the worst-case sensitivity
of the size of \emph{Compact Directed Acyclic Word Graphs} (\emph{CDAWGs})~\cite{BlumerBHME87}.
The CDAWG of a string $T$, denoted $\CDAWG(T)$,
is a string data structure that represents the set of suffixes of $T$,
such that the number $v$ of internal nodes in $\CDAWG(T)$ is equal to
the number of distinct maximal repeats in $T$,
and the number $e$ of edges in $\CDAWG(T)$ is equal to
the number of right-extensions of maximal repeats occurring in $T$.
Therefore, the smaller $\CDAWG(T)$ is, the more repetitive $T$ is.
Since $v \leq e$ always holds,
we simply use the number $e$ of edges in the CDAWG
as the \emph{size} of $\CDAWG(T)$, and denote it by $e(T)$.
It is known (c.f.~\cite{BelazzouguiC17})
that $\CDAWG(T)$ induces a grammar-based compression of size $e$
for $T$.

\subsection{Lower bounds for the sensitivity of $e$}

\begin{theorem}
  The following lower bounds on
  the sensitivity of $e$ hold: \\
  \textbf{deletions:} $\liminf_{n \rightarrow \infty}\MSensDel(e,n) \geq 2$.
    $\ASensDel(e,n) \geq e-4$ and $\ASensDel(e,n) \geq n-4$. \\
  \textbf{substitutions:} $\liminf_{n \rightarrow \infty}\MSensSub(e,n) \geq 2$.
    $\ASensSub(e,n) \geq e-2$ and $\ASensSub(e,n) \geq n-2$. \\
  \textbf{insertions:} $\liminf_{n \rightarrow \infty}\MSensIns(e,n) \geq 2$.
    $\ASensIns(e,n) \geq e-2$ and $\ASensIns(e,n) \geq e-2$.
\end{theorem}

\begin{proof}
\textbf{deletions}:
Consider string $T = \mathtt{a}^m\mathtt{b}\mathtt{a}^m\mathtt{b}$
of length $n = 2m+2$.
All the maximal repeats of $T$ are either of form
(1) $\mathtt{a}^h$ with $1 \leq h < m$ or
(2) $\mathtt{a}^m\mathtt{b}$.
Each of those in group (1) has exactly two out-going edges
labeled with $\mathtt{a}$ and $\mathtt{b}$,
and the one in (3) has exactly one out-going edge labeled $\mathtt{a}^m\mathtt{b}$.
Summing up these edges together with the two out-going edges from the source,
the total number of edges in $\CDAWG(T)$ is $2m+1 = n-1$
(see also the left diagram of Figure~\ref{fig:CDAWG}).
Consider string $T' = \mathtt{a}^{2m}\mathtt{b}$ of length $n-1 = 2m+1$
that can be obtained by removing the middle $\mathtt{b}$ from $T$.
$\CDAWG(T)$ has $2m$ internal nodes each of which represents
maximal repeat $a^{k}$ for $1 \leq k < 2m$ and has two out-going edges
labeled with $\mathtt{a}$ and $\mathtt{b}$.
Thus, $\CDAWG(T')$ has exactly $4m = 2n-4$ edges,
including the two out-going edges from the source
(see also the right diagram of Figure~\ref{fig:CDAWG}).
Thus we have $e(T')/e(T) = \frac{4m}{2m+2} = \frac{2n-4}{n}$ which tends to $2$,
and $e(T') - e(T) = 2m-2 = n-4 = e(T)-4$.
This gives us $\liminf_{n \rightarrow \infty}\MSensDel(e,n) \geq 2$
$\ASensDel(e,n) \geq n-4$ and $\ASensDel(e, n) \geq e - 4$.

\textbf{substitutions}:
By replacing the middle $\mathtt{b}$ of $T$ with $\mathtt{a}$,
we obtain string $T'' = \mathtt{a}^{2m+1}\mathtt{b}$,
which gives us similar bounds
$\liminf_{n \rightarrow \infty}\MSensDel(e,n) \geq 2$, 
$\ASensDel(e,n) \geq n-2$ and $\ASensDel(e, n) \\ \geq e - 2$.

\textbf{insertions}:
Consider string $S = \mathtt{a}^n$ of length $n$.
The maximal repeats of $\CDAWG(S)$ are
all of form $\mathtt{a}^h$ with $1 \leq h < n$
and each of them has exactly one out-going edge labeled by $\mathtt{a}$.
The total number of edges in $\CDAWG(S)$ is thus $n$
including the one from the source.
Consider string $S' = \mathtt{a}^n \mathtt{b}$ of length $n+1$.
The set of maximal repeats does not change from $S$,
but $\mathtt{b}$ is a right-extension of $\mathtt{a}^h$ for each $1 \leq h < n$.
Thus, $\CDAWG(S')$ has a total of $2n-2$ edges,
including the two out-going edges from the source.
Thus we have $e(S')/e(S) = \frac{2n-2}{n}$
and $e(S') - e(S) = n-2$.
This gives us $\liminf_{n \rightarrow \infty}\MSensIns(e,n) \geq 2$
$\ASensIns(e,n) \geq n-2$ and $\ASensIns(e, n) \geq e - 2$.
\end{proof}

\begin{figure}[tbh]
  \centerline{
    \includegraphics[scale=0.35]{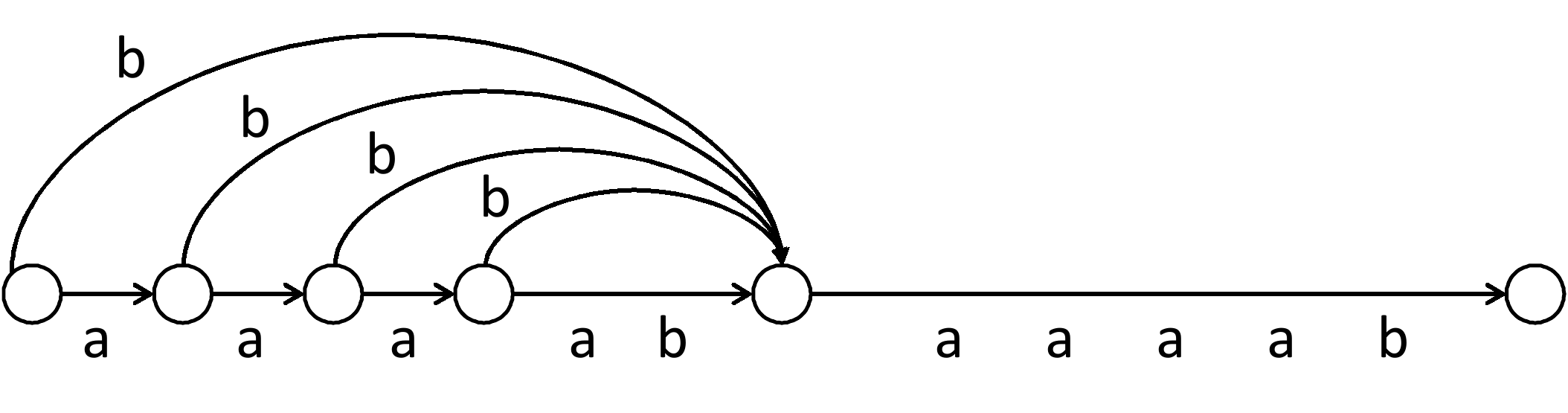}
    \hfill
    \includegraphics[scale=0.35]{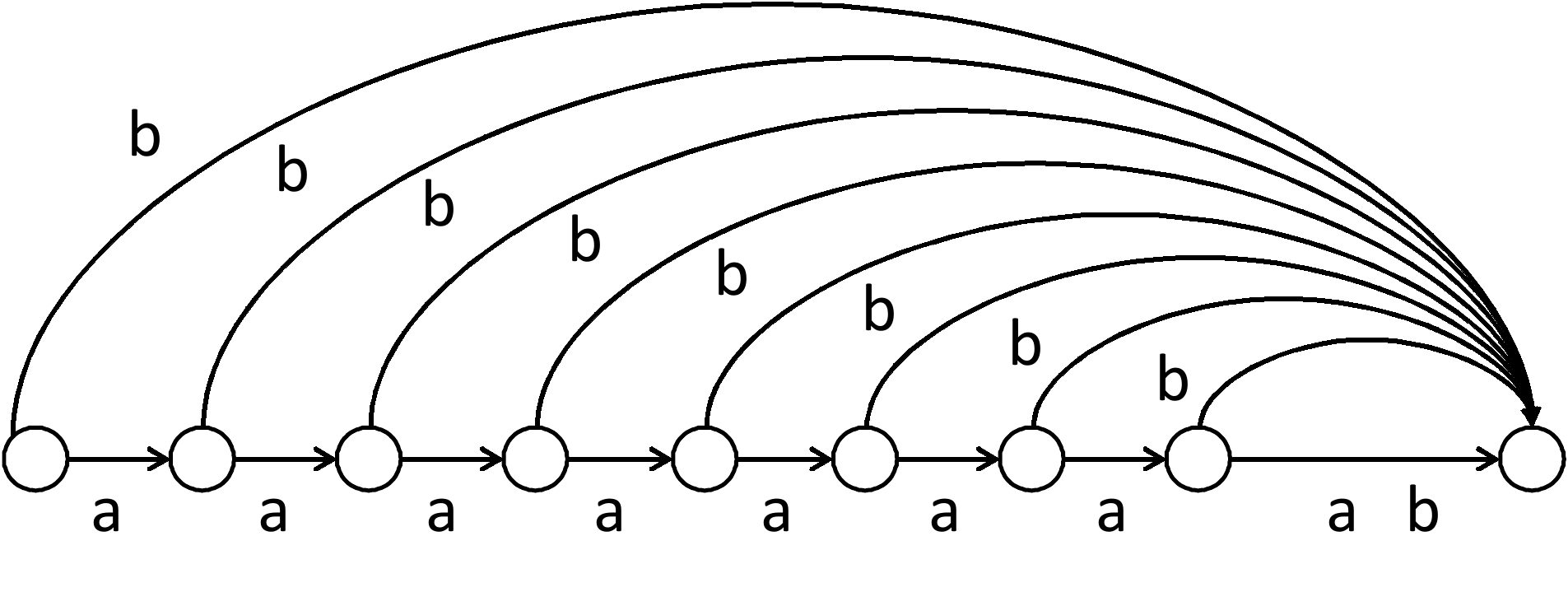}
  }
  \caption{The CDAWGs for strings $T = \mathtt{a}^4\mathtt{b}\mathtt{a}^4\mathtt{b}$ (left) and $T' = \mathtt{a}^8\mathtt{b}$ (right).}
  \label{fig:CDAWG}
\end{figure}

 \section{Concluding remarks and future work}
\label{sec:conclusions}

In the seminal paper by Varma and Yoshida~\cite{VarmaY21}
which first introduced the notion of sensitivity for (general) algorithms
and studied the sensitivity of graph algorithms, the authors wrote:
\begin{quote}
``\textit{Although we focus on graphs here, we note that our definition can also be extended to the study of combinatorial objects other than graphs such as strings and constraint satisfaction problems.}''
\end{quote}
Our study was inspired by the afore-quoted suggestion, and 
our sensitivity for string compressors and repetitiveness measures
enables one to evaluate the robustness and stability of compressors and
repetitiveness measures.

The major technical contributions of this paper are
the \emph{tight and constant upper and lower bounds} for the multiplicative
sensitivity of the LZ77 family, the smallest bidirectional scheme $b$,
and the substring complexity $\delta$.
We also presented tight and constant upper and lower bounds
for the multiplicative sensitivity of the recently proposed grammar compressor GCIS,
which is based on the idea of the Induced Sorting algorithm for suffix sorting.
We also reported non-trivial upper and/or lower bounds
for other string compressors, including RLBWT, LZ-End, LZ78, 
AVL-grammar, $\alpha$-balanced grammar, RePair, LongestMatch, Greedy, Bisection, and CDAWG.
Some of the upper bounds reported here follow from previous important work~\cite{KempaP18, KociumakaNP20, kempa2020resolution, KreftN13, KempaS22, CharikarLLPPSS05, Rytter03, Jez16}.

Apparent future work is to
complete Tables~\ref{tbl:multiplicative_sensitivity} and~\ref{tbl:additive_sensitivity}
by filling the missing pieces and closing the gaps between
the upper and lower bounds which are not tight there.

While we dealt with a number of string compressors and repetitiveness measures,
it has to be noted that our list is far from being comprehensive:
It is intriguing to analyze the sensitivity of
other important and useful compressors and repetitiveness measures
including
the size $\nu$ of the smallest NU-systems~\cite{NavarroU21},
the sizes of the other locally-consistent compressed indices such as
ESP-index~\cite{MaruyamaNKS13} and SE-index~\cite{NishimotoIIBT20}.

Our notion of the sensitivity for string compressors/repetitiveness measures
can naturally be extended to labeled tree compressors/repetitiveness measures.
It would be interesting to analyze the sensitivity for
the smallest tree attractor~\cite{Prezza21},
the run-length XBWT~\cite{Prezza21},
the tree LZ77 factorization~\cite{GawrychowskiJ16},
tree grammars~\cite{LohreyMM13,GanardiHLN18},
and top-tree compression of trees~\cite{BilleGGLW19}.
 
\section*{Acknowledgments}
This work was supported by JSPS KAKENHI Grant Numbers JP20J21147 (MF) and JP22H03551 (SI),
and by JST PRESTO Grant Number JPMJPR1922 (SI).
The authors thank Yuichi Yoshida for his helpful comments.
The authors thank anonymous referees for pointing out some errors in the earlier version of this work and for their suggestions to improve the paper.

\clearpage
\appendix

\section{Omitted proofs}

In this section, we present omitted proofs.

\subsection{Proof for Theorem~\ref{theo:lz77_lowerbounds_sqrtn} ($\Omega(\sqrt{n})$ additive sensitivity for $\zorig$)} \label{app:lz77_lowerbounds_sqrtn}

\begin{proof}
  Let $p = 2^h$ where $h \geq 1$.

  \textbf{substitutions:}  
  Consider the following string $T$ of length $n = \Theta(p^2)$:
  \[
    T = \mathtt{a^{2\mathit{p}-2}b \cdot a^{\mathit{p}}b\#_1  \cdot a^{\mathit{p}+1}b\#_2 \cdot a^{\mathit{p}+2}b\#_3 \cdots a^{2\mathit{p}-2} b \#_{\mathit{p}-1}},
  \]
  where $\#_j$ for every $1 \leq j \leq p-1$ is a distinct character.
The non self-referencing LZ77 factorization of $T$ is
  $$\LZorig(T) = \mathtt{a|a^2|a^4| \cdots |a^{2^{\mathit{h}-1}}|a^{\mathit{p}-1}b|a^{\mathit{p}}b\#_1|a^{\mathit{p}+1}b\#_2|a^{\mathit{p}+2}b\#_3|\cdots | a^{2\mathit{p}-2}b\#_{\mathit{p}-1}|}$$
  with $\zorig(T) = h+p$.
  Then, we consider the string 
  \[
    T' = \mathtt{a^{\mathit{p}-1}ca^{\mathit{p}-2} b \cdot a^{\mathit{p}}b\#_1  \cdot a^{\mathit{p}+1}b\#_2 \cdot a^{\mathit{p}+2}b\#_3
      \cdots a^{2\mathit{p}-2} b \#_{\mathit{p}-1}},
    \]
  which can be obtained from $T$ by substituting the $\mathit{p}$-th $\mathtt{a}$ with $\mathtt{c}$.
  Let us analyze the structure of the non self-referencing LZ77 factorization of $T'$.
  It is clear that $h$ factors in the interval $[1..p-1]$ are unchanged.
  Since $\mathtt{c}$ is a fresh character, it becomes a factor of length 1.
  Also, $\mathtt{a^{\mathit{p}-2} b}$ becomes a factor.
  The following each factor $\mathtt{a}^{p+k-2} \mathtt{b} \#_{k-1}$ with $2 \leq k \leq p$
  is divided into two factors $\mathtt{a}^{p+k-2}$ and $\mathtt{b}\#_{k-1}$,
  since there are no previous occurrences of $\mathtt{a}^{p+k-2}$ and $\#_{k-1}$.
  Thus, the non self-referencing LZ77 factorization of $T'$ is
  $$\LZorig(T') = \mathtt{a|a^2|a^4| \cdots |a^{2^{\mathit{h}-1}}|c|a^{\mathit{p}-2}b|a^{\mathit{p}}|b\#_1|a^{\mathit{p}+1}|b\#_2|a^{\mathit{p}+2}|b\#_3|\cdots | a^{2\mathit{p}-2}|b}\#_{p-1}|$$
  with $\zorig(T') = h+2p$, which leads to
  $\liminf_{n \rightarrow \infty}\MSensSub(\zorig,n) \geq \liminf_{p \to \infty} (h+2p)/(h+p) = 2$, $\ASensSub(\zorig,n) \geq (h+2p)-(h+p) = p = \Omega(\sqrt{n})$.

  \textbf{insertions:}
  As for the same string $T$, we consider the string 
  \[
    T' = \mathtt{a^{\mathit{p}-1}ca^{\mathit{p}-1} b \cdot a^{\mathit{p}}b\#_1  \cdot a^{\mathit{p}+1}b\#_2 \cdot a^{\mathit{p}+2}b\#_3
      \cdots a^{2\mathit{p}-2} b \#_{\mathit{p}-1}},
    \]
  which can be obtained from $T$ by inserting $\mathtt{c}$ between position $p-1$ and position $p$ in $T$.
  Then, by similar arguments to the case of substitutions, the non self-referencing LZ77 factorization of $T'$ is
  $$\LZorig(T') = \mathtt{a|a^2|a^4| \cdots |a^{2^{\mathit{h}-1}}|c|a^{\mathit{p}-1}b|a^{\mathit{p}}|b\#_1|a^{\mathit{p}+1}|b\#_2|a^{\mathit{p}+2}|b\#_3|\cdots | a^{2\mathit{p}-2}|b}\#_{p-1}|$$
  with $\zorig(T') = h+2p$, which leads to
  $\liminf_{n \rightarrow \infty}\MSensIns(\zorig,n) \geq \liminf_{p \to \infty} (h+2p)/(h+p) = 2$, $\ASensIns(\zorig,n) \geq p = \Omega(\sqrt{n})$.

  \textbf{deletions:}
  Consider the following string $T$ of length $n = \Theta(p^2)$:
  \[
    T = \mathtt{a^{\mathit{p}-1}cb \cdot acb\#_1  \cdot a^{2}cb\#_2 \cdot a^{3}cb\#_3
      \cdots a^{\mathit{p}-1} cb \#_{\mathit{p}-1}}.
  \]
The non self-referencing LZ77 factorization of $T$ is
  $$\LZorig(T)= \mathtt{a|a^2|a^4| \cdots |a^{2^{\mathit{h}-1}}|c|b|acb\#_1|a^{2}cb\#_2|a^{3}cb\#_3|\cdots | a^{\mathit{p}-1}cb\#_{\mathit{p}-1}|}$$
  with $\zorig(T) = h+p+1$.
  Then, we consider the string 
  \[
    T' =\mathtt{a^{\mathit{p}-1}b \cdot acb\#_1  \cdot a^{2}cb\#_2 \cdot a^{3}cb\#_3
      \cdots a^{\mathit{p}-1} cb \#_{\mathit{p}-1}},
    \]
  which can be obtained from $T$ by deleting the first $\mathtt{c}$ in $T$.
  Let us analyze the structure of the non self-referencing LZ77 factorization of $T'$.
  It is clear that $h$ factors in the interval $[1..p-1]$ are unchanged.
  The next factor is $\mathtt{b}$ of length 1.
  The following each factor $\mathtt{a}^{k} \mathtt{cb} \#_k$ with $1 \leq k \leq p-1$
  is divided into two factors $\mathtt{a}^{k}\mathtt{c}$ and $\mathtt{b}\#_k$,
  since there are no previous occurrences of $\mathtt{a}^{k}\mathtt{c}$ and $\mathtt{b}\#_k$.
  Thus, the non self-referencing LZ77 factorization of $T'$ is
  $$\LZorig(T')= \mathtt{a|a^2|a^4| \cdots |a^{2^{\mathit{h}-1}}|b|ac|b\#_1|a^{2}c|b\#_2|a^{3}c|b\#_3|\cdots | a^{\mathit{p}-1}c|b\#_{\mathit{p}-1}|}$$
  with $\zorig(T') = h+1+ 2(p-1) = h+2p-1$, which leads to
  $\liminf_{n \rightarrow \infty}\MSensDel(\zorig,n) \geq \liminf_{p \to \infty} (h+2p-1)/(h+p+1) = 2$, $\ASensDel(\zorig,n) \geq (h+2p-1)-(h+p+1) = p-2 = \Omega(\sqrt{n})$.
\end{proof}

\subsection{Proof for Theorem~\ref{theo:lz77sr_lowerbounds_sqrtn} ($\Omega(\sqrt{n})$ additive sensitivity of $\zsrorig$)} \label{app:lz77sr_lowerbounds_sqrtn}

\begin{proof}
  \textbf{substitutions:}
  Consider the following string $T$ of length $n = \Theta(p^2)$:
  \[
    T = \mathtt{a^{\mathit{p}-1}a \cdot a^{\mathit{p}}b \cdot a^{\mathit{p}+1}b\#_1 \cdot a^{\mathit{p}+2}b\#_2  \cdots a^{2\mathit{p}-1}b\#_{\mathit{p}-1}}
  \]
  which consists of $p+1$ components.
The self-referencing LZ77 factorization of $T$ is
  \[
  \LZorigsr(T) = \mathtt{a|a^{2\mathit{p}-1}b|a^{\mathit{p}+1}b\#_1|a^{\mathit{p}+2}b\#_2| \cdots |a^{2\mathit{p}-1}b\#_{\mathit{p}-1}|}
  \]
  with $\zsrorig(T) = p+1$. Notice that the second factor $\mathtt{a^{2\mathit{p}-1}1}$ is self-referencing.

  Consider the string $T'$
  \[
  T' = \mathtt{a^{\mathit{p}-1}c \cdot a^{\mathit{p}}b \cdot a^{\mathit{p}+1}b\#_1  \cdot a^{\mathit{p}+2}b\#_2  \cdots a^{2\mathit{p}-1}b\#_{\mathit{p}-1}}
  \]
  that can be obtained from $T$ by substituting the $p$-th $\mathtt{a}$ with $\mathtt{c}$.
  The self-referencing LZ77 factorization of $T'$ is
  \[
  \LZorigsr(T') = \mathtt{a|a^{\mathit{p}-2}c|a^{\mathit{p}}|b|a^{\mathit{p}+1}|b\#_1|a^{\mathit{p}+2}|b\#_2| \cdots |a^{2\mathit{p}-1}|b\#_{\mathit{p}-1}|}
  \]
  with $\zsrorig(T') = 2p+2$, which leads to
  $\MSensSub(\zsrorig,n) \geq (2p+2)/(p+1) = 2$,
  $\ASensSub(\zsrorig,n) \geq (2p+2)-(p+1) = p+1 = \zsrorig$, and
  $\ASensSub(\zsrorig,n) = \Omega(\sqrt{n})$.

  \textbf{insertions:}
  Consider the following string $T$ of length $n = \Theta(p^2)$:
  \[
    T = \mathtt{a^{\mathit{p}-1} \cdot a^{\mathit{p}}b \cdot a^{\mathit{p}+1}b\#_1  \cdot a^{\mathit{p}+2}b\#_2  \cdots a^{2\mathit{p}-1}b\#_{\mathit{p}-1}}.
  \]
The self-referencing LZ77 factorization of $T$ is
  \[
  \LZorigsr(T) = \mathtt{a|a^{2\mathit{p}-2}b|a^{\mathit{p}+1}b\#_1|a^{\mathit{p}+2}b\#_2| \cdots |a^{2\mathit{p}-1}b\#_{\mathit{p}-1}|}
  \]
  with $\zsrorig(T) = p+1$. Notice that the second factor $\mathtt{a^{2\mathit{p}-1}1}$ is self-referencing.

  Consider the string $T'$
  \[
  T' = \mathtt{a^{\mathit{p}-1}c \cdot a^{\mathit{p}}b \cdot a^{\mathit{p}+1}b\#_1  \cdot a^{\mathit{p}+2}b\#_2  \cdots a^{2\mathit{p}-1}b\#_{\mathit{p}-1}}
  \]
  that can be obtained from $T$ by inserting $\mathtt{c}$ between position $p-1$ and position $p$.
  The self-referencing LZ77 factorization of $T'$ is
  \[
  \LZorigsr(T') = \mathtt{a|a^{\mathit{p}-2}c|a^{\mathit{p}}|b|a^{\mathit{p}+1}|b\#_1|a^{\mathit{p}+2}|b\#_2| \cdots |a^{2\mathit{p}-1}|b\#_{\mathit{p}-1}|}
  \]
  with $\zsrorig(T') = 2p+2$, which leads to
  $\MSensIns(\zsrorig,n) \geq 2$,
  $\ASensIns(\zsrorig,n) \geq p+1 = \zsrorig$, and
  $\ASensIns(\zsrorig,n) = \Omega(\sqrt{n})$.

  \textbf{deletions:}
  Consider the following string $T$ of length $n = \Theta(p^2)$:
  \[
    T = \mathtt{a^{\mathit{p}}bc \cdot abc\#_1  \cdot a^{2}bc\#_2 \cdots a^{\mathit{p}}bc\#_{\mathit{p}}}.
  \]
The self-referencing LZ77 factorization of $T$ is
  \[
  \LZorigsr(T) = \mathtt{a|a^{\mathit{p}-1}b|c|abc\#_1|a^{2}bc\#_2| \cdots |a^{\mathit{p}}bc\#_{\mathit{p}}|}
  \]
  with $\zsrorig(T) = p+3$. Notice that the second factor $\mathtt{a^{\mathit{p}-2}1}$ is self-referencing.

  Consider the string $T'$
  \[
  T' = \mathtt{a^{\mathit{p}}b \cdot abc\#_1 \cdot a^{2}bc\#_2 \cdots a^{\mathit{p}}bc\#_{\mathit{p}}}
  \]
  that can be obtained from $T$ by deleting the first $\mathtt{c}$ of position $p+2$.
  Let us analyze the structure of the self-referencing LZ77 factorization of $T'$.
  The first two factors are unchanged.
  The third factor $\mathtt{c}$ of $\LZorigsr(T)$ is removed, and each of the remaining factors of form $\mathtt{a}^{\mathit{k}}\mathtt{bc}\#_{\mathit{k}}$ in $\LZorigsr(T)$ is divided into two factors as $\mathtt{a}^{\mathit{k}}\mathtt{bc}|\#_{\mathit{k}}|$.
  Thus the self-referencing LZ77 factorization of $T'$ is
  \[
  \LZorigsr(T') = \mathtt{a|a^{\mathit{p}-1}b|abc|\#_1|a^{2}bc|\#_2| \cdots |a^{\mathit{p}}bc|\#_{\mathit{p}}|}
  \]
  with $\zsrorig(T') = 2p+2$, which leads to
  $\liminf_{n \rightarrow \infty} \MSensDel(\zsrorig,n) \geq \liminf_{p \rightarrow \infty}(2p+2)/(p+3) = 2$,
  $\ASensDel(\zsrorig,n) \geq 2p+2-(p+3) = p-1 = \Omega(\sqrt{n})$.
  \end{proof}

  It is also possible to binarize the strings $T$ and $T'$ in the above proof
  for the cases of substitutions and insertions,
  while retaining the same lower bounds:
  \begin{corollary}
  For the self-referencing LZ77 factorization, there are binary strings of length $n$ that satisfy $\MSensSub(\zsrorig,n) \geq 2$, $\MSensIns(\zsrorig,n) \geq 2$, respectively.
  \end{corollary}

  \begin{proof}
  Let $p \geq 2$.

  \textbf{substitutions:}
  Consider the following string $T$ of length $n = \Theta(p^2)$:
  $$T = \mathtt{0^{\mathit{p}-1}\cdot0 \cdot 0^{2\mathit{p}}1 \cdot 0^{2\mathit{p}+1}101 \cdot 0^{2\mathit{p}+2}10^21 \cdots 0^{3\mathit{p}}10^{\mathit{p}}1}.$$
  The self-referencing LZ77 factorization of $T$ is: 
  $$\LZorigsr(T) = \mathtt{0|0^{3\mathit{p}-1}1|0^{2\mathit{p}+1}101|0^{2\mathit{p}+2}10^21| \cdots |0^{3\mathit{p}}10^{\mathit{p}}1}|$$
  with $p+2$ factors.
  Then, we consider the string
  $$T' = \mathtt{0^{\mathit{p}-1}\cdot1 \cdot 0^{2\mathit{p}}1 \cdot 0^{2\mathit{p}+1}101 \cdot 0^{2\mathit{p}+2}10^21 \cdots 0^{3\mathit{p}}10^{\mathit{p}}1}$$
  is obtained by substituting $p$-th $\mathtt{0}$ with $\mathtt{1}$. The self-referencing LZ77 factorization of $T'$ is: 
  $$\LZorigsr(T') = \mathtt{0|0^{\mathit{p}-2}1|0^{\mathit{p}}|0^\mathit{p}1|0^{2\mathit{p}+1}|101|0^{2\mathit{p}+2}|10^21| \cdots |0^{3\mathit{p}}|10^{\mathit{p}}1}|$$
  with $2p+4$ factors.
  Then we obtain $\MSensSub(\zsrorig,n) \geq (2p+4)/(p+2) = 2$.

  \textbf{insertions:}
  Consider the following string $T$ of length $n = \Theta(p^2)$:
  $$T = \mathtt{0^{\mathit{p}-1} \cdot 0^{2\mathit{p}}1 \cdot 0^{2\mathit{p}+1}101 \cdot 0^{2\mathit{p}+2}10^21 \cdots 0^{3\mathit{p}}10^{\mathit{p}}1}.$$
  The self-referencing LZ77 factorization of $T$ is: 
  $$\LZorigsr(T) = \mathtt{0|0^{3\mathit{p}-2}1|0^{2\mathit{p}+1}101|0^{2\mathit{p}+2}10^21| \cdots |0^{3\mathit{p}}10^{\mathit{p}}1}|$$
  with $p+2$ factors.
  Consider the string
  $$T' = \mathtt{0^{\mathit{p}-1}\cdot1 \cdot 0^{2\mathit{p}}1 \cdot 0^{2\mathit{p}+1}101 \cdot 0^{2\mathit{p}+2}10^21 \cdots 0^{3\mathit{p}}10^{\mathit{p}}1}$$
  is obtained by inserting $\mathtt{1}$ between $p-1$ and $p$. The self-referencing LZ77 factorization of $T'$ is: 
  $$\LZorigsr(T') = \mathtt{0|0^{\mathit{p}-2}1|0^{\mathit{p}}|0^\mathit{p}1|0^{2\mathit{p}+1}|101|0^{2\mathit{p}+2}|10^21| \cdots |0^{3\mathit{p}}|10^{\mathit{p}}1}|$$
  with $2p+4$ factors.
  Then we get $\MSensIns(\zsrorig,n) \geq (2p+4)/(p+2) = 2$.

\end{proof}

\clearpage
\bibliographystyle{abbrv}
\bibliography{ref}

\end{document}